\documentclass[11pt, a4paper]{article}
\pdfoutput=1

\usepackage[a4paper, margin=1in]{geometry}
\usepackage{libertine}
\usepackage{amsmath,amssymb}
\usepackage{amsthm}
\usepackage{mathtools}
\usepackage{thmtools,thm-restate}
\usepackage{microtype}
\usepackage{xparse}
\usepackage[table,dvipsnames]{xcolor}
\usepackage{setspace}
\usepackage[hypertexnames=false]{hyperref}
\usepackage[capitalize,nameinlink,noabbrev]{cleveref}
\usepackage{crossreftools}
\usepackage{xspace}
\usepackage{graphicx}
\usepackage{caption}
\usepackage{enumitem}
\usepackage{makecell}
\usepackage{extarrows}
\usepackage{float}
\usepackage[natbib,
            backend=biber,
            maxbibnames=99,
            maxcitenames=99,
            style=alphabetic,
            maxalphanames=5,
            minalphanames=5
]{biblatex}

\DeclareDelimFormat{multicitedelim}{\addcomma\space}
\AtEveryBibitem{
  \clearfield{url}
  \clearfield{eprint}
}
\DeclareFieldFormat{doi}{
  DOI:\space
  \addcolon\space\href{https://doi.org/#1}{\nolinkurl{#1}}%
}
\appto{\bibsetup}{\emergencystretch=3em}

\makeatletter
\@ifundefined{newcounteralias}{}{%
  \renewcommand\thmt@autorefsetup{%
    \@xa\def\csname\thmt@envname autorefname\@xa\endcsname\@xa{\thmt@thmname}%
  }%
}
\makeatother

\definecolor{notionblue}{HTML}{1F4E79}
\definecolor{ao(english)}{rgb}{0.0, 0.5, 0.0}

\graphicspath{ {pictures} }

\usepackage[xcolor,hyperref,notion,electronic,makeidx]{knowledge}
\knowledgestyle{notion}{color=notionblue}
\knowledgestyle{intro notion}{color=notionblue,italic}
\knowledge{notion}
    | Ordered
    | Ordered PCSP
    | Ordered PCSPs
\knowledge{notion}
    | Multiple choice hardness condition
    | multiple choice hardness condition
    | Multiple choice condition
    | multiple choice condition
\knowledge{notion}
    | Linear Threshold@PCSP
    | Linear Threshold PCSP
    | Linear Threshold PCSPs
\knowledge{notion}
    | Linear threshold function
    | Linear threshold functions
    | Linear threshold@func
    | Linear thresholds@func
    | linear threshold function
    | linear threshold@func
    | linear threshold functions
    | linear thresholds@func
\knowledge{notion}
    | Unate@PCSP
    | Unate PCSP
    | Unate PCSPs
\knowledge{notion}
    | Unate@func
    | unate@func
    | unate
\knowledge{notion}
    | Polynomial Threshold@PCSP
    | Polynomial Threshold PCSP
    | Polynomial Threshold PCSPs
\knowledge{notion}
    | Boolean@PCSP
    | Boolean PCSP
    | Boolean PCSPs
\knowledge{notion}
    | Boolean function
    | Boolean functions

\knowledge{notion}
    | universe
    | Universe
\knowledge{notion}
    | relation
    | relations
\knowledge{notion}
    | similar
\knowledge{notion}
    | Constraint Satisfaction Problem
    | CSP
\knowledge{notion}
    | structure
    | structures
    | relational structure
    | relational structures
    | relational templates
\knowledge{notion}
    | homomorphism
    | homomorphisms
    | homomorphically
\knowledge{notion}
    | CSP template
    | template@CSP
    | templates@CSP
\knowledge{notion}
    | Promise Constraint Satisfaction Problem
    | PCSP
\knowledge{notion}
    | PCSP template
    | template@PCSP
    | templates@PCSP
\knowledge{notion}
    | polymorphism
    | polymorphisms
    | compatible with
    | compatibility with
\knowledge{notion}
    | minion
    | minions
    | (function) minion
\knowledge{notion}
    | negated
    | negation
\knowledge{notion}
    | minor
    | minors
    | Minor
    | Minors
\knowledge{notion}
    | minor map
    | Minor map
    | minor maps
    | Minor maps
\knowledge{notion}
    | 2-to-1
    | 2-to-1 map
    | 2-to-1 minor map
    | 2-to-1 minor
    | 2-to-1 minors
    | 2-to-1 minor maps
\knowledge{notion}
    | projection
    | projections
    | dictator
\knowledge{notion}
    | idempotent
    | idempotency
\knowledge{notion}
    | essential
    | non-essential
    | essentiality
\knowledge{notion}
    | minion generated by
\knowledge{notion}
    | folded
    | foldedness
\knowledge{notion}
    | increasing@func
    | increasingness@func
\knowledge{notion}
    | increasing
    | decreasing
\knowledge{notion}
    | influence
    | influences
    | Influence
    | Influences
\knowledge{notion}
    | Total influence
    | total influence
    | average sensitivity
\knowledge{notion}
    | Margulis-Russo formula
\knowledge{notion}
    | Sharp threshold condition
    | sharp threshold condition
\knowledge{notion}
    | Random 2-to-1 hardness condition
    | random 2-to-1 hardness condition
    | random 2-to-1 condition
    | Random 2-to-1 condition
\knowledge{notion}
    | satisfy@MC
    | satisfied@MC
    | satisfies@MC
    | satisfiable@MC
    | assignment@MC
\knowledge{notion}
    | trivial@MC
    | triviality@MC

\knowledge{notion}
    | Biased distribution
    | biased distribution
    | biased
    | Biased
    | biased distributions
    | Biased distributions
\knowledge{notion}
    | Fourier characters
    | Fourier character
\knowledge{notion}
    | Fourier decomposition
    | Fourier expansion
\knowledge{notion}
    | Fourier coefficients
\knowledge{notion}
    | Parseval-Plancherel identity
    | Parseval-Plancherel
    | Parseval identity
    | Plancherel identity
    | Parseval
    | Plancherel

\knowledge{notion}
    | junta
\knowledge{notion}
    | Friedgut's Theorem

\knowledge{notion}
    | Shapley distribution
    | Shapley values
    | Shapley value
    | Shapley
\knowledge{notion}
    | pull-back distribution
    | Pull-back distribution
    | pull-back
    | Pull-back
\knowledge{notion}
    | symmetry
    | symmetric
    | Symmetry
    | Symmetric
\knowledge{notion}
    | flat
    | flatness
    | Flat
    | Flatness
\knowledge{notion}
    | smoothness
    | smooth
    | Smoothness
    | Smooth
\knowledge{notion}
    | consistency
    | consistent
    | Consistency
    | Consistent
\knowledge{notion}
    | pull-back compatibility
    | pull-back compatible
    | Pull-back compatibility
    | Pull-back compatible
\knowledge{notion}
    | reasonable
    | reasonability
    | Reasonable
    | Reasonability
    | reasonable distributions
    | Reasonable distributions
\knowledge{notion}
    | Cube Influence Preservation
\knowledge{notion}
    | Shapley Influence Preservation

\knowledge{notion}
    | Polynomial threshold function
    | Polynomial threshold functions
    | polynomial threshold function
    | polynomial threshold functions
    | polynomial threshold
    | Polynomial Threshold
    | polynomial thresholds
    | Polynomial thresholds
\knowledge{notion}
    | low-degree
    | Low-degree concentration
    | low-degree concentration
    | low-degree concentrated

\knowledge{notion}
    | Noise operator
    | noise
    | noise operator
    | Noise
\knowledge{notion}
    | noisy distribution
    | Noisy distribution
\knowledge{notion}
    | Noise sensitivity
    | noise sensitivity

\knowledge{notion}
    | Reasonable parameters
    | reasonable parameters

\knowledge{notion}
    | Sign-representation
    | sign-representation
\knowledge{notion}
    | Normalized@PTF
    | normalized@PTF
    | normalization@PTF
    | Normalization@PTF
\knowledge{notion}
    | Regularity@PTF
    | regularity@PTF
    | regular@PTF
    | Regular@PTF
\knowledge{notion}
    | Invariance Principle
\knowledge{notion}
    | Carbery-Wright anti-concentration
    | Carbery-Wright anti-concentration bound
    | Carbery-Wright bound
\knowledge{notion}
    | Hypercontractivity
    | hypercontractivity
    | (2,4)-hypercontractivity
\knowledge{notion}
    | Two-function hypercontractivity
    | two-function hypercontractivity
\knowledge{notion}
    | Paley-Zygmunt inequality
\knowledge{notion}
    | restriction@PTF
    | Restriction@PTF
    | Restrictions@PTF
    | restrictions@PTF
\knowledge{notion}
    | Critical index
    | critical index
\knowledge{notion}
    | Determined
    | determined

\knowledge{notion}
    | Label Cover
    | label cover
\knowledge{notion}
    | Gap Label Cover
    | gap label cover
\knowledge{notion}
    | Richness
    | richness
\knowledge{notion}
    | Rich 2-to-1 Conjecture
    | rich 2-to-1 conjecture
\knowledge{notion}
    | Rich 2-to-1 Gap Label Cover
\knowledge{notion}
    | Rich 2-to-1 Label Cover
\knowledge{notion}
    | Minor condition
    | minor condition
\knowledge{notion}
    | Promise Minor Condition
\knowledge{notion}
    | trivial
    | Trivial
\knowledge{notion}
    | satisfied
    | Satisfied
\knowledge{notion}
    | symbol interpretation
    | interpretation
    | Symbol interpretation
    | Interpretation
\knowledge{notion}
    | constraints
    | Constraints
\knowledge{notion}
    | labeling
    | Labeling
    | labelings
    | Labelings
\knowledge{notion}
    | identities
    | Identities
\knowledge{notion}
    | completeness
    | Completeness
\knowledge{notion}
    | soundness
    | Soundness

\hypersetup{
  colorlinks=true,
  linkcolor=ao(english),
  citecolor=ao(english),
  urlcolor=ao(english)
}

\DeclareUnicodeCharacter{0301}{\'a}

\newcommand{\PTF}{\ensuremath{\mathsf{PTF}}}
\renewcommand{\P}{\ensuremath{\mathsf{P}}}
\newcommand{\NP}{\ensuremath{\mathsf{NP}}}
\newcommand{\CSP}{\ensuremath{\mathsf{CSP}}}
\newcommand{\PCSP}{\ensuremath{\mathsf{PCSP}}}

\DeclarePairedDelimiter\floor{\lfloor}{\rfloor}
\DeclarePairedDelimiter\ceil{\lceil}{\rceil}
\DeclarePairedDelimiter\norm{\lVert}{\rVert}
\DeclarePairedDelimiter\ham{\lvert}{\rvert}

\DeclareMathOperator{\ar}{ar}
\DeclareMathOperator{\Pol}{Pol}
\DeclareMathOperator{\Sel}{\mathsf{Sel}}

\newcommand{\minion}[1]{\mathcal{#1}}

\knowledgenewrobustcmd{\AT}{\cmdkl{\mathrm{AT}}}
\knowledgenewrobustcmd{\MAX}{\cmdkl{\mathrm{MAX}}}
\knowledgenewrobustcmd{\MIN}{\cmdkl{\mathrm{MIN}}}
\knowledgenewrobustcmd{\THR}[1]{\cmdkl{\mathrm{THR}_{#1}}}
\knowledgenewrobustcmd{\XOR}{\cmdkl{\mathrm{XOR}}}

\newcommand{\tuple}[1]{\overline{#1}}
\newcommand{\relstr}[1]{{\mathbb{#1}}}

\newcommand{\A}{\relstr A}
\newcommand{\B}{\relstr B}

\newcommand{\Bool}{\{0,1\}}

\newcommand{\richconj}{Rich 2-to-1 Conjecture}

\newrobustcmd\gen[2]{\withkl{\kl[\gen{#1}]}{\cmdkl{\mathrm{#1}^{(#2)}}}}
\knowledge{\gen{max}}{notion}
\knowledge{\gen{min}}{notion}
\knowledge{\gen{at}}{notion}
\knowledge{\gen{xor}}{notion}
\knowledge{\gen{fan}}{notion}
\knowledge{\gen{maj}}{notion}
\knowledge{\gen{proj}}{notion}
\knowledge{\gen{an}}{notion}

\knowledgenewrobustcmd{\genthr}[2]{\cmdkl{\mathrm{thr}_{#1}^{(#2)}}}

\knowledgenewrobustcmd{\gentribes}[2]{\cmdkl{\mathrm{tribes}^{(#1,#2)}}}

\newcommand{\GapRich}{\ensuremath{\mathsf{Gap\text{-}Rich\text{-}2\text{-}to\text{-}1}}}
\newcommand{\Gap}{\ensuremath{\mathsf{Gap\text{-}Label\text{-}Cover}}}

\newcounter{restatecounter}
\makeatletter
\NewDocumentEnvironment{restate}{ m m o }{%
  \begingroup
    \def\restate@env{#1}%
    \setcounter{restatecounter}{\value{\restate@env}}%
    \expandafter\renewcommand\csname the#1\endcsname{\ref{#2}}%
    \IfValueTF{#3}
      { \begin{\restate@env}[#3] }
      { \begin{\restate@env}     }%
}{%
    \end{\restate@env}%
    \setcounter{\restate@env}{\value{restatecounter}}%
  \endgroup
}
\makeatother

\definecolor{Empty}{gray}{0.85}
\definecolor{Blue}{rgb}{0.2,0.6,0.8}

\newcolumntype{E}{>{\columncolor{Empty}}c}
\newcolumntype{B}{>{\columncolor{Blue!80}}c}

\newcommand{\relation}[2]{
\begin{tabular}{#1}
    #2
\end{tabular}
}

\let\emptyset\varnothing

\newcommand{\minor}[2]{{#1}^{#2}}

\DeclareMathOperator{\domOP}{dom}
\knowledgenewrobustcmd{\domainup}[1]{\cmdkl{\domOP^{\uparrow}(#1)}}
\knowledgenewrobustcmd{\domaindown}[1]{\cmdkl{\domOP^{\downarrow}(#1)}}

\newcommand{\ip}[2]{\left\langle #1, #2 \right\rangle}

\newcommand{\sgn}{\mathsf{sgn}}

\NewDocumentCommand{\cube}{>{\SplitArgument{1}{,}}m}{%
  \cubeaux#1%
}
\NewDocumentCommand{\cubeaux}{mm}{%
  \IfNoValueTF{#2}{\ensuremath{\mu_{#1}}}{\ensuremath{\mu_{#1}^{#2}}}%
}

\NewDocumentCommand{\chardist}{>{\SplitArgument{1}{,}}m}{%
  \chardistaux#1%
}
\NewDocumentCommand{\chardistaux}{mm}{%
  \IfNoValueTF{#2}{\ensuremath{X_{#1}}}{\ensuremath{X_{#1}^{#2}}}%
}

\newcommand{\shap}[1][]{\Phi_{#1}}

\NewDocumentCommand{\pullbackaux}{mm}{%
  \IfNoValueTF{#2}{\ensuremath{\nu_{#1}}}{\ensuremath{\nu_{#1}^{(#2)}}}%
}

\newcommand{\gpullback}[1]{\ensuremath{\mathbf{PB}^{#1}}}

\newcommand{\rest}[3]{\ensuremath{#1^{#2 \to #3}}}

\ExplSyntaxOn

\DeclareMathOperator*{\opE}{\mathbb E}
\NewDocumentCommand{\E}{o m}
 {
  \ensuremath{
    \IfNoValueTF{#1}
      {\opE\left[#2\right]}
      {{\opE\sb{#1}}{\left[#2\right]}}
  }
 }

\newcommand{\opNS}{\mathbf{NS}}
\NewDocumentCommand{\NS}{o >{\SplitArgument{1}{,}}m}
 {
   \__NS:nnn { #1 } #2
 }
\cs_new:Nn \__NS:nnn
 {
   \ensuremath{
     \IfNoValueTF{#1}
        {\opNS\sb{#3}\left[#2\right]}
        {\opNS\sp{#1}\sb{#3}\left[#2\right]}
   }
 }

\newcommand{\opTN}{\mathbf{T}}
\NewDocumentCommand{\TN}{o >{\SplitArgument{1}{,}}m}
 {
   \__TN:nnn { #1 } #2
 }
\cs_new:Nn \__TN:nnn
 {
   \ensuremath{
     \IfNoValueTF{#1}
        {\opTN\sb{#3}\left[#2\right]}
        {\opTN\sp{#1}\sb{#3}\left[#2\right]}
   }
 }

\newcommand{\opN}{\mathbf{N}}
\NewDocumentCommand{\N}{o >{\SplitArgument{1}{,}}m}
 {
   \__N:nnn { #1 } #2
 }
\cs_new:Nn \__N:nnn
 {
   \ensuremath{
     \IfNoValueTF{#1}
        {\opN\sb{#3}(#2)}
        {\opN\sp{#1}\sb{#3}(#2)}
   }
 }

\newcommand{\opInf}{\mathbf{Inf}}
\NewDocumentCommand{\Inf}{o >{\SplitArgument{1}{,}}m}
 {
   \__Inf:nnn { #1 } #2
 }
\cs_new:Nn \__Inf:nnn
 {
   \ensuremath{
     \IfNoValueTF{#1}
        {{\opInf\sb{#3}}{\left[#2\right]}}
        {{\opInf\sp{#1}\sb{#3}}{\left[#2\right]}}
   }
 }

\newcommand{\opI}{\mathbf{I}}
\NewDocumentCommand{\I}{o >{\SplitArgument{1}{,}}m}
 {
   \__I:nnn { #1 } #2
 }
 \cs_new:Nn \__I:nnn
 {
   \ensuremath{
    \opI\sp{\IfNoValueTF{#1}{}{#1}}\sb{\IfNoValueTF{#3}{}{#3}}{\left[#2\right]}
   }
 }

\ExplSyntaxOff

\newcommand{\curve}[2]{C_{#1}[#2]}

\declaretheorem{theorem}%
[style=plain,numberwithin=section]

\declaretheorem{lemma,corollary,proposition}%
[style=plain,sharenumber=theorem]

\declaretheorem{claim}[name=Claim, refname={Claim,Claims}, Refname={Claim,Claims}, style=plain, sharenumber=theorem]
\declaretheorem{observation,fact}%
[style=plain,sharenumber=theorem]

\declaretheorem{definition}%
[style=definition,numbered=no]

\declaretheorem{example}%
[style=definition,sharenumber=theorem]

[style=definition,numbered=no]

\declaretheorem{conjecture,question}%
[style=plain]
\declaretheorem{remark}%
[style=remark,numbered=no]

\begin{document}

\title{Boolean PCSPs through the lens of Fourier Analysis\thanks{This research was funded in whole or in part by the National Science Centre, Poland under the Weave-UNISONO call in the Weave programme 2021/03/Y/ST6/00171. For the purpose of Open Access, the author has applied a CC-BY public copyright licence to any Author Accepted Manuscript (AAM) version arising from this submission.}}

\author{
Demian Banakh \thanks{Faculty of Mathematics and Computer Science, and Doctoral School of Exact and Natural Sciences; Jagiellonian University, Krak{\'o}w, Poland.} \\ {\footnotesize demian.banakh@doctoral.uj.edu.pl}
\and
Katzper Michno \thanks{Faculty of Mathematics, Informatics and Mechanics, and Doctoral School of Exact and Natural Sciences; University of Warsaw, Poland.} \\ {\footnotesize katzper.michno@gmail.com}
}

\date{\today}

\maketitle

\singlespacing

\begin{abstract}
We develop an analytical framework for Boolean promise constraint satisfaction problems (PCSPs) that studies polymorphisms through the notion of influence from Fourier analysis of Boolean functions. Extending the work of Brakensiek, Guruswami, and Sandeep [ICALP'21] on Ordered PCSPs, we identify two general phenomena in Boolean minions indicative of hardness or tractability: (1) preservation of coordinate influence under random 2-to-1 minors and (2) the presence of sharp thresholds. We demonstrate that these phenomena occur in broader settings than previously established, yielding new hardness/tractability results for minions consisting of unate or polynomial threshold functions.
\end{abstract}

\tableofcontents

\section{Introduction}\label{sec:intro}
A key question in theoretical computer science is why some computational problems can be solved in polynomial time while others are hard.
The computational class \NP~consists of a vast variety of computational problems, and currently there is no unified framework that could, for every problem in \NP, explain the underlying mathematical structure that governs its tractability or hardness.
One subclass of \NP, for which such a framework does exist is the class of (non-uniform) constraint satisfaction problems (\CSP s).
In such a problem, the input consists of a set of variables along with constraints on their values, and the objective is to find an evaluation of the variables satisfying all constraints.
The allowed constraints are restricted to a fixed set of relations, known as the \CSP\ \reintro(CSP){template}, defined over a fixed finite domain.%

\CSP s capture many classic computational problems, such as 3-SAT and graph coloring.
In 1978, Schaefer~\cite{schaefer1978} showed that, assuming $\P \neq \NP$, every \kl(CSP){template} over a 2-element (Boolean) domain defines a \CSP\ which is either \NP-hard or solvable in polynomial time.
The next milestone was established by Hell and Ne\v{s}et\v{r}il~\cite{hellnesetril}, who showed that an analogous dichotomy is true for undirected graph \CSP s, that is, whose constraints are limited to a single binary symmetric relation.
Several years later, Feder and Vardi~\cite{federvardi} conjectured that the same dichotomy holds for the class of all finite \CSP\ \kl(CSP){templates}.
Progress on this conjecture was limited until the development of the \emph{algebraic approach to CSPs}~\cite{jeavons1997,jeavons1998,BJK05}, which relates the computational complexity of a \CSP\ to algebraic closure properties of the \kl(CSP){template} through multivariate generalizations of endomorphisms known as \kl{polymorphisms}.
This inspired an active research programme, culminating in 2017 with Bulatov~\cite{bulatov2017dichotomy} and Zhuk~\cite{zhuk2017proof} independently proving the \CSP\ Dichotomy Conjecture.

The \CSP\ Dichotomy Theorem tells us exactly what underlying mathematical structure of \CSP s governs their computational complexity.
This raises the question of whether the tools developed for \CSP s can be used to understand more general subclasses of \NP.
One plausible candidate for such a subclass is the class of \emph{Promise} \CSP s (\PCSP s), originally introduced in~\cite{austrin20172+sat}, which has enjoyed significant attention in recent years. These problems bridge the gap between decision and approximation problems by acting as a ``qualitative'' approximation variant of \CSP s.

A \PCSP\ \reintro(PCSP){template} $(\A, \B)$ consists of two sets of relations --- the \emph{strong} relations $\A$ and the \emph{weak} relations $\B$.
Given an instance satisfiable in the strong relations, the task is to produce a solution that satisfies all constraints in the weak relations.
A classic example of a \PCSP\ is the \emph{Approximate Graph Coloring} problem~\cite{agc1976}. This problem asks for a proper $\ell$-coloring of a graph that is promised to admit a proper $k$-coloring, where $k \leq \ell$ are fixed parameters.
Clearly, the problem is tractable if $k \leq 2$, and it was long conjectured to be \NP-hard in all other cases.
This was previously established conditional on the $d$-to-$1$ Conjecture (\cite[Theorem 1.1]{dinur2006approxcol}, \cite[Theorem 1]{GS2020agc}), and the recent preprint of Fei, Minzer, and Wang~\cite[Theorem 1.6]{dto1proof} claims to prove this conjecture, thereby resolving this long-standing open problem.

Another important example is $\PCSP(\mathsf{1in3, NAE})$: given a 3-uniform hypergraph that admits a red/blue coloring with exactly one blue vertex in every hyperedge, find a red/blue coloring with no monochromatic hyperedge.
Although separately $\mathsf{1in3}$-$\mathsf{SAT}$ and $\mathsf{NAE}$-$\mathsf{SAT}$ are both \NP-hard problems, the aforementioned promise problem is solvable in polynomial time (\cite[Theorem 1.4(c)]{brakensiek2017promise}, \cite[Theorem 1]{ficak2019symmetric}).
Notably, it is also an example of a \kl(PCSP){Boolean} \PCSP.

\AP
A crucial advance in the study of \PCSP s was the realization that the algebraic approach is relevant for \PCSP s as well: the complexity of $\PCSP(\A, \B)$ is captured by the set $\Pol(\A, \B)$ of its \kl{polymorphisms}~\cite{austrin20172+sat,brakensiek2017promise,Bible} which are multivariate \emph{homomorphisms} from the strong \emph{relational structure} $\A$ to the weak \emph{structure} $\B$.
In plain words, a function $f \colon A^n \to B$ of arity $\ar(f) = n$ (throughout, we write $[n] = \{1, \dots, n\}$) is a \reintro{polymorphism} of $(\A, \B)$ if for any relation in $\A$ and any matrix with $n$ columns --- elements of the relation --- the row-wise application of $f$ produces an element of the corresponding relation in $\B$.

The algebraic properties of $\Pol(\A, \B)$ explain most of the existing results on \PCSP s.
In broad terms, infinitely many \kl{polymorphisms} with no ``distinguished'' coordinates allow for efficient algorithms. Conversely, hardness is usually established when every polymorphism possesses a small set of ``important'' coordinates.
For example, the authors of \cite[Theorem 3.2]{blp+aip} devised a polynomial-time algorithm that solves\footnote{To be precise, this statement refers to the decision version of \PCSP; to solve the search version, this algorithm needs to ``know'' a symmetric polymorphism of arbitrarily large arity.} any \PCSP\ that admits infinitely many \emph{symmetric} \kl{polymorphisms}, that is, functions which are invariant under permutations of their arguments. On the other front, the following \emph{multiple choice condition} is known to imply \NP-hardness of a \PCSP, and has been widely utilized to obtain numerous \PCSP\ hardness results~\cite{austrin20172+sat, brakensiek2017promise, ficak2019symmetric, KO, afotw2025topology, michno2025booleanpcspspolynomialthreshold}.
\begin{definition}[\intro{Multiple choice hardness condition} \cite{brakensiek2017promise,Bible,barto2022babypcp}] \AP
    We say that $\Pol(\A, \B)$ satisfies the \emph{multiple choice condition} if there exists $\mathcal C > 0$ and a mapping $\Sel$ that assigns to each $f \in \Pol(\A, \B)$ a subset of its coordinates such that
    \begin{enumerate}
        \item $\lvert\Sel(f)\rvert \le \mathcal C$, and
        \item\label{item2:multiple-choice} $\pi(\Sel(f)) \cap \Sel(\minor{f}{\pi}) \neq \emptyset$ for any \kl{minor}%
        \footnote{We define minors formally later; in plain words, it is a function obtained from $f$ by identifying, permuting, and introducing dummy variables.}%
        ~$\minor{f}{\pi}$ of $f$.
    \end{enumerate}
\end{definition}

A complete complexity classification of all \PCSP s remains open. The history of \CSP s suggests that studying restricted subclasses of \PCSP s may help to better understand the underlying mathematical structure of hard and tractable problems.
However, even analogues of (1) the Hell-Ne\v{s}et\v{r}il's dichotomy for graph \CSP s and (2) the Schaefer's Boolean \CSP\ dichotomy remain open in the broader promise setting.
It is worth adding that the former has been resolved in~\cite{multislices,pmmsnp} conditional on the recently proposed \richconj~\cite[Conjecture 7]{onrich2to1}.

In this work, we focus on the latter: \emph{Boolean} \PCSP s.
A \PCSP\ \kl(PCSP){template} $(\A, \B)$ is \AP\intro(PCSP){Boolean} when both
strong and weak relations are over the two-element domain $\Bool$; its
\kl{polymorphisms} are then \AP\intro{Boolean functions}, that is, functions of the form $f \colon \Bool^n \to \Bool$.
Despite a flurry of partial results~\cite{austrin20172+sat, brandts2021pcsp1in3nae, brakensiek2017promise, ficak2019symmetric, brakensiek2021conditional, injhardness, austrin2025usefulness}, the complexity classification of all Boolean \PCSP s seems to currently be out of reach.
Among the most general results obtained so far is a complexity dichotomy~\cite[Theorem 1.1]{brakensiek2021conditional} for the so-called \kl{Ordered} Boolean \PCSP s (that is, \PCSP s that allow to impose inequalities between variables of the input instance), conditional on the \richconj.
The unconditional results are limited and challenging.
Therefore, employing perfect-completeness extensions of the Khot's Unique Games Conjecture~\cite{khot2002dto1} --- such as the $d$-to-$1$ Conjecture\footnote{The recent preprint of Fei, Minzer, and Wang~\cite[Theorem 1.6]{dto1proof} claims to resolve the $d$-to-$1$ Conjecture for $d=4$. Following it, OpenAI published preprints claiming proofs of the Unique Games Conjecture~\cite[Theorem 1.1]{unique-games-theorem} and of the $2$-to-$1$ Conjecture with perfect completeness~\cite[Theorem 1.1]{2to1-perfect-completeness}.}~\cite{khot2002dto1} or the \richconj~\cite[Conjecture 7]{onrich2to1} --- as the source of hardness offers a way to bypass the limitations of the available hardness tools and gain a deeper understanding of the underlying mathematical structure of hard Boolean \PCSP s.
In this sense, the \richconj\ may serve a role similar to that of the Unique Games Conjecture in many inapproximability results.

The aforementioned work of \cite{brakensiek2021conditional} on \kl{Ordered} \PCSP s also pioneered\footnote{To be precise, the authors of \cite{brakensiek2021conditional} took inspiration from a seminar talk by Libor Barto~\cite{barto2018cyclic}.} the use of Fourier analysis to study the properties of \kl[Boolean functions]{Boolean} \kl{polymorphisms}, although in an indirect way.
We believe that using the full power of Fourier analysis will push the understanding of \kl(PCSP){Boolean} \PCSP s closer towards resolving their complexity classification, and also has the potential to unify most known hardness and tractability results on Boolean \PCSP s.
In this paper, the first steps in this direction are presented. We build the foundation of an analytical framework for studying Boolean \PCSP s; as an application of the new theory, combinatorial arguments of \cite{brakensiek2021conditional} are substituted with our analytical approach, extending their observations beyond the landscape of Ordered \PCSP s.

\subsection*{Coordinate influence}
We have already alluded to the observation that in most complexity classifications for PCSPs, tractability or hardness goes hand in hand with the lack or presence of \emph{important} coordinates in \kl{polymorphisms}.
In our framework, this vague ``importance'' is exhibited through the notion of \kl{influence} introduced by Ben-Or and Linial~\cite{collective_coin_flipping}, which is one of the most fundamental concepts of the Fourier analysis of \kl{Boolean functions}, a subject treated comprehensively in the monograph of O'Donnell~\cite{odonnell}.\footnote{Boolean functions in \cite{odonnell} are of the form $\{-1,1\}^n \to \{-1,1\}$, in contrast to our functions of the form $\{0,1\}^n \to \{0,1\}$. The difference is mostly cosmetic; our presentation is similar to e.g. \cite{hypergraphremovallemmas}.}

\AP
In the simplest setting, for a function $f \colon \Bool^n \to \Bool$, the \reintro{influence} of a coordinate $i \in [n]$ --- denoted by $\Inf{f,i}$ --- is the fraction of inputs $\tuple x \in \Bool^n$ such that flipping the $i$-th entry $x_i$ flips the output $f(\bar x)$.
For example, in the function $\tuple x \mapsto \max_i x_i$, each coordinate has influence $2^{1-n}$, which is quite small.
On the other hand, in the \emph{dictatorship} function $\tuple x \mapsto x_i$, the dictating coordinate $i$ has influence $1$, while the other coordinates are irrelevant --- their influence is $0$.

One may surmise that admitting \kl{polymorphisms} with diminishing maximum influences characterizes tractability of \kl(PCSP){Boolean} \PCSP s.
Unfortunately, such a guess would be too simplistic and false.
A counterexample is the problem of solving linear equations over $\mathbb{Z}_2$, which is solvable in polynomial time by standard Gaussian elimination. However, the only non-trivial polymorphisms of the \CSP\ \kl(CSP){template} defining this problem are odd-arity parity functions $\tuple x \mapsto \sum x_i \pmod 2$. Since flipping any $x_i$ always flips the output, the influence of every coordinate is 1. Therefore, large influences in all polymorphisms are not necessarily an obstruction to tractability.
Nevertheless, in certain restricted regimes within Boolean \PCSP s, non-vanishing influences do turn out to be an indicator of hardness, and understanding when this happens is one of the motivations of the present paper.

Measuring a coordinate influence solely under the uniform distribution of the input space is usually insufficient, because \emph{unbalanced} inputs, with many more $0$s than $1$s or vice versa, are severely underrepresented. To address this issue, we usually incorporate other probability measures, such as the (product) $p$-\kl{biased distribution} (also called Binomial distribution).
To be more precise, given a Boolean function $f$ and $p \in [0,1]$, we define the influence of a coordinate $i$ under the $p$-biased distribution as the probability that flipping the $i$-th entry $x_i$ flips the output $f(\tuple x)$, given that each entry of $\tuple x$ is independently set to 1 with probability $p$ --- we denote this value with $\Inf[(p)]{f,i}$.

For classifying the \kl{Ordered PCSPs} in~\cite{brakensiek2021conditional}, the authors measure the coordinate influences under a yet different distribution, called the \kl{Shapley distribution}, which we also consider in the present paper.
The influences over the Shapley distribution, commonly known as \emph{Shapley values}, were originally introduced by Shapley and Shubik in the study of voting systems \cite{shapley_values_introduction}. Before being brought to the study of \PCSP s by \cite{brakensiek2021conditional}, they had also found multiple applications in other fields such as game theory \cite{shapley_values_game_theory_1, shapley_values_game_theory_2}.

\kl{Shapley values} have a neat combinatorial definition in the context of \kl(func){increasing} functions, that is, functions satisfying $f(\tuple x) \leq f(\tuple y)$ for any pair of tuples $\tuple x \leq \tuple y$:\footnote{By $\tuple x \le \tuple y$ we mean the coordinatewise order, i.e., $x_i \le y_i$ for all $i$. Such functions are also called \emph{monotone}; we use the name increasing in order to avoid confusion with decreasing functions.}
We uniformly draw a permutation $\sigma$ of the coordinates of $f$ and then, starting from the tuple $(0, \dots, 0)$, we flip the entries from $0$ to $1$ one by one, in the order determined by $\sigma$, until we reach $(1, \dots, 1)$. Since $f$ is increasing, at most once does the value of $f$ evaluated on the current tuple shifts from $0$ to $1$. Then, the Shapley value of a coordinate $i$ is the probability that this occurs exactly when flipping the value at the $i$-th coordinate.
Observe that this definition immediately implies that the sum of Shapley values in an increasing function is at most $1$.

One of our goals is to study the \kl{Shapley values} beyond \kl(func){increasing} functions.
We conclude with a beautiful connection between \kl{Shapley values} and $p$-\kl{biased} \kl{influences}, originally observed by Owen~\cite[Section 2]{owen1972} proved in~\cref{shapley:claim:connection_between_shapley_and_biased}:
he proved that a \kl{Shapley value} of a coordinate $i$ of a Boolean function $f$ can be expressed as
\[
    \int_0^1 \Inf[(p)]{f, i} \, dp.
\]

\subsection*{Subclasses of Boolean PCSPs}
As the understanding of all \kl(PCSP){Boolean} \PCSP s remains shallow, the research community has been focusing on smaller classes, enforcing structural or algebraic restrictions on \PCSP\ \kl(PCSP){templates}.
In the \kl(PCSP){Boolean} world, among the first systematic works in this direction is the conditional dichotomy for \kl{Ordered} \PCSP s of Brakensiek, Guruswami and Sandeep~\cite[Theorem 1.1]{brakensiek2021conditional}.
In algebraic terms, \AP\intro{Ordered} \PCSP s amount to restricting the \kl{polymorphisms} to \kl(func){increasing} functions.
\AP This was followed by the work of Banakh and Kozik, who provided a complete complexity dichotomy for \intro(PCSP){Linear Threshold} \PCSP s~\cite[Theorem 1.2]{injhardness}. These are \PCSP s whose all polymorphisms are \intro{linear threshold functions}, which take the form $\mathbf{1} [\sum a_i \cdot x_i > t]$. In this paper, we investigate through an analytic lens the following two extensions of Ordered and Linear Threshold \PCSP s:

\begin{enumerate}[leftmargin=*]
    \itemAP \intro(PCSP){Unate} \PCSP s are \kl(PCSP){Boolean} \PCSP s whose \kl{polymorphisms} are all \intro(func){unate} functions, meaning that each coordinate is \emph{increasing} or \emph{decreasing}: a coordinate $i$ of a function $f\colon \Bool^n \to \Bool$ is \AP\intro{increasing} if $f(\tuple x, 0, \tuple y) \le f(\tuple x, 1, \tuple y)$ for all $\tuple x \in \Bool^{i-1}$ and $\tuple y \in \Bool^{n-i}$; it is \reintro{decreasing} if the inequality is reversed.
    \kl(func){Unate} functions have been prominent in switching theory \cite{unate1961}, computational learning theory \cite{unate1998}, and complexity theory \cite{unate2023}.
    Austrin, H\r{a}stad and Martinsson \cite{austrin2025usefulness} have recently investigated \kl(func){unate} functions in the context of \PCSP s from a combinatorial point of view.
    Nevertheless, the complexity landscape of Unate \PCSP s remains largely unexplored.

    \itemAP \intro(PCSP){Polynomial Threshold} \PCSP s are \PCSP s whose \kl{polymorphisms} are all \AP\intro{polynomial threshold functions} (\PTF s) of degree bounded by a constant $k$ independent of the arity, meaning that each such polymorphism is of the form $\mathbf{1}[Q(\tuple x) > 0]$ for some real multilinear polynomial $Q$ of degree at most $k$. While \PTF s have been extensively studied in learning theory~\cite{ptf_learning_1, learning_ptf_2, learning_ptf_3, learning_ptf_4, learning_ptf_5} and from structural and extremal perspectives~\cite{gotsman_linial_ptfs, extremal_properties_of_ptfs, Harsha2009BoundingTS, bounding_ptfs_2, bounding_ptfs_3, bounding_ptfs_4}, Polynomial Threshold \PCSP s have not enjoyed much attention.
\end{enumerate}

\subsection*{Contributions} The main contribution of this paper is a general analytical theory for studying \kl{polymorphisms} of \kl(PCSP){Boolean} \PCSP s based on coordinate \kl{influence}. We identify two general phenomena which hint toward hardness or tractability of a Boolean \PCSP.

\begin{description}
    \item[Preservation of influence under random 2-to-1 minors.] If the sum of coordinate \kl{influences} of a \kl{Boolean function} over a \emph{\kl{reasonable}}\footnote{This includes both the \kl{biased distribution} and the \kl{Shapley distribution}.} probability distribution is not too large, and one randomly identifies its coordinates in pairs, then any coordinate influential in the original function retains a constant fraction of its influence in the new function with constant probability. In many cases, this allows for a polynomial time reduction from the \richconj, following a pattern similar to \cite[Theorem 4.8]{brakensiek2021conditional}. In particular, our result vastly generalizes \cite[Lemma 4.4]{brakensiek2021conditional}. As a showcase of its applicability, we prove a hardness result for both classes of \kl(PCSP){Unate} and \kl(PCSP){Polynomial Threshold} \PCSP s based on influence preservation over the $p$-\kl{biased distribution}.
    \begin{theorem}\label{main:theorem:hardness_for_ptfs_and_unate_biased_distribution}
        Suppose that $(\A, \B)$ is either a \kl(PCSP){Unate} or \kl(PCSP){Polynomial Threshold} \PCSP\ \kl(PCSP){template}. If
        \[
            \exists \, p, \delta > 0 : \forall \, f \in \Pol(\A, \B) : \max_i \, \Inf[(p)]{f, i} \geq \delta,
        \]
        then $\PCSP(\A, \B)$ is \NP-hard, assuming the \richconj.
    \end{theorem}
    Both hardness results also hold with the \kl{Shapley distribution} in place of the $p$-\kl{biased} one~(\cref{main:theorem:hardness_for_ptfs_over_shapley,main:theorem:hardness_for_unate_over_shapley}).

    \item[Sharp thresholds imply tractability.]
    Sharp threshold is a phenomenon where a property of a random system transitions from being almost certainly absent to almost certainly present over a small change in a parameter. In the context of \kl{Boolean functions}, the property of interest is ``$f(\tuple x) = 1$'', measured under the $p$-\kl{biased distribution} $\cube{p}$, where $p \in (0,1)$ is the parameter.
    The tractability result of \cite{brakensiek2021conditional} essentially asserts that an \kl{Ordered} \PCSP\ is solvable in polynomial time as long as its set of \kl{polymorphisms} admits a sharp threshold. We demonstrate that a similar connection arises in the larger class of \kl(PCSP){Unate} \PCSP s, with an appropriately defined notion of sharp threshold.
    \begin{theorem}[Informal]
        A \kl(PCSP){Unate} $\PCSP(\A, \B)$ is solvable in polynomial time if $\Pol(\A, \B)$ admits a \emph{sharp threshold}.
    \end{theorem}

    The study of sharp threshold behavior occupies a central role in probabilistic combinatorics and statistical physics.
    The formal mathematical investigation began in the late 1950s, when Erd\"os and R\'enyi discovered that many properties of random graphs (e.g. connectivity or appearance of a ``giant component'') do not emerge gradually as the density increases but rather depend on it in a ``sharp threshold manner'' \cite{erdos1959, erdos1960}.
    In fact, the later works of Bollob\'as and Thomason \cite{bollobas1987}, as well as Friedgut and Kalai \cite{friedgut1996sharp}, proved that \emph{every monotone\footnote{A graph property is monotone if it is closed under adding new edges.} graph property exhibits a sharp threshold behavior}.

    Sharp threshold behavior can also be observed in computational complexity theory, for the transition from ``easy'' to ``hard'' computational problems is usually abrupt.
    For example, Friedgut \cite{friedgut1999sat} discovered that the satisfiability probability of a random $k$-CNF formula on $n$ variables exhibits a sharp threshold when the ``clause density'' increases.
    In particular, sufficiently sparse $k$-CNF formulas are almost always satisfiable, while sufficiently dense ones are almost never satisfiable.
    Therefore the hardest formulas are those with density close to this critical threshold point.
\end{description}

\subsection*{Related work}
There have been attempts at obtaining an analytical proof of the CSP Dichotomy since 2009 (see~\cite{analyticalCSP,correlationdecay}); so far, these attempts have had limited success.

The already mentioned manuscript of Brakensiek, Guruswami, and Sandeep~\cite{brakensiek2021conditional} on \kl{Ordered} \PCSP s was the stepping stone towards the results obtained in the present paper. To our knowledge, their proof of a conditional dichotomy for Boolean Ordered \PCSP s was inspired by a talk given by Libor Barto~\cite{barto2018cyclic}, in which he showed how to obtain a symmetric threshold function by a suitable identification of coordinates of a cyclic\footnote{A function is cyclic if it is invariant under cyclic shifts of its arguments.} \kl{increasing} function.
Another noteworthy contribution on Boolean \PCSP s, although not relying on Fourier analysis, was the dichotomy for \kl(PCSP){Linear Threshold} \PCSP s (a subclass of both \kl(PCSP){Unate} and \kl(PCSP){Polynomial Threshold} \PCSP s) established by Banakh and Kozik~\cite[Theorem 1.2]{injhardness}.

Beyond the Boolean domain, Braverman, Khot, Lifshitz, and Minzer \cite{multislices} recently provided a general pipeline which, given a perfect-completeness \emph{dictatorship test} for a \CSP\ \kl(CSP){template}, produces a reduction from the \richconj\ to the problem of approximating (quantitatively) the corresponding \CSP.
Building on the ideas of \cite{multislices}, Banakh, Barsukov, and Nakajima \cite{pmmsnp} have obtained a reduction from the \richconj\ to PCSPs of the form $(\mathsf{R, NAE})$ where $\mathsf R$ is symmetric and connected, in a certain technical sense.
Although these two papers do not study the properties of the \kl{polymorphism} sets of the underlying \kl(PCSP){templates}, their hardness analysis shares several features with our results regarding influence preservation.

\subsection*{Organization of the paper}
In \cref{sec:notation}, the necessary background on (promise) constraint satisfaction problems and on the Fourier analysis of \kl{Boolean functions} is covered.
In \cref{sec:general}, we study influence preservation under \kl{minors}: we first exhibit adversarial \kl{minors} that destroy \kl{influence}, then introduce \kl{reasonable} distributions and state the Influence Preservation Lemma (\cref{lemma:influence_preservation}), whose proof is deferred to \cref{appendix:sec:influence_preservation}.
In \cref{sec:polynomials}, the Influence Preservation Lemma is applied to obtain hardness results for \kl(PCSP){Polynomial Threshold} \PCSP s, and in \cref{sec:monotonicity}, we study the phenomenon of sharp thresholds in \kl{Ordered} and \kl(PCSP){Unate} \PCSP s.
\Cref{sec:conclusions} concludes with open questions that follow from this work.

\section{Preliminaries and notation}\label{sec:notation}
We write tuples with a bar, e.g. $\tuple x \in \Bool^n$, and denote by $x_i$ the $i$-th entry of $\tuple x$.
If a tuple $\tuple x$ is indexed by $[n]$, then for any $I \subseteq [n]$ we denote by $\tuple x_I$ the subtuple of $\tuple x$ indexed by $I$.

By $f(n) \leq \mathcal{O}_{a,b}(g(n))$ we denote that $f(n) \leq C \cdot g(n)$ for all sufficiently large $n$ and some positive constant $C$ depending on $a,b$. The other standard asymptotical bounds are defined similarly: $o_{a,b}(\cdot)$, $\Omega_{a,b}(\cdot)$, and $\Theta_{a,b}(\cdot)$.

\subsection{Constraint Satisfaction Problems}\label{sec:prelims:csp}

This section briefly defines \CSP s, \PCSP s, and related notions. For a more elaborate overview of this topic, we refer the reader to \cite{Bible}.

\AP
Given a \intro{universe} $A$ and $r \ge 1$, an $r$-ary \intro{relation} over $A$ is a set $R \subseteq A^r$.
\AP A \intro{relational structure} $\relstr A$ is a tuple $(A; R_1^{\relstr A}, \dots, R_\ell^{\relstr A})$ where each $R_i^{\relstr A}$ is a relation over $A$.\footnote{Relational structures are usually denoted using the blackboard font, and the underlying \kl{universe} is denoted with the same letter in the normal font.}
\AP We call two relational structures \intro{similar} if they have the same number of relations ($\ell$), and the $i$-th relation has the same arity in both for every $i \in [\ell]$.
\AP For two \kl{similar} structures $\relstr A$ and $\relstr B$, a \intro{homomorphism} from $\relstr A$ to $\relstr B$ is a function $h : A \to B$ that maps each tuple $(a_1, \dots, a_r) \in R_i^{\relstr A}$ to $(h(a_1), \dots, h(a_r)) \in R_i^{\relstr B}$, for every $i$.
We write $\A \to \B$ to denote that $\A$ admits a \kl{homomorphism} to $\B$.

This is sufficient to define \CSP s.
In the~\nameref{sec:intro}, we briefly introduced \CSP s in the language of variables and constraints.
Here we will use the equivalent formulation in terms of \kl{relational structures} and \kl{homomorphisms}.

\begin{definition}[\intro{CSP}]\AP
    Given a \kl{relational structure} $\relstr B$, the problem $\CSP(\relstr B)$ comes in two flavors. For a given input structure $\relstr I$ which is \kl{similar} to $\relstr B$, it consists in
    \begin{itemize}
        \itemindent=2em
        \item[(decision)] deciding whether $\relstr I \to \relstr B$, or
        \item[(search)] finding a \kl{homomorphism} from $\relstr I$ to $\relstr B$.
    \end{itemize}
    \AP We say that $\relstr B$ is the \intro(CSP){template} of $\CSP(\relstr B)$.
\end{definition}

The decision and search variants of every \CSP\ are polynomial-time equivalent~\cite[Section 3.2]{polymorphisms_and_how_to_use_them}.
Similarly, we define two versions of \PCSP s:

\begin{definition}[\intro{PCSP}]\AP
    Let $\relstr A$ and $\relstr B$ be two \kl{similar} structures such that $\relstr A \to \relstr B$.
    Given an input structure $\relstr I$ \kl{similar} to $\relstr A$ and $\relstr B$, the problem $\PCSP(\relstr A, \relstr B)$ comes in two variants:
    \begin{itemize}
        \itemindent=2em
        \item[(decision)] output \textsf{YES} if $\relstr I \to \relstr A$, and output \textsf{NO} if $\relstr I \not\to \relstr B$.
        \item[(search)] find a \kl{homomorphism} from $\relstr I$ to $\relstr B$, given the promise that $\relstr I \to \relstr A$.
    \end{itemize}
    \AP We say that the pair $(\relstr A, \relstr B)$ is the \intro(PCSP){template} of $\PCSP(\relstr A, \relstr B)$.
\end{definition}

Search is always at least as hard as the decision counterpart, but, unlike for CSPs, for PCSPs the converse relation is currently unknown.\footnote{Recently, there was a certain progress in the negative direction by \cite{larrauri2025}.}
In our results, we provide algorithms for the search version and proofs of hardness for the decision version.

\paragraph{Polymorphisms.}
\AP
The arity $\ar(f)$ of a function $f \colon A^n \to B$ is the number of its arguments, that is, $n$. A function $f$ is a \intro{polymorphism} of, or \reintro{compatible with}, the \kl(PCSP){template} $(\relstr A, \relstr B)$ if it is a \kl{homomorphism} of the form $\relstr A^n \to \relstr B$.
The $n$-th power of a \kl{relation} $R \subseteq A^r$ is $R^n \subseteq (A^n)^r$, defined coordinatewise by $(\tuple a^1, \dots, \tuple a^r) \in R^n$ iff $(a^1_i, \dots, a^r_i) \in R$ for each $i$; the $n$-th power $\relstr A^n$ of a \kl{relational structure} $\relstr A$ has \kl{universe} $A^n$ and, for each relation $R$ of $\relstr A$, the relation $R^n$.
We say that a set of functions $\mathcal F$ is \kl{compatible with} the \kl(PCSP){template} $(\A, \B)$ if every member of $\mathcal F$ is.
The set of all polymorphisms of a fixed \kl(PCSP){template} $(\A, \B)$ is denoted by $\Pol(\A, \B)$.
Note that $\Pol(\A, \B)$ effectively captures the complexity of $\PCSP(\A, \B)$:
\begin{theorem}[{\cite[Theorem 6.1]{brakensiek2017promise}}]
    If $\Pol(\A, \B) \subseteq \Pol(\A', \B')$, then $\PCSP(\A', \B')$ is polynomial-time reducible to $\PCSP(\A, \B)$.
\end{theorem}
\noindent\AP For any $\A$ and $\B$, the set $\Pol(\A, \B)$ is a \intro{minion}, that is, a subset of $\{f\colon A^n \to B \mid n \ge 1\}$ closed under taking \emph{minors}:
\begin{definition}[\intro{Minor}]\AP
    The minor of $f \colon A^n \to B$ given by a \intro{minor map} $\pi \colon [n] \to [m]$ is the function $\minor{f}{\pi}$ defined as
    \[
    (x_1, \dots, x_m) \mapsto f(x_{\pi(1)}, \dots, x_{\pi(n)}).
    \]
\end{definition}

\AP
A function $f\colon A^n\rightarrow A$ is \intro{idempotent} if $f(a,\dotsc,a)=a$ for all $a\in A$.
It is a \intro{projection}, or \reintro{dictator}, if, for some $i$, $f(a_1,\dotsc,a_n) = a_i$ for all $a_1,\dotsc,a_n \in A$.
Every \kl{minor} of a projection is a projection, and every minor of an idempotent function is idempotent.
A \kl{minion} is \reintro{idempotent} if its members are idempotent.
\AP Finally, we say that a coordinate $i$ of a function $f \colon A^n \to B$ is \intro{essential} if there exist $a_1, \dots, a_n, b_i \in A$ such that $f(a_1, \dots, a_i, \dots, a_n) \neq f(a_1, \dots, b_i, \dots a_n)$. A function $f$ doesn't depend on its non-essential coordinates.

\paragraph{Tractability of Boolean PCSPs.}
\AP In the current literature, algorithms solving the search variant of \kl(PCSP){Boolean} \PCSP s usually rely on the occurrence of one of the following \kl{minions}, each generated by a single family of functions of the indicated form.
\[
    \mbox{\phantomsection}\text{``constant 0''} = \left\langle \tuple x \mapsto 0 \right\rangle
    \quad\quad
    \intro*\XOR = \left\langle \tuple x \overset{\intro*\gen{xor}{m}}{\longmapsto} \sum_{i=1}^m x_i \bmod 2 \right\rangle_{\text{odd } m \in \mathbb N}
\]
\[
    \mbox{\phantomsection}\intro*\MAX = \left\langle \tuple x \overset{\intro*\gen{max}{m}}{\longmapsto} \max\{x_1, \dots, x_m\} \right\rangle_{m \in \mathbb N}
    \quad
    \intro*\MIN = \left\langle \tuple x \overset{\intro*\gen{min}{m}}{\longmapsto} \min\{x_1, \dots, x_m\} \right\rangle_{m \in \mathbb N}
\]
\[
    \mbox{\phantomsection}\intro*\AT = \left\langle (\tuple x, \tuple y) \overset{\intro*\gen{at}{m}}{\longmapsto} \mathbf 1 {\left[\sum_{i=1}^{m+1} x_i > \sum_{i=1}^m y_i\right]} \right\rangle_{m \in \mathbb N}
    \quad
    \intro*\THR{t} = \left\langle \tuple x \overset{\intro*\genthr{t}{m}}{\longmapsto} \mathbf 1 {\left[\sum_{i=1}^m x_i > tm\right]} \right\rangle_{m \in \mathbb N \, : \, tm \notin \mathbb N}
\]
\AP The list is complemented with the \emph{negated} versions of these minions:
the negation of $f$ is $\tuple x\mapsto 1-f(\tuple x)$ and the \intro{negation} of a minion consists of negations of all its members.
\begin{theorem}[\cite{brakensiek2017promise,ficak2019symmetric,injhardness}]
\label{th:tractablelist}
    If $\Pol(\A, \B)$ includes one of the \kl{minions} on the list above, then search $\PCSP(\A, \B)$ is solvable in polynomial time.
\end{theorem}

\subsection{Fourier analysis of Boolean functions}\label{sec:prelims:fourier}

We now formally introduce notions and facts from the Fourier analysis of \kl{Boolean functions}, which are required for the proofs in the following sections. In particular, we define the notion of \emph{noise sensitivity} and related concepts, which play a crucial role in \cref{appendix:sec:noise_stability_of_ptfs}. We refer to \cite{odonnell} for a more comprehensive treatment of this topic.

\paragraph{Extended notation.}
For a tuple $\tuple x \in \Bool^n$, we denote by $\tuple x \oplus i$ the tuple obtained from $\tuple x$ by flipping its $i$-th entry, and more generally $\tuple x \oplus I$ is obtained by flipping the $i$-th entry for every $i \in I$.
For tuples $\tuple x \in \Bool^n$ and $\tuple y \in \Bool^m$, we denote their concatenation by $\tuple x \tuple y \in \Bool^{n+m}$; in particular, we write $\tuple x 0, \tuple x 1 \in \{0,1\}^{n+1}$ for the tuples obtained by appending $0$ or $1$ to $\tuple x$. Furthermore, we write $\tuple x^{\, i \to 0}$ and $\tuple x^{\, i \to 1}$ for the tuples obtained from $\tuple x$ by setting the value of $x_i$ to $0$ or $1$, respectively.
The \emph{Hamming weight} $\ham{\tuple x}$ of a tuple $\tuple x \in \Bool^n$ is the number of its entries equal to $1$.

Given a probability distribution $\Omega \colon \Bool^n \to [0,1]$, we write $\tuple x \sim \Omega$ to denote that $\tuple x$ is sampled according to $\Omega$.
By $\Omega(\tuple x)$ we denote the probability measure of $\{\tuple x\}$ under $\Omega$.\footnote{For conciseness, we usually call it just \emph{the measure of $\tuple x$}.} If $A \subseteq \{0,1\}^n$, by $\Omega(A)$ we denote the measure of $A$, i.e., the sum of measures of all elements of $A$.
We use $\E[\Omega]{f}$ as a shorthand for $\E[\tuple x \sim \Omega]{f(\tuple x)}$.

\paragraph{Influence.}
We shall formally define the coordinate influence discussed in~\nameref{sec:intro}; it is parameterized by the probability distribution over inputs, under which the influence is measured.
Although we study Boolean functions of the form $\Bool^n \to \Bool$, the Fourier analysis theory is mostly oblivious of the codomain, therefore many notions will be defined for functions taking values in the reals.

\begin{definition}[\intro{Influence}]\AP
    Suppose that $f\colon \Bool^n \to \mathbb{R}$ and $i \in [n]$. The \emph{influence} of coordinate $i$ in $f$ under a probability distribution $\Omega$ is defined as
    \[
        \Inf[\Omega]{f, i} = \E[\tuple x \sim \Omega]{\left( f(\tuple x) -  f(\tuple x \oplus i) \right)^2}.
    \]
    If $f$ is a \kl{Boolean function}, the \kl{influence} of coordinate $i$ can be interpreted as the probability that flipping the value of the $i$-th coordinate changes the function value.

    \AP The \intro{total influence} of $f$, sometimes called its \reintro{average sensitivity}, is the sum of individual influences of all coordinates, that is, $\sum_{i=1}^n \Inf[\Omega]{f,i}$, which we denote with $\I[\Omega]{f}$.
\end{definition}

\AP
The first probability distribution introduced in the~\nameref{sec:intro} was the $p$-\intro{biased distribution}.
Formally, let $\cube{p}$ denote the Bernoulli distribution over $\Bool$, i.e., $\cube{p}(1) = p$.
Then $\cube{p, n}$ is the product $p$-biased distribution over $\{0,1\}^n$, where each bit is drawn independently from $\cube{p}$.
For conciseness, we sometimes write $\cube{p}$ instead of $\cube{p, n}$ if the dimension is irrelevant or clear from the context.
We use $\Pr_p[\cdot]$, $\E[p]{f}$, and $\Inf[(p)]{f,i}$ as shorthands for $\Pr_{\tuple x \sim \cube{p}}[\cdot]$, $\E[\tuple x \sim \cube{p}]{f(\tuple x)}$, and $\Inf[\cube{p}]{f,i}$, respectively.
The influence $\Inf{f, i}$ without an explicit distribution implicitly assumes the uniform distribution $\cube{1/2}$.
More generally, for $\tuple p = (p_1,\dots,p_n) \in [0,1]^n$ we denote by $\cube{\tuple p}$ the product distribution over $\Bool^n$ in which the $i$-th bit is drawn independently from $\cube{p_i}$.

\AP
The other notion of influence mentioned in the~\nameref{sec:intro} were the \emph{Shapley values}.
We state the analytical definition of Shapley values, which coincides with the described combinatorial interpretation in the case of increasing functions (see \cite[Exercise 8.31]{odonnell}).

\begin{definition}[\intro{Shapley distribution}]\AP
    For every $n \geq 1$, by $\shap[n]$ we denote the \emph{Shapley distribution} over $\{0,1\}^n$, which is defined as
    \[
        \shap[n](\tuple x) = \left( (n + 1) \cdot \binom{n}{\ham{\tuple x}} \right)^{-1}.
    \]
    By the \emph{Shapley value} of a coordinate $i \in [n]$, we mean its \kl{influence} under the \kl{Shapley distribution}: $\Inf[\shap]{f, i}$.
\end{definition}

We now prove the connection between the \kl{Shapley values} and the \kl{influences} over $p$-\kl{biased distributions} stated in the~\nameref{sec:intro}.

\begin{proposition}[{{Based on \cite[Section 2]{owen1972}}}]\label{shapley:claim:connection_between_shapley_and_biased}
    Suppose that $f \colon \{0,1\}^n \to \{0,1\}$ and $i \in [n]$. Then
    \[
        \Inf[\shap]{f, i} = \int_0^1 \Inf[(p)]{f, i} \, dp.
    \]
\end{proposition}
\begin{proof}
    \newcommand{\Piv}[0]{\mathbf{Piv}}
    Fix a function $f$ of arity $n$ and a coordinate $i \in [n]$. Let $\Piv(i)$ be the set of points $\tuple x \in \{0,1\}^n$ such that $f(\tuple x) \neq f(\tuple x \oplus i)$. We rewrite the integral on the right-hand side of the statement as follows:
    \begin{align*}
        \int_0^1 &\Inf[(p)]{f, i} \, dp = \int_0^1 \sum_{\tuple x \in \Piv(i)} \cube{p}(\tuple x) \, dp
        \\
        &= \sum_{\tuple x \in \Piv(i)} \int_0^1 \cube{p}(\tuple x) \, dp = \sum_{\tuple x \in \Piv(i)} \int_0^1 p^{|\tuple x|}(1-p)^{n - |\tuple x|} \, dp.
    \end{align*}
    Notice that we obtained the Beta function $B(z_1, z_2) = \int_0^1 t^{z_1 - 1}(1-t)^{z_2-1} \, dt$. Using the well known identity $B(z_1+1,z_2+1) = (z_1 ! \cdot z_2!)/(z_1 + z_2 + 1)!$ that holds for non-negative integers $z_1, z_2$, we obtain
    \begin{align*}
        \sum_{\tuple x \in \Piv(i)} &\int_0^1 p^{\ham{\tuple x}}(1-p)^{n - \ham{\tuple x}} \, dp \\
        &= \sum_{\tuple x \in \Piv(i)} B(\ham{\tuple x} + 1, n - \ham{\tuple x} + 1)  \\
        &= \sum_{\tuple x \in \Piv(i)}\frac{\ham{\tuple x}! \cdot (n - \ham{\tuple x})!}{(n + 1)!} = \Inf[\shap]{f, i}. \qedhere
    \end{align*}
\end{proof}

\paragraph{Fourier decomposition.}
Suppose that $p \in (0,1)$. We consider $L^2(\Bool^n, \cube{p})$ — the Hilbert space of functions $f : \{0,1\}^n \to \mathbb{R}$ equipped with an inner product $\ip{f}{g} = \E[p]{f(\tuple x) \cdot g(\tuple x)}$ and standard norm $\norm{f} = \sqrt{\ip{f}{f}}$. We distinguish a set of functions in $L^2(\{ 0,1\}^n, \cube{p})$ called \emph{Fourier characters}, defined as follows.

\begin{definition}[\intro{Fourier characters}]\AP
    Let $p \in (0, 1)$ and $i \in [n]$. The \kl{Fourier character} corresponding to singleton $\{i\}$ is the function $\chi_i^{(p)} \in L^2(\Bool^n, \cube{p})$ defined as
    \[
        \chi_i^{(p)}(\tuple x) = \frac{x_i - p}{\sqrt{p(1-p)}}.
    \]
    \noindent
    In general, the \kl{Fourier character} corresponding to set $S$ is $\chi_S^{(p)} = \prod_{i \in S} \chi_i^{(p)}$, for every $S \subseteq [n]$. If the distribution is clear from the context, we simply write $\chi_i$ and $\chi_S$.
\end{definition}
\AP
It is well known that \kl{Fourier characters} form an orthonormal basis of $L^2(\Bool^n, \cube{p})$. Therefore, every function $f : \Bool^n \to \mathbb{R}$, seen as an element of $L^2(\{0,1\}^n, \cube{p})$, has a unique representation $f = \sum_{S \subseteq[n]} \hat{f}(S) \chi_S$, where $\hat{f}(S) = \ip{f}{\chi_S}$. To emphasize the underlying distribution, we sometimes write $f^{(p)}$ and $\hat{f}^{(p)}(S)$. \AP This representation is called the \intro{Fourier decomposition}, while the collection $\{ \hat{f}(S) : S \subseteq [n]\}$ are referred to as \intro{Fourier coefficients}. Observe that $\E[p]{\chi_S(\tuple x)} = 0$ for every non-empty $S$, which implies that $\E[p]{f} = \hat{f}(\emptyset)$. \AP Orthonormality of \kl{Fourier characters} yields the classic \intro{Parseval-Plancherel identity}:
\begin{gather*}
    \ip{f}{g} = \sum_{S \subseteq [n]} \hat{f}(S) \cdot \hat{g}(S) \qquad \text{and} \qquad \norm{f}^2 = \sum_{S \subseteq [n]} \hat{f}(S)^2.
\end{gather*}

The influences over $p$-\kl{biased distributions} are tightly connected to the \kl{Fourier decomposition}.

\begin{restatable}{proposition}{alternativedefinitionofinfluence}\label{proposition:influence:alternative_definitions_of_inf}
    Suppose that $p \in (0,1)$ and $f \colon \Bool^n \to \mathbb{R}$. For every $i \in [n]$ we have
    \[
        \Inf[(p)]{f, i} = \frac{1}{p(1-p)} \cdot \sum_{S \ni i} \hat{f}^{(p)}(S)^2.
    \]
\end{restatable}

\begin{proof}
    Fix any $i \in [n]$. Consider the function $g \colon \{0,1\}^n\to \mathbb{R}$ defined as
    \[
        g(\tuple x) = f(\tuple x) - \E[s \sim \cube{p}]{f(\tuple x^{\, i \to s})}.
    \]
    We want to find the \kl{Fourier decomposition} of $g$. From the fact that $\E[p]{\chi_i} = 0$ we obtain that
    \[
        \E[s \sim \cube{p}]{ f(\tuple x^{\, i \to s}) } = \sum_{S \subseteq [n]} \hat{f}(S) \cdot \E[s \sim \cube{p}]{ \chi_S(\tuple x^{\, i \to s}) } = \sum_{S \not \ni i} \hat{f}(S) \cdot \chi_S(\tuple x).
    \]
    Therefore, we have $g = \sum_{S \ni i} \hat{f}(S) \cdot \chi_S$ and the \kl{Plancherel identity} gives $\norm{g}^2 = \sum_{S \ni i} \hat{f}(S)^2$. Let $\Delta(\tuple x) = f(\tuple x^{\, i \to 1}) - f(\tuple x^{\, i \to 0})$. Observe that $\Delta(\tuple x)$ does not depend on $x_i$. It follows that
    \[
        g(\tuple x) = \begin{cases}
            - p \cdot \Delta(\tuple x) & \text{ if } x_i = 0, \\
            (1-p) \cdot \Delta(\tuple x) & \text{ if } x_i = 1.
        \end{cases}
    \]
    We compute $\norm{g}^2$ directly:
    \begin{align*}
        \norm{g}^2 &= (1-p) \cdot \E[p]{ p^2 \Delta(\tuple x)^2 \, \big | \, x_i = 0 } + p \cdot \E[p]{ (1-p)^2 \Delta(\tuple x)^2 \, \big | \, x_i = 1 } \\
         &= \big[(1-p) p^2  + (1-p)^2 p \big] \cdot \E[p]{ \Delta(\tuple x)^2 } \\
         &= p(1-p) \cdot \E[p]{ \big( f(\tuple x) - f(\tuple x \oplus i) \big)^2 }. \qedhere
    \end{align*}
\end{proof}

In particular, \cref{proposition:influence:alternative_definitions_of_inf} implies the following connection between \kl{total influence} and \kl{Fourier decomposition} over the $p$-\kl{biased distribution}.

\begin{corollary}\label{corollary:influence:total_influence_formula}
    Suppose that $p \in (0,1)$ and $f \colon \{0,1\}^n \to \mathbb{R}$. Then, the \kl{total influence} of $f$ is
    \[
        \I[(p)]{f} = \frac{1}{p(1-p)} \cdot \sum_{S \subseteq [n]} |S| \cdot \hat{f}^{(p)}(S)^2.
    \]
\end{corollary}

\clearpage
\paragraph{Noise sensitivity.}
\begin{definition}[\intro{Noise operator}]\AP
    Given $\delta, p \in [0,1]$ and $\tuple x \in \{0,1\}^n$, the $(p, 1-\delta)$-\intro{noisy distribution} of $\tuple x$ denotes the distribution over elements $\tuple y \in \{0,1\}^n$ defined as follows: for every $i \in [n]$ independently, let
    \[
        y_i = \begin{cases}
            x_i & \text{ with probability $1-\delta$}, \\
            \text{drawn from } \cube{p} & \text{ with probability } \delta.
        \end{cases}
    \]
    The $(p, 1-\delta)$-noisy distribution of $\tuple x$ is denoted by $\N[(p)]{\tuple x, 1-\delta}$.

    The \kl{noise operator} $\mathbf{T}_{1-\delta}^{(p)}$ on the space $L^2(\Bool^n, \cube{p})$ is the operator that assigns to every function $f$ a corresponding function defined as
    \[
        \TN[(p)]{f,1-\delta}(\tuple x) = \E[\tuple y \sim \N[(p)]{\tuple x,1-\delta}]{ f(\tuple y)}.
    \]
\end{definition}

We note that if $\tuple x \sim \cube{p, n}$ and $\tuple y \sim \N[(p)]{\tuple x, 1-\delta}$, then the marginal distribution of $\tuple y$ is $\cube{p,n}$. The parameter $\delta$ is the level of \emph{noise} between the original point and its perturbation; the lower the noise (the closer $\delta$ is to $0$), the more similar to $\tuple x$ we expect $\tuple y$ to be. This is further reflected in the following fact, which states that the \kl{Fourier decomposition} of $\TN[(p)]{f,1-\delta}$ resembles the decomposition of $f$, as long as the noise parameter is not far from $0$.

\begin{proposition}\label{proposition:noise:coefficients_of_noise_operator}
    Suppose that $f \colon \Bool^n \to \mathbb{R}$ and $\delta \in [0,1]$. The \kl{Fourier decomposition} of $\TN[(p)]{f,1-\delta}$ is
    \[
        \TN[(p)]{f,1-\delta} = \sum_{S \subseteq[n]} (1-\delta)^{|S|} \cdot \hat{f}^{(p)}(S) \cdot \chi_S.
    \]
\end{proposition}

\begin{proof}
    Observe that $\mathbf{T}^{(p)}_{1-\delta}$ is a linear operator. Therefore, it suffices to show that $\TN[(p)]{\chi_S,1-\delta} = (1-\delta)^{|S|} \cdot \chi_S$. This holds trivially if $S = \emptyset$. If $S$ is not empty, suppose that $\tuple x \in \Bool^n$ and $\tuple y \sim \N[(p)]{\tuple x, 1-\delta}$. We obtain:
    \begin{align*}
        \TN[(p)]{\chi_S,1-\delta}(\tuple x) &= \E[\tuple y]{\chi_S(\tuple y)} = \prod_{i \in S} \E[y_i]{ \chi_i(y_i) } \\
        &= \prod_{i \in S} (1-\delta) \cdot \chi_i(x_i) + \delta \cdot \E[p]{\chi_i} \\
        &= (1-\delta)^{|S|} \cdot \chi_S(\tuple x). \qedhere
    \end{align*}
\end{proof}

In pair with the \kl{noise operator} is the notion of \textit{noise sensitivity}, which is a measure of how prone to perturbation of input values a function is.

\begin{definition}[\intro{Noise sensitivity}]\AP
    Let $p, \delta \in [0,1]$ and $f \colon \Bool^n \to \Bool$. Let $\tuple x \sim \cube{p, n}$ be a $p$-biased vector and $\tuple y \sim \N[(p)]{\tuple x, 1-\delta}$. The \kl{noise sensitivity} over $\cube{p}$ of $f$ at $1-\delta$ is defined as
    \[
        \NS[(p)]{f,1-\delta}= \Pr \big[ f(\tuple x) \neq f(\tuple y)].
    \]
\end{definition}

\begin{claim}\label{noise:claim:connection_between_noise_and_coefficients}
    Suppose that $p, \delta \in [0,1]$ and $f \colon \{0,1 \}^n \to \{0,1\}$. Then
    \[
        \ip{f^{(p)}}{\TN[(p)]{f, 1-\delta}} = \E[p]{f} - \frac{1}{2} \cdot \NS[(p)]{f, 1-\delta}.
    \]
\end{claim}

\begin{proof}
    Fix $p,\delta$ and $f$. Let $\tuple x \sim \cube{p,n}$ and $\tuple y \sim \N[(p)]{\tuple x, 1-\delta}$. Observe that $\ip{f}{\TN[(p)]{f, 1-\delta}} = \E{f(\tuple x) f(\tuple y)}$. We obtain:
    \begin{align*}
        \E{f(\tuple x) f(\tuple y)} &= \frac{1}{2} \left(\E{f(\tuple x)} + \E{f(\tuple y)} - \Pr[f(\tuple x) \neq f(\tuple y)] \right) \\
        &= \E[p]{f} - \frac{1}{2} \cdot \NS[(p)]{f, 1-\delta}. \qedhere
    \end{align*}
\end{proof}

\section{Influence and random 2-to-1 minors}\label{sec:general}
In this section, we aim to identify and generalize the phenomenon of what we call \emph{influence preservation} under random \kl{2-to-1} \kl{minors}.

Recall the \kl{multiple choice hardness condition} which implies \NP-hardness of a \PCSP\ \kl(PCSP){template} $(\relstr A, \relstr B)$. In order to apply it, we need a witness assignment $\Sel$: a criterion to select a constant number of ``important'' coordinates of every \kl{polymorphism} $f \in \Pol(\relstr A, \relstr B)$, which is compatible with taking \kl{minors}. In the Fourier analysis --- and hence in our analytical approach to \kl(PCSP){Boolean} \PCSP s --- \kl{influence} is the canonical indicator of a coordinate's importance. This prompts us to select largest-influence coordinates to obtain hardness reductions.
The serious obstacle to this strategy is~\cref{item2:multiple-choice} of the \kl{multiple choice hardness condition}:
upon identifying coordinates of a function $f$, at least one of the selected coordinates of $f$ must remain influential in the obtained \kl{minor}, which is not guaranteed in most cases.

The hardness results of \cite{brakensiek2021conditional} and \cite{injhardness} for \kl{Ordered} and \kl(PCSP){Linear Threshold} \PCSP s, respectively, overcome this obstacle by relaxing the~\cref{item2:multiple-choice}: by giving up selected-coordinate preservation under certain types of \kl{minors} and using different notions of \kl{influence}, each tailored to its application. However, to study general \kl(PCSP){Boolean} \PCSP s, one is likely forced to settle for the classical notion of influence. Consequently, the central question of this section becomes that of influence preservation:
\begin{center}
    \textit{When can we deduce that $\Inf{\minor{f}{\pi}, \pi(i)}$ is significant from the assumption that $\Inf{f, i}$ is?}
\end{center}

\subsection{General minors are poorly behaved}\label{sec:influence-preservation-edversary-examples}

Unfortunately, \kl{influence} is not preserved under general \kl{minors}, even within very simple classes of functions.
The first obstacle to influence preservation is that general \kl{minors} can identify coordinates in groups of arbitrarily large size. Allowing this, significant coordinate influence can vanish even within the class of \kl(func){linear thresholds} with positive weights (which are, in particular, \kl{increasing} functions).

\begin{example}\label{example:increasing-linear-thresholds-with-no-influence-preservation}
    There exists a family of \kl{linear threshold functions} $\left\{ f_n : \{0,1\}^{4n+1} \to \{0,1\} \right\}_{n \geq 1}$ defined by linear polynomials with positive weights, and a family of mappings $\{ \pi_n : [4n+1] \to [n+2] \}_{n \geq 1}$, such that, as $n \to \infty$,
    \begin{enumerate}
        \item $\Inf[(1/2)]{f_n, 1} \to 1$ and $\Inf[(1/2)]{f_n, i} \to 0$ for every $i \in \{2, \dots, 4n+1\}$,
        \item $\Inf[(1/2)]{\minor{(f_n)}{\pi_n}, 1} \to 0$ while $\Inf[(1/2)]{\minor{(f_n)}{\pi_n}, 2} \to 1$.
    \end{enumerate}
\end{example}

\begin{proof}
    Consider the linear polynomials
    \[
        Q_n(\tuple x) = \underbrace{n^{3/4} \cdot x_1}_{\text{``big" coordinate}} + \underbrace{x_2 + \dots + x_{3n+1}}_{3n \text{ ``small" coordinates}} + \underbrace{x_{3n+2} + \dots + x_{4n+1}}_{n \text{ ``small" coordinates}} - \underbrace{2n + \tfrac{1}{2} n^{3/4}}_{\text{threshold}}.
    \]
    By concentration of the Bernoulli sum, the $4n$ small coordinates sum to within $(2n \, \pm \, n^{3/4}/4)$ with probability $1 - o(1)$, so the sign of $Q_n$ is determined by $x_1$ with high probability and all other \kl{influences} tend to $0$.

    Identifying the $3n$ small coordinates into a single coordinate yields one with weight $3n$, which alone exceeds the threshold. Meanwhile, the total weight of all remaining coordinates is $n + o(n)$, which falls below the threshold for large $n$, so coordinate $1$ loses all \kl{influence}. Conversely, the merged coordinate has influence exactly $1$: if it equals $1$, then $Q_n \geq 3n - 2n - \tfrac{1}{2}n^{3/4} > 0$, and if it equals $0$, then $Q_n \leq n^{3/4} + n - 2n - \tfrac{1}{2}n^{3/4} < 0$ for large $n$, regardless of the remaining coordinates.
\end{proof}

The functions in \cref{example:increasing-linear-thresholds-with-no-influence-preservation} do admit influence preservation in a relaxed sense, via the ``layering" technique of \cite{brandts2021llc}, which replaces direct influence preservation with a guarantee along a constant-length chain of \kl{minors} --- used, respectively, for \kl{polynomial thresholds} with non-negative weights and for general \kl(func){linear thresholds}\footnote{In \cite{injhardness}, a further relaxation in the form of ``injectiveness" is introduced, which eliminates \kl{minor maps} that identify influential coordinates with one another. It is also shown therein that dropping the injectiveness property leads to poor influence preservation within general linear thresholds.} in \cite{michno2025booleanpcspspolynomialthreshold, injhardness}. However, we do not pursue layerings in this paper. Instead, we restrict the \emph{size} of groups of coordinates a minor may identify --- primarily to pairs.

\paragraph{2-to-1 minors.} From now on, we will only consider \kl{minor maps} which identify coordinates in pairs:

\begin{definition}[\intro{2-to-1 map}]\AP
    A map $\pi\colon [2n] \to [n]$ is called \emph{2-to-1} if $|\pi^{-1}(i)| = 2$ for all $i \in [n]$; the corresponding \kl{minors} are called \emph{2-to-1 minors}.
\end{definition}

Influence preservation under \kl{2-to-1 minors} yields hardness results for \PCSP s which are \emph{conditional}, depending on variants of the Unique Games Conjecture --- mainly the \emph{$2$-to-$1$ Conjecture} \cite{khot2002dto1}, whose imperfect-completeness version has been established in a celebrated sequence of papers \cite{2TO1_1, 2TO1_side1, 2TO1_side2, 2TO1_2, 2TO1_3, 2TO1_4}.
The proof of the perfect-completeness variant is being claimed by OpenAI~\cite{2to1-perfect-completeness}.

The family of functions presented in \cref{example:increasing-linear-thresholds-with-no-influence-preservation} enjoys good influence preservation under \kl{2-to-1 minors}. Unfortunately, influence preservation still fails in quite simple classes of functions: \cite[Theorem 5.1]{brakensiek2021conditional} gives an example in the setting of \kl(func){increasing} functions, while we give another example for \kl{polynomial threshold functions} of degree $2$.

\begin{figure}[!htb]
    \centering
    \includegraphics[scale=0.8,trim=0 0 0 0,clip]{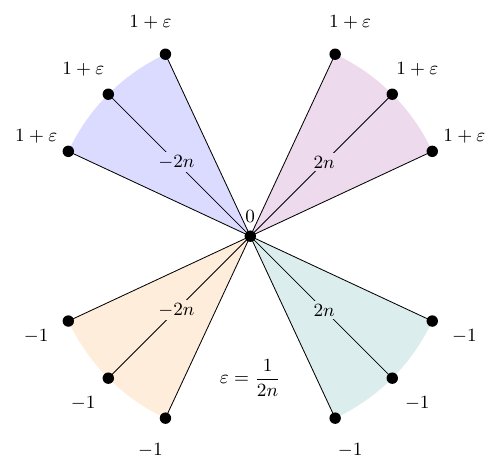}
    \caption{The weighted graph defining $\gen{fan}{n}$. The middle vertex has linear weight $0$ and is joined by an edge to each of the remaining $4n$ vertices, which are split into four blocks of $n$ vertices each. Within a block, all vertices have the same linear weight (displayed next to the vertices) and all edges have the same quadratic weight (displayed on the edges); here $\varepsilon = \tfrac{1}{2n}$.}\label{fig:fan}
\end{figure}

\begin{example}\label{ptfs:example:fan}\AP
    For every $n \geq 1$, let $\intro*\gen{fan}{n}\colon \Bool^{4n+1} \to \Bool$ be the degree-$2$ \kl{polynomial threshold function} defined as the sign of the polynomial encoded by the weighted graph in \cref{fig:fan}, where vertices are coordinates and vertex/edge weights are linear/quadratic coefficients.
\end{example}

The function $\gen{fan}{n}$ does not admit influence preservation under non-random $2$-to-$1$ \kl{minors} over $\cube{1/2}$: for every fixed constant-size selection of coordinates, there is a $2$-to-$1$ minor that routes all \kl{influence} away from it. Namely, the middle vertex of $\gen{fan}{n}$ has constant \kl{influence} under $\cube{1/2}$. However, for every coordinate set $I \subseteq [4n+1]$ of size at most $n-1$, there exists a \kl{minor map} $\pi\colon [4n+1] \to [2n+1]$ that (1) identifies all vertices except the middle one into pairs, (2) does not pair coordinates inside $I$ (i.e., $\pi$ is injective on $I$), and (3) destroys the influence of the middle vertex while creating a new influential coordinate outside $\pi(I)$. Formally, we prove the following.

\begin{claim}\label{claim:adversary_fan}
    Suppose that $n \geq 1$ and assume that $4n+1$ is the unique coordinate with linear weight $0$ in $\gen{fan}{n}$. Then $\Inf[(1/2)]{\gen{fan}{n}, 4n+1} \geq \frac{1}{24} - o_n(1)$ and for every $I \subseteq [4n+1]$ with $|I| \leq n-1$, there exists a \kl{minor map} $\pi \colon [4n+1] \to [2n+1]$ such that:
    \begin{enumerate}
        \item $\pi^{-1}(2n+1) = \{4n+1\}$,
        \item $\forall \, i \in [2n] : |\pi^{-1}(i)| = 2$ and $|\pi^{-1}(i) \cap I| \leq 1$,
        \item $\Inf[(1/2)]{\minor{ \left(\gen{fan}{n} \right)}{\pi}, \pi(4n+1)} = 0$, and
        \item $\exists \, i \in [4n] \setminus I : \Inf[(1/2)]{\minor{ \left(\gen{fan}{n} \right)}{\pi}, \pi(i)} \geq \frac{1}{2} - o_n(1)$.
    \end{enumerate}
\end{claim}

\begin{proof}
    \newcommand{\1}[0]{\mathbf{1}}
    Fix $n \geq 1$ and let $f = \gen{fan}{n}$. We partition the set $[4n]$ into $4$ sets: for every $(\circ, \triangle) \in \{+,-\}^2$, by $A(\circ, \triangle)$ we denote the set of $n$ coordinates with linear weights of sign $\circ$, and the single quadratic weight of sign $\triangle$.

    Let $\tuple x \sim \cube{1/2, 4n+1}$. Observe that whenever $\ham{\tuple x_{A(-,+)}} > \ham{\tuple x_{A(-,-)}} > \ham{\tuple x_{A(+, +)}} > \ham{\tuple x_{A(+,-)}}$, we have $f(\tuple x) \neq f(\tuple x \oplus (4n+1))$. Since the probability of such an event is $1/4! - o_n(1)$, we obtain the lower-bound for $\Inf[(1/2)]{f, 4n+1}$.

    Fix a set $I \subseteq [4n+1]$ with $|I| \leq n-1$. There exists a tuple $(i, j, k, \ell) \in A(-, +) \times A(-,-) \times A(+, +) \times A(+,-)$ such that $\{i,j,k,\ell\} \cap I = \emptyset$. We obtain the desired \kl{minor map} $\pi$ by pairing $\pi(i) = \pi(j)$, $\pi(k) = \pi(\ell)$, and pairing coordinates in $A(-,+) \setminus \{i\}$ with $A(+,-) \setminus \{\ell\}$ and $A(+,+) \setminus \{k\}$ with $A(-, -) \setminus \{j \}$, in such a way to avoid pairing two coordinates from $I$ --- this can be done because of the bound $|I| \leq n-1$.

    The resulting function $g = \minor{f}{\pi}$ is a \kl{linear threshold function}, where the weight of $\pi(4n+1)$ is $0$ (and therefore $\pi(4n+1)$ has \kl{influence} $0$), the weight of $\pi(\{i,j\})$ is $-2$, the weight of $\pi(\{k, \ell\})$ is $2+2\varepsilon$, and all of the remaining coordinates have weight $\varepsilon = 1/(2n)$. We finish by observing that given $\tuple y \sim \cube{1/2, 2n+1}$, we have $g(\tuple y) \neq g(\tuple y \oplus \pi(i))$ whenever $y_{\pi(k)} = 0$ and $\ham{\tuple y} \geq 2$, which happens with probability $\frac{1}{2} - o_n(1)$.
\end{proof}

Since this construction works for any fixed $I$, no finite set of coordinates can anticipate which coordinate will become influential after taking the \kl{minor}. Note that $\pi$ is not exactly a \kl{2-to-1 minor} (the arity is odd) --- this can be fixed by adding a single dummy variable to the example, which is then identified with the middle vertex.

This shows that in order to ensure good influence preservation properties within the class of \kl{polynomial threshold functions}, one must explore further restrictions on the allowed \kl{minor maps}.

\subsection{Adding randomness}\label{sec:random-minors}

The further natural restriction is to give up influence preservation across all possible \kl{2-to-1} \kl{minor maps}, requiring only that a \emph{typical} \kl{2-to-1} \kl{minor} preserves \kl{influence}.
Intuitively, this should withstand the possibility of ``bad'' minors since, as \cref{example:increasing-linear-thresholds-with-no-influence-preservation,ptfs:example:fan} illustrate, these adversarial minor maps are very sensitive to changing identification pairs, and should not constitute a large fraction of all \kl{2-to-1} minor maps.

Introducing randomness to the distribution of constraints in the instance of the \kl{Label Cover} problem was first proposed in~\cite{onrich2to1}, where the authors put forward another variant of the Unique Games Conjecture --- the \richconj\ --- and proved that in the imperfect-completeness setting it is in fact equivalent to the UGC \cite[Theorem 8]{onrich2to1}.\footnote{Notably, the proof of this equivalence relies on an influence preservation type result, although in the opposite direction to ours: for so-called ``\kl{low-degree concentrated}'' functions, with high probability over a random \kl{2-to-1} \kl{minor}, every influential coordinate of the \kl{minor} originates from an influential coordinate of the original function \cite[Lemma 23]{onrich2to1}.}
In the perfect-completeness setting, which is exactly the one relevant to deriving hardness results for \PCSP s, it is a strengthening of the $2$-to-$1$ Conjecture.

In the context of \PCSP s, the first application of \richconj\ was by \cite[Theorem 4.8]{brakensiek2021conditional}, where the authors derive hardness of \kl{Ordered} \PCSP s by studying their \kl{polymorphism} \kl{minions}. We abstract out their reduction into the following algebraic condition.

\begin{definition}[\intro{Random 2-to-1 hardness condition}]\AP
    We say that a \kl{minion} $\minion M$ satisfies the \emph{random 2-to-1 condition} if
    there exist $\mathcal{C}, \tau > 0$ and a mapping $\Sel$ that assigns to each $f \in \minion M$ a subset of its coordinates such that
    \begin{enumerate}
        \item $\lvert\Sel(f)\rvert \leq \mathcal{C}$, and
        \item if $f$ has even arity $2n$, then
        \[
            \Pr_{\pi} \left[ \pi(\Sel(f)) \cap \Sel(\minor{f}{\pi}) \neq \emptyset \right] \geq \tau,
        \]
        where $\pi \colon [2n] \to [n]$ is a uniformly random \kl{2-to-1} \kl{minor map}.
    \end{enumerate}
\end{definition}

In this language, the hardness reduction of \cite[Theorem 4.8]{brakensiek2021conditional} can be summarized as the following statement. It can be proved by following the standard reduction framework of \cite{Bible}; we provide the reasoning in \cref{appendix:hardness_reduction} for completeness.

\begin{theorem}[Based on \cite{brakensiek2021conditional}]\label{prelims:theorem:random_2_to_1_condition_implies_hardness}
    Let $(\relstr A, \relstr B)$ be a \PCSP\ \kl(PCSP){template}.
    If $\Pol(\relstr A, \relstr B)$ satisfies the \kl{random 2-to-1 condition}, then $\PCSP(\relstr A, \relstr B)$ is \NP-hard, assuming the \richconj.
\end{theorem}

To prove the hardness result for \kl{Ordered} \PCSP s, the authors of \cite{brakensiek2021conditional} derive an influence preservation result for \kl(func){increasing} functions under the \kl{Shapley distribution}: they show that if a coordinate of an \kl(func){increasing} function has significant \kl{influence} under the \kl{Shapley distribution}, then in expectation over all \kl{2-to-1} \kl{minors}, it retains a constant fraction of its influence.

\begin{theorem}[{\cite[Lemma 4.4]{brakensiek2021conditional}}]\label{influence:theorem:conditional_sv_preservation}
    Let $\delta > 0$. There are $N = N(\delta) > 0$ and $\tau = \tau(\delta) > 0$ such that the following holds. Suppose that $f : \{0,1\}^{2n} \to \{0,1\}$ is an \kl(func){increasing} function with $2n \geq N$. Then
    \[
        \forall \, i \in [2n] : \Inf[\shap]{f, i} \geq \delta \implies \Pr_\pi \Big[ \Inf[\shap]{\minor{f}{\pi}, \pi(i)} \geq \tau \Big] \geq \tau,
    \]
    where $\pi : [2n] \to [n]$ is a uniformly random \kl{2-to-1 minor map}.
\end{theorem}

As \cref{example:increasing-linear-thresholds-with-no-influence-preservation,ptfs:example:fan} illustrate, the conclusion of \cref{influence:theorem:conditional_sv_preservation} is far from trivial. \cite[Theorem 5.1]{brakensiek2021conditional} provides one more adversarial example specific to the setting of \cref{influence:theorem:conditional_sv_preservation}, which witnesses the fact that the probability in the conclusion of \cref{influence:theorem:conditional_sv_preservation} cannot be pushed all the way up to 1.
In other words, the authors argue that richness is necessary even for \kl{increasing} functions.

In the remainder of this section, we take \cref{influence:theorem:conditional_sv_preservation} as a starting point, and aim to generalize the phenomenon of influence preservation over uniformly random \kl{2-to-1} \kl{minors}. The original proof of \cref{influence:theorem:conditional_sv_preservation} in \cite[Lemma 4.4]{brakensiek2021conditional} is combinatorial in nature and relies heavily on monotonicity; by taking an analytical approach instead, we are able to extend it to more general classes of functions and distributions.

In \cref{sec:polynomials}, we return to \kl{Polynomial Threshold PCSPs}, showing that --- despite the adversarial \cref{ptfs:example:fan} given by $\gen{fan}{n}$ --- randomization suffices to recover hardness for all bounded-degree \PTF s (\cref{main:theorem:hardness_for_ptfs_over_cube}). A further application, to \kl(PCSP){Unate} \PCSP s (\cref{main:theorem:hardness_for_unate_over_cube}), is given in \cref{sec:monotonicity}; combined, the two yield \cref{main:theorem:hardness_for_ptfs_and_unate_biased_distribution}.

\subsection{Pull-back and reasonable distributions}

Before stating \Cref{lemma:influence_preservation} formally, we must define relevant notions. Throughout this section, by a \textit{family of distributions} we mean a collection $\Omega = \{ \Omega_n : n \geq 1 \}$ such that $\Omega_n$ is a probability distribution over $\{0,1\}^n$ for every $n$. In particular, $\cube{p}$ and $\shap$ are treated as such families. If $\tuple x \in \{0,1\}^n$, we might write $\Omega(\tuple x)$ instead of $\Omega_n(\tuple x)$ if the dimension is irrelevant or clear from the context. Given $\pi: [n] \to [m]$ and $\tuple y \in \{0,1\}^m$, the \emph{pull-back} of $\tuple y$ through $\pi$, denoted by $\pi^{-1}(\tuple y)$, is the tuple $\tuple x \in \{0,1\}^n$ such that $x_i = y_{\pi(i)}$.

We now present \emph{pull-back distributions}, whose significance in the context of \kl{influence} and random \kl{2-to-1} \kl{minors} was first observed in \cite{onrich2to1}.

\begin{definition}[\intro{Pull-back distribution}]\AP
    Suppose that $\Omega$ is a family of distributions. By $\gpullback{\Omega}$ we denote the \textit{pull-back distribution} of $\Omega$, which is the family of distributions $\{ \gpullback{\Omega}_m : m \geq 1 \}$, defined as follows. For every even dimension $2n$, we have
    \[
        \gpullback{\Omega}_{2n}(\tuple x) = \E[\pi]{\Pr_{\tuple y \sim \Omega_n} \Big[ \pi^{-1}(\tuple y) = \tuple x \Big]},
    \]
    where $\pi : [2n] \to [n]$ is a uniformly random \kl{2-to-1 minor map}. For odd dimensions $m$, we define $\gpullback{\Omega}_m$ to be arbitrary.
\end{definition}

\AP
We will only be interested in \kl[pull-back distribution]{pull-back distributions} of even dimensions. Intuitively, a tuple $\tuple x \sim \gpullback{\Omega}_{2n}$ is obtained by sampling a random \kl{2-to-1 map} $\pi : [2n] \to [n]$ and a tuple $\tuple y \sim \Omega_n$, and then ``pulling back" the tuple $\tuple y$ through $\pi^{-1}$ to output $\pi^{-1}(\tuple y)$. In particular, $\ham{\tuple x}$ must be even whenever $\gpullback{\Omega}(\tuple x) > 0$.

Pull-back distributions will be our main tool for modeling expected coordinate \kl{influence} under random \kl{2-to-1} \kl{minors}. Therefore, we should expect our extension of \cref{influence:theorem:conditional_sv_preservation} to capture only distributions that are somehow \emph{compatible} with their corresponding pull-back counterparts.
Apart from that, we want our distributions to satisfy a few other properties, which we list in the following definition of \emph{reasonable distributions}.

\begin{definition}[\intro{Reasonable distributions}]\AP
    We say that a family of distributions $\Omega$ is \emph{reasonable} if it satisfies all the following properties:
    \begin{enumerate}[leftmargin=*]
        \itemAP (\intro{Symmetry}) For every $n \geq 1$ and $\tuple x, \tuple y \in \Bool^n$, we have
        \[
            \ham{\tuple x} = \ham{\tuple y} \implies \Omega_n(\tuple x) = \Omega_n(\tuple y).
        \]
        When $\Omega$ is symmetric, we write $\Omega_n(k)$ to denote the measure of any tuple $\tuple x \in \{0,1\}^n$ such that $\ham{\tuple x} = k$.
        \itemAP (\intro{Pull-back compatibility}) There exists $C > 0$ such that for every $n \geq 1$ and $\tuple x \in \{0,1\}^{2n}$, we have
        \[
            |\tuple x| \text{ is even } \implies \gpullback{\Omega}(\tuple x) \geq C \cdot \Omega(\tuple x).
        \]
        \itemAP For every $\varepsilon > 0$, there exist a range $[\alpha, \beta] \subseteq [0,1]$, $\lambda \in (0,1)$, and $N \geq 1$ such that for every $n \geq N$:
        \begin{gather*}
            \Omega_n \Big( \big\{ \tuple x : |\tuple x| \in [\alpha n, \beta n] \big\} \Big) \geq 1 - \varepsilon, \\
            \forall \, k \in \{0, 1, \dots, n\} : \Omega_n\Big( \big\{ \tuple x : |\tuple x| = k \big\} \Big) < \varepsilon, \tag{flatness} \\
            \forall \, k \in [\alpha n, \beta n] : \frac{\Omega_n(k \pm 1)}{\Omega_n(k)} \in (\lambda, 1/\lambda), \tag{smoothness} \\
            \forall \, k \in [\alpha n, \beta n] : \frac{\Omega_{n \pm 1}(k)}{\Omega_n(k)} \in (\lambda, 1/\lambda), \tag{consistency}
        \end{gather*}
        where $\pm 1$ means that the statement should hold regardless of the sign. \AP We refer to the last three properties as \intro{flatness}, \intro{smoothness}, and \intro{consistency}, respectively.
    \end{enumerate}
\end{definition}

The notion of \kl{reasonable distributions} is designed to capture both the $p$-\kl{biased distribution} and \kl{Shapley distribution}, although it excludes some distributions of potential interest: \kl{flatness} rules out distributions defined on a \emph{slice} of the Boolean cube \cite{invariance_principle_on_the_slice, multislices}, and \kl{symmetry} rules out product distributions that weight coordinates by semantic role (as with \kl{unate} functions, discussed in \cref{sec:monotonicity}). We leave lifting these restrictions to future work; here, we settle for unifying the field-tested \kl{Shapley distribution} with arguably the most fundamental probability distribution in the study of \kl{Boolean functions}, the $p$-\kl{biased distribution}.

With the notion of \kl{reasonable distributions} in hand, we are in a position to state our main abstract result: the \emph{Influence Preservation Lemma}, which asserts that significant coordinate \kl{influence} is preserved with constant probability under random \kl{2-to-1 minors}, as long as the underlying distribution is \kl{reasonable}, and the considered function's \kl{total influence} is not very large. We defer the proof to \cref{appendix:sec:influence_preservation}.

\begin{restatable}[Influence Preservation Lemma]{lemma}{influencepreservation}\label{lemma:influence_preservation}
    Suppose that $\Omega$ is a \kl{reasonable} family of distributions and $\delta > 0$. There are constants $N, \gamma, \tau > 0$ depending only on $(\Omega, \delta)$, such that the following holds. Suppose that $f \colon \{0,1\}^{2n} \to \{0,1\}$ is a function with arity $2n \geq N$. If $\, \sum_i \Inf[\Omega]{f,i} \leq \gamma \cdot (2n)$, then
    \[
        \forall \, i \in [2n] : \Inf[\Omega]{f, i} \geq \delta \implies \Pr_\pi \Big[ \Inf[\Omega]{\minor{f}{\pi}, \pi(i)} \geq \tau \Big] \geq \tau,
    \]
    where $\pi \colon [2n] \to [n]$ is a uniformly random \kl{2-to-1 minor map}.
\end{restatable}

To obtain hardness results for \kl(PCSP){Boolean} \PCSP s in practice, one may combine \Cref{lemma:influence_preservation} with the \kl{random 2-to-1 condition} according to the following recipe:
\begin{enumerate}
    \item Select a \kl{reasonable} family of distributions $\Omega$ (typically $\cube{p}$ or $\shap$);
    \item Show that every \kl{polymorphism} has small (e.g., sublinear) \kl{total influence} under $\Omega$;
    \item Prove that the number of influential coordinates of any \kl{polymorphism} under $\Omega$ is bounded by a constant.
\end{enumerate}

We emphasize that (3) alone does not directly imply hardness under weaker sources of hardness, such as the \kl{multiple choice hardness condition}: this is because of possible adversarial \kl{minors} that erase all significant \kl{influence}, such as those displayed in \cref{sec:influence-preservation-edversary-examples}.

Once conditions (1)--(3) are met, the \kl{random 2-to-1 hardness condition} is satisfied by selecting, for every \kl{polymorphism}, the set of its influential coordinates; \PCSP\ hardness then follows from \Cref{prelims:theorem:random_2_to_1_condition_implies_hardness}. The two requirements of the \kl{random 2-to-1 hardness condition} are verified separately: condition (3) guarantees that the selected set has constant size, whereas condition (2) lets us invoke \Cref{lemma:influence_preservation}, which ensures that the selection survives a random \kl{2-to-1 minor} with constant probability.
Although \Cref{lemma:influence_preservation} abstracts away much of this framework, verifying conditions (2) and (3) still requires understanding the structural properties of the specific \PCSP\ under consideration.

We illustrate this reasoning in two settings: \kl(PCSP){Polynomial Threshold} \PCSP s, discussed in \cref{sec:polynomials}, and \kl(PCSP){Unate} \PCSP s, discussed in \cref{sec:monotonicity}. Before turning to these applications, we conclude this section by establishing two instantiations of \Cref{lemma:influence_preservation}, corresponding to the $p$-biased and \kl[Shapley distribution]{Shapley distributions}.

\subsection{Biased and Shapley distributions are reasonable}\label{sec:cube_and_shapley}

We now prove that both the $p$-\kl{biased distribution} and the \kl{Shapley distribution} are \kl{reasonable}.

\begin{proposition}\label{influence:proposition:cube_and_shapley_are_reasonable}
    The family of distributions $\cube{p} = \big\{ \cube{p,n} : n \geq 1 \big\}$ is \kl{reasonable} for every $p \in (0,1)$, and the family of distributions $\shap = \big\{ \shap[n] : n \geq 1 \big\}$ is \kl{reasonable}.
\end{proposition}

\begin{proof}[Proof for $\cube{p}$]
    Fix $p \in (0,1)$. \kl{Symmetry} is trivial. \kl{Smoothness} and \kl{consistency} are clearly satisfied with $[\alpha, \beta] = [0,1]$, $\lambda = \min(p, 1-p)$, for every $\varepsilon > 0$. Since $pn$ is the mode of $\cube{p,n}$, to show \kl{flatness}, we want to argue that $\Pr[|\tuple x| = pn]$ approaches $0$ as $n$ grows. We use Stirling's approximation $n! \sim \sqrt{2\pi n}(n/e)^n$:
    \[
        \binom{n }{pn} = \frac{n!}{(pn)! \cdot ((1-p)n)!} \sim \frac{1}{\sqrt{2\pi p(1-p) n}} \cdot \frac{1}{p^{pn} \cdot (1-p)^{(1-p)n}}.
    \]
    Therefore, $\Pr[|\tuple x| = pn] = \binom{n }{pn} \cdot p^{pn} \cdot (1-p)^{(1-p)n} = \mathcal{O}(1/\sqrt{n})$, which approaches $0$.

    It remains to show the \kl{pull-back compatibility}. Fix $n \in \mathbb{N}_+$ and a tuple $\tuple z \in \{0,1\}^{2n}$ with $|\tuple z| = 2k$ for some $k \in \{0, 1, \dots, n\}$. First, we calculate the probability that for a uniformly chosen \kl{2-to-1 map} $\pi : [2n] \to [n]$, the tuple $\tuple z$ is consistent with $\pi$, i.e. that $\tuple z$ is in the image of $\tuple x \mapsto \pi^{-1}(\tuple x)$ (in other words, that $\pi$ pairs ones with ones and zeros with zeros in $\tuple z$). Let $\mathcal{S}$ be the set of all 2-to-1 maps. Every map is consistent with $\binom{n }{k}$ tuples, and hence there are $|\mathcal{S}| \cdot \binom{n }{k}$ pairs $(\pi, \tuple z)$ such that $\tuple z$ is consistent with $\pi$. By symmetry, every tuple is consistent with the same number of maps, equal to $|\mathcal{S}| \cdot \binom{n }{k} / \binom{2n }{2k}$. Dividing by the number of all maps, we obtain:
    \[
        \forall \, \tuple z \in \{0,1\}^{2n} : \Pr_\pi\left[\tuple z \, \text{ is consistent with } \pi\right] = \frac{\binom{n }{k}}{\binom{2n }{2k}}.
    \]
    We obtain the equality $\gpullback{\cube{p}}(\tuple z) = p^k \cdot (1-p)^{n-k} \cdot \binom{n }{k}/\binom{2n }{2k}$. If $k = 0$ or $k = n$, then the lemma statement clearly holds for $C = 1$. Otherwise, we use Stirling's approximation to bound the ratio of binomial coefficients:
    \begin{align*}
        \frac{\binom{n }{k}}{\binom{2n }{2k}} &= \frac{n! \cdot (2k)! \cdot (2n-2k)!}{k! \cdot (n-k)! \cdot (2n)!} \\
        &\sim \frac{\sqrt{4nk(n-k)} \cdot n^n \cdot (2k)^{2k} \cdot (2n-2k)^{2n-2k}}{\sqrt{2nk(n-k)} \cdot k^k \cdot (n-k)^{n-k} \cdot (2n)^{2n}} \\
        &\geq C \cdot \frac{k^k \cdot (n-k)^{n-k}}{n^n},
    \end{align*}
    where $C$ is a universal positive constant. Moving forward, we want to show that
    \[
        \frac{k^k \cdot (n-k)^{n-k}}{n^n} \geq p^k \cdot (1-p)^{n-k}.  \tag{{\color{ao(english)}$1$}}\label{eq:influence:stirling_approximation}
    \]
    To this end, we consider the derivative of $\zeta: p \mapsto p^k \cdot (1-p)^{n-k}$ for $p \in (0,1)$. Simple calculations show that $\zeta'$ is positive in the range $(0, k/n)$, negative in the range $(k/n, 1)$, and equal to $0$ in $k/n$. Therefore, $\zeta$ attains its maximum in the range $(0,1)$ for $p = k/n$, in which case \eqref{eq:influence:stirling_approximation} becomes an equality. Finally, we obtain the desired inequality:
    \[
        \gpullback{\cube{p}}(\tuple z) = p^k \cdot (1-p)^{n-k} \cdot \frac{\binom{n }{k}}{\binom{2n }{2k}} \geq C \cdot  p^{2k} \cdot (1-p)^{2n-2k} = C \cdot \cube{p}(\tuple z). \qedhere
    \]
\end{proof}

\begin{proof}[Proof for $\shap$]
    \kl{Symmetry} and \kl{flatness} are immediate. We now show \kl{smoothness} and \kl{consistency}. Fix any $\varepsilon > 0$ and let $\alpha = \varepsilon/4, \beta= 1-\varepsilon/4$. First, we observe that if $n$ is sufficiently large, we have
    \begin{align*}
        \shap[n] \left( \left\{ \tuple x : |\tuple x| \in [\alpha n, \beta n] \right\} \right) &= 1 - \sum_{k=0}^{\floor{\alpha n}} \frac{1}{n+1} - \sum_{k=\ceil{\beta n}}^{n} \frac{1}{n+1} \\
        &\geq 1 - \frac{\varepsilon n/2 + 2}{n+1} \geq 1 - \varepsilon.
    \end{align*}
    Suppose that $k \in [(\varepsilon/4) n, (1-\varepsilon/4)n]$. To show \kl{smoothness}, we write:
    \begin{gather*}
        \frac{\Omega_n(k+1)}{\Omega_n(k)} = \frac{\binom{n }{k}}{\binom{n }{k+1}} = \frac{k+1}{n-k} = \frac{k}{n-k} + o_n(1), \\
        \frac{\Omega_n(k-1)}{\Omega_n(k)} = \frac{\binom{n }{k}}{\binom{n }{k-1}} = \frac{n-k+1}{k} = \frac{n-k}{k} + o_n(1).
    \end{gather*}
    The function $x \mapsto x/(n-x)$ is increasing, so on the range $[(\varepsilon/4)n, (1-\varepsilon/4)n]$ it attains its maximum $(1-\varepsilon/4)/(\varepsilon/4)$ at $x = (1-\varepsilon/4)n$ and its minimum $(\varepsilon/4)/(1-\varepsilon/4)$ at $x = (\varepsilon/4)n$. We deduce that the ratios above are indeed controlled by a constant depending only on $\varepsilon$. Similarly, for \kl{consistency} we have
    \begin{gather*}
        \frac{\Omega_{n+1}(k)}{\Omega_n(k)} = \frac{n+1}{n+2} \cdot \frac{\binom{n }{k}}{\binom{n+1 }{k}} = \frac{n-k+1}{n+2} = \frac{n-k}{n} \cdot ( 1 + o_n(1) ), \\
        \frac{\Omega_{n-1}(k)}{\Omega_n(k)} = \frac{n+1}{n} \cdot \frac{\binom{n }{k}}{\binom{n-1 }{k}} = \frac{n+1}{n-k} = \frac{n}{n-k} + o_n(1).
    \end{gather*}
    Again, the function $x \mapsto (n-x)/n$ is decreasing, so on the same range it takes values between $\varepsilon/4$ and $1-\varepsilon/4$; hence both ratios above and their reciprocals are controlled by a constant depending only on $\varepsilon$. To summarize, the choice $\lambda = c \cdot (\varepsilon/4)(1-(\varepsilon/4))$ for some universal constant $c$ satisfies the required properties.

    It remains to show \kl{pull-back compatibility}. Fix $n \geq 1$ and a tuple $\tuple z \in \{0,1\}^{2n}$ with $|\tuple z| = 2k$. Following an identical argument as in the case of the \kl{biased distribution} above, we have
    \[
        \gpullback{\shap}(\tuple z) = \shap[n](k) \cdot \frac{\binom{n }{k}}{\binom{2n }{2k}} = \frac{1}{n+1} \cdot \frac{1}{\binom{2n }{2k}} = \frac{2n+1}{n+1} \cdot \shap[2n](\tuple z).
    \]
    Therefore, \kl{pull-back compatibility} is satisfied with $C = 1$.
\end{proof}

As a corollary of \cref{influence:proposition:cube_and_shapley_are_reasonable} and \Cref{lemma:influence_preservation}, we immediately obtain the following two instances for $\cube{p}$ and $\shap$. We note that, by a closer inspection of the proof of \Cref{lemma:influence_preservation}, in the $\cube{p}$ instance the constants $\gamma, \tau, N$ weaken as $p$ approaches $0$ or $1$ --- as is typical of results for $\cube{1/2}$ that extend to $\cube{p}$ bounded away from $0$ and $1$. We therefore replace the dependence on $\Omega$ with $\lambda > 0$ such that $p \in (\lambda, 1-\lambda)$, yielding our influence preservation result for $p$-\kl{biased distributions}: the \emph{Cube Influence Preservation}.

\begin{corollary}[\intro{Cube Influence Preservation}]\label{main:theorem:cube_influence_preservation}\AP
    Suppose that $\lambda \in (0,1/2)$ and $\delta > 0$. There are constants $N = N(\lambda, \delta)$, $\gamma = \gamma(\lambda, \delta) > 0$, and $\tau = \tau(\lambda, \delta) > 0$ such that the following holds. Suppose that $p \in (\lambda, 1-\lambda)$ and $f : \{0,1\}^{2n} \to \{0,1\}$ is a \kl{Boolean function} with $2n \geq N$. If $\, \I[(p)]{f} \leq \gamma \cdot (2n)$, then
    \[
        \forall \, i \in [2n] : \Inf[(p)]{f, i} \geq \delta \implies \Pr_\pi \Big[ \Inf[(p)]{\minor{f}{\pi}, \pi(i)} \geq \tau \Big] \geq \tau,
    \]
    where $\pi : [2n] \to [n]$ is a uniformly random \kl{2-to-1 minor map}.
\end{corollary}

\AP Analogously, since $\shap$ is \kl{reasonable}, \Cref{lemma:influence_preservation} also yields a \intro{Shapley Influence Preservation} corollary, obtained by setting $\Omega \coloneqq \shap$. Since the \kl{total influence} of \kl(func){increasing} functions over $\shap$ is at most $1$, this is a direct extension of \cref{influence:theorem:conditional_sv_preservation} to general \kl{Boolean functions}.

\section{Hardness for Polynomial Threshold PCSPs}\label{sec:polynomials}
In this section, we discuss how our general hardness framework based on \Cref{lemma:influence_preservation} can be applied to \kl(PCSP){Polynomial Threshold} \PCSP s, which we briefly introduced in the~\nameref{sec:intro}. We start by recalling their formal definition.

\begin{definition}[\reintro{Polynomial thresholds}]\AP
    We say that $f\colon \{0,1\}^n \to \{0,1\}$ is a degree-$k$ \reintro{polynomial threshold function} ($\PTF$) if there exists a real multilinear polynomial $Q\colon \{0,1\}^n \to \mathbb{R}$ of degree at most $k$, such that
    \[
        f(\tuple x) = \mathbf 1 {\left[ Q(\tuple x) > 0 \right]}.
    \]
    For every $k \geq 1$, by $\PTF_k$ we denote the class of all degree-$k$ $\PTF$s. Additionally, if $f \in \PTF_1$, then we say that $f$ is a \reintro{linear threshold function}.

    \AP
    \sloppy We say that $\PCSP(\A, \B)$ is a \reintro{Polynomial Threshold PCSP} if $\Pol(\A, \B) \subseteq \PTF_k$ for some $k \geq 1$. If $\Pol(\A, \B) \subseteq \PTF_1$, then we say that $\PCSP(\A, \B)$ is a \reintro{Linear Threshold PCSP}.
\end{definition}

In spite of the adversarial \cref{ptfs:example:fan}, we show that introducing randomness is sufficient to recover influence preservation, not only for quadratic examples such as $\gen{fan}{n}$, but for all bounded-degree $\PTF$s. Namely, we prove the following theorem, along with an analogous hardness result under the \kl{Shapley distribution} (\cref{main:theorem:hardness_for_ptfs_over_shapley}).

\begin{theorem}\label{main:theorem:hardness_for_ptfs_over_cube}
    Suppose that $(\A,\B)$ is a \kl{Polynomial Threshold PCSP} \kl(PCSP){template}. If
    \[
        \exists \, p, \delta > 0 : \forall \, f \in \Pol(\A, \B) : \max_i \Inf[(p)]{f, i} \geq \delta,
    \]
    then $\PCSP(\A, \B)$ is \NP-hard, assuming the \richconj.
\end{theorem}

\AP The proof follows the general recipe outlined in \cref{sec:general}: we must show (1) that the \kl{total influence} of any \kl{polymorphism} of a \kl{Polynomial Threshold PCSP} is sublinear, and (2) that only a bounded number of coordinates can have significant \kl{influence}. Both properties follow from the \intro{low-degree concentration} of $\PTF$s: informally, most of the weight of their \kl{Fourier coefficients} is concentrated on sets of small size. We derive this fact in \cref{appendix:sec:noise_stability_of_ptfs} from an extension of the \kl{noise sensitivity} bounds (see \cref{sec:prelims:fourier}) for \PTF s established by \cite{Harsha2009BoundingTS, bounding_ptfs_2} over $\cube{1/2}$,\footnote{The noise sensitivity bounds of \cite{Harsha2009BoundingTS, bounding_ptfs_2} were later improved in \cite{bounding_ptfs_3, bounding_ptfs_4} for the unbiased distribution $\cube{1/2}$. We are content with the weaker bounds for our purposes.} combined with a proof of the following proposition.

\begin{restatable}[Low-degree concentration of PTFs]{proposition}{ptfslowdegreeconcentration}\label{ptfs:proposition:ptfs_are_low_degree_concentrated}
    Suppose that $\lambda \in (0,1/2), p \in (\lambda, 1-\lambda), \varepsilon > 0$ and $k \geq 1$. There exists a constant $d = d(\lambda, \varepsilon, k)$ such that if $f \in \PTF_k$, then
    \[
        \sum_{|S| > d} \hat{f}^{(p)}(S)^2 \leq \varepsilon.
    \]
\end{restatable}

In the remainder of this section, we prove our hardness results for \kl{Polynomial Threshold PCSPs}, by combining \Cref{lemma:influence_preservation} with the \kl{low-degree concentration} of \kl{polynomial threshold functions}.

\subsection{Hardness over biased distribution}\label{sec:ptfs:cube}
We now show how to utilize the \kl{low-degree concentration} of $\PTF$s to deduce \cref{main:theorem:hardness_for_ptfs_over_cube} --- our hardness result for \kl{Polynomial Threshold PCSPs} over the $p$-\kl{biased distribution}.

Following the general recipe described in \cref{sec:general}, we first show that $\PTF$s have sublinear \kl{total influence} over $\cube{p}$.

\begin{proposition}\label{ptfs:proposition:ptfs_have_sublinear_influence_biased}
    Suppose that $\lambda \in(0,1/2), k \geq 1, p \in (\lambda, 1-\lambda)$ and $f \colon \{0,1\}^n \to \{0,1\}$ is a $\PTF$ of degree $k$. Then
    \[
        \I[(p)]{f} \leq o_{\lambda, k}(n).
    \]
\end{proposition}

\begin{proof}
    Fix $\lambda,k,p$ and $f$ as in the statement. Combining \kl{Parseval identity}, \cref{corollary:influence:total_influence_formula} and \cref{ptfs:proposition:ptfs_are_low_degree_concentrated}, for every $\varepsilon > 0$ there exists $d = d(\lambda, \varepsilon, k)$ such that for sufficiently large $n$,
    \begin{align*}
        \I[(p)]{f} &\leq \frac{1}{\lambda(1-\lambda)} \cdot \sum_{S \subseteq [n]} |S| \cdot \hat{f}(S)^2 \\
        &\leq \frac{d \cdot \norm{f}^2}{\lambda(1-\lambda)} + \frac{1}{\lambda(1-\lambda)} \sum_{|S| > d} |S| \cdot \hat{f}(S)^2 \leq \frac{2\varepsilon \cdot n}{\lambda(1-\lambda)}.
    \end{align*}
    Since $\varepsilon > 0$ can be arbitrarily small, $\I[(p)]{f}$ must be sublinear in terms of $n$.
\end{proof}

With \cref{ptfs:proposition:ptfs_have_sublinear_influence_biased} in hand, we already know that $\PTF$s preserve \kl{influence} through random \kl{2-to-1 minors} by the \kl{Cube Influence Preservation}. However, we still need the last ingredient: we must show that there cannot be arbitrarily many coordinates with influence greater than a set threshold. Conveniently, this fact is also a simple consequence of \kl{low-degree concentration}.

\begin{proposition}\label{ptfs:proposition:ptfs_have_constant_number_of_indices_with_big_influence}
    Suppose that $\lambda \in (0,1/2),  \delta > 0, k \geq 1$ and $p \in (\lambda, 1-\lambda)$. There exists a constant $\mathcal{C} = \mathcal{C}(\lambda, \delta, k)$ such that the following holds. If $f \in \PTF_k$, then
    \[
        \Big| \big\{ i \in [\ar(f)] : \Inf[(p)]{f, i} \geq \delta \big \} \Big| \leq \mathcal{C}.
    \]
\end{proposition}

\begin{proof}
    Fix $\lambda, \delta, k, p$ and $f$ as in the statement, and let $A$ be the set of coordinates $i \in [\ar(f)]$ such that $\Inf[(p)]{f, i} \geq \delta$. \cref{ptfs:proposition:ptfs_are_low_degree_concentrated} implies that there exists $d = d(\lambda,\delta,k)$ such that $\sum_{|S| > d} \hat{f}(S)^2 \leq \lambda(1-\lambda) \cdot\delta/2$. Furthermore, by \cref{proposition:influence:alternative_definitions_of_inf} we have that if $\Inf[(p)]{f, i} \geq \delta$, then $\sum_{S \ni i} \hat{f}(S)^2 \geq \lambda(1-\lambda) \cdot \delta$. Therefore, we have $\sum_{S \ni i, |S| \leq d} \hat{f}(S)^2 \geq \lambda(1-\lambda)(\delta/2)$. Summing over the coordinates in $A$, it follows that
    \[
        |A| \cdot \lambda(1-\lambda)(\delta/2) \leq \sum_{i \in A} \sum_{\substack{S \ni i \\ |S| \leq d}} \hat{f}(S)^2 \leq \sum_{|S| \leq d} |S| \hat{f}(S)^2 \leq d \norm{f}^2 \leq d.
    \]
    As a consequence, $|A|$ is bounded by a constant depending only on $\lambda, \delta$ and $k$.
\end{proof}

Finally, we are able to combine \kl{Cube Influence Preservation} with all the properties of $\PTF$s established in this section to prove that \kl{minions} consisting of $\PTF$s of bounded degree with influential coordinates satisfy the \kl{random 2-to-1 hardness condition}, which together with \cref{prelims:theorem:random_2_to_1_condition_implies_hardness} implies \cref{main:theorem:hardness_for_ptfs_over_cube}.

\begin{proposition}\label{ptfs:proposition:ptfs_minions_satisfy_random_2_to_1_condition}
    Suppose that $p \in (0,1)$, $k \geq 1$ and $\minion{M} \subseteq \PTF_k$ is a \kl{minion}. If
    \[
        \exists \, \delta > 0 : \forall \, f \in \minion M : \max_{i} \, \Inf[(p)]{f, i} \geq \delta,
    \]
    then $\minion{M}$ satisfies the \kl{random 2-to-1 hardness condition}.
\end{proposition}

\begin{proof}
    Fix $p, k,\minion M$ and $\delta$ as in the statement. Let $\gamma, \tau$ and $N$ be the constants obtained from \kl{Cube Influence Preservation} applied to the $p$-\kl{biased distribution} and $\delta$. By \cref{ptfs:proposition:ptfs_have_sublinear_influence_biased} we can assume that $N$ is large enough so that if $\ar(f) \geq N$, then $\I[(p)]{f} \leq \gamma  \cdot \ar(f)$. Moreover, we assume that $\tau \leq \delta$; otherwise we set $\tau := \delta$.

    We define the choice function $\Sel \colon \minion M \to \mathcal{P}(\mathbb{N})$ as
    \[
        \Sel(f) = \begin{cases}
            \big\{ i \in [\ar(f)] \text{ such that } \Inf[(p)]{f, i} \geq \tau \big\} & \text{ if } \ar(f) \geq N, \\
            [\ar(f)] & \text{ if } \ar(f) < N.
         \end{cases}
    \]
    We argue that $\Sel$ satisfies required conditions. First, \cref{ptfs:proposition:ptfs_have_constant_number_of_indices_with_big_influence} implies that there exists a constant $\mathcal{C} = \mathcal{C}(\minion M)$ such that $\lvert\Sel(f)\rvert \leq \mathcal{C}$ for every $f \in \minion M$. It remains to show that $\Sel$ is compatible with random \kl{2-to-1 minors}. Suppose that $f \in \minion M$ has arity $2n$. From assumptions and the fact that $\tau \leq \delta$, we have $\Sel(f) \neq \emptyset$. If $2n < N$, then for every $\pi \colon [2n] \to [n]$ we have $\Sel(\minor{f}{\pi}) = [\ar(\minor{f}{\pi})]$ and clearly $\Pr[\pi(\Sel(f)) \cap \Sel(\minor{f}{\pi}) \neq \emptyset] = 1$.

    In the remaining case, we have $2n \geq N$, which means that we can apply the \kl{Cube Influence Preservation} to $f$. Since $\max_{j} \Inf[(p)]{f, j} \geq \delta$, we can pick a coordinate $i \in [\ar(f)]$ with $\Inf[(p)]{f, i} \geq \delta$; note that $i \in \Sel(f)$ because $\tau \leq \delta$. We deduce that
    \begin{align*}
        \tau \leq \Pr_{\pi: [2n] \to [n]} \bigg[ \Inf[(p)]{\minor{f}{\pi}, \pi(i)} \geq \tau \bigg] &\leq \Pr_{\pi : [2n] \to [n]} \bigg[ \pi(i) \in \Sel(\minor{f}{\pi}) \bigg] \\
             &\leq \Pr_{\pi : [2n] \to [n]} \bigg[ \pi(\Sel(f)) \cap \Sel(\minor{f}{\pi}) \neq \emptyset \bigg].
    \end{align*}
    This means that $\Sel$ together with constants $\mathcal{C}$ and $\tau$ witness the fact that $\minion M$ satisfies the \kl{random 2-to-1 hardness condition}.
\end{proof}

\subsection{Hardness over Shapley distribution}\label{sec:ptfs:shapley}
We now explain how to translate the hardness proof for \kl{Polynomial Threshold PCSPs} from $p$-\kl{biased distributions} to the \kl{Shapley distribution}. The argument leverages \kl{low-degree concentration} across a spectrum of $p$-biased distributions, together with the relationship between \kl{influence} over $\cube{p}$ and $\shap$ established in \cref{shapley:claim:connection_between_shapley_and_biased}.

Using this approach, we obtain the following hardness result, which is an analogue of \cite[Theorem 4.8]{brakensiek2021conditional} for \kl{Polynomial Threshold PCSPs}.

\begin{theorem}\label{main:theorem:hardness_for_ptfs_over_shapley}
    Suppose that $(\A,\B)$ is a \kl{Polynomial Threshold PCSP} \kl(PCSP){template}. If
    \[
        \exists \, \delta > 0 : \forall \, f \in \Pol(\A, \B) : \max_i \Inf[\shap]{f, i} \geq \delta,
    \]
    then $\PCSP(\A, \B)$ is \NP-hard, assuming the \richconj.
\end{theorem}

We show that, under $\shap$, bounded-degree $\PTF$s have sublinear \kl{total influence} and only constantly many influential coordinates. With these two facts in hand, the rest of the proof of \cref{main:theorem:hardness_for_ptfs_over_shapley} is identical to the case of $\cube{p}$ in \cref{main:theorem:hardness_for_ptfs_over_cube}; we omit it for the sake of brevity.

We start with a bound on \kl{total influence}.

\begin{proposition}\label{shapley:proposition:ptfs_have_small_total_shapley}
    Suppose that $k \geq 1$ and $f : \{0,1\}^n \to \{0,1\}$ is a $\PTF$ of degree $k$. Then
    \[
        \I[\shap]{f} \leq o_k(n).
    \]
\end{proposition}

\begin{proof}
    Fix $k$ and $f$ as in the statement. Combining \cref{shapley:claim:connection_between_shapley_and_biased} with \cref{ptfs:proposition:ptfs_have_sublinear_influence_biased} gives that for every $\varepsilon > 0$ and sufficiently large $n$, we have
    \begin{align*}
        \I[\shap]{f} &= \int_0^{\frac{\varepsilon}{4}} \I[(p)]{f} \, dp + \int_{\frac{\varepsilon}{4}}^{1 - \frac{\varepsilon}{4}} \I[(p)]{f} \, dp + \int_{1-\frac{\varepsilon}{4}}^1 \I[(p)]{f} \, dp
                    \\ &\leq \frac{\varepsilon \cdot n}{2} + \int_{\frac{\varepsilon}{4}}^{1 - \frac{\varepsilon}{4}} o_{\varepsilon,k}(n) \, dp \leq \varepsilon \cdot n.
    \end{align*}
    As $\varepsilon > 0$ can be arbitrarily small, the proof is finished.
\end{proof}

Next, we show that only a constant number of coordinates can have significant \kl{influence} over $\shap$ in a $\PTF$ of bounded degree.

\begin{lemma}\label{shapley:lemma:constant_number_of_influential_coordinates_shapley_ptfs}
    Suppose that $\delta > 0$ and $k \geq 1$. There exists a constant $\mathcal{C} = \mathcal{C}(\delta, k)$ such that the following holds. If $f \in \PTF_k$, then
    \[
        \Big| \big\{ i \in [\ar(f)] : \Inf[\shap]{f, i} \geq \delta \big \} \Big| \leq \mathcal{C}.
    \]
\end{lemma}

\begin{proof}
    Fix $\delta, k$, and $f$ as in the statement and let $\lambda = \delta/4$. Let $A$ be the set of coordinates $i \in [\ar(f)]$ such that $\Inf[\shap]{f, i} \geq \delta$. \cref{ptfs:proposition:ptfs_are_low_degree_concentrated} implies that there exists $d = d(\delta, k)$ such that
    \[
        \forall \, p \in (\lambda, 1-\lambda) : \sum_{|S| > d} \hat{f}^{(p)}(S)^2 \leq \lambda(1-\lambda) \cdot (\delta/4).
    \]
    Let $i \in A$. By \cref{proposition:influence:alternative_definitions_of_inf} and the choice of $d$, for every $p \in (\lambda, 1-\lambda)$, we have
    \begin{align*}
        \sum_{\substack{S \ni i \\ |S| \leq d}} \hat{f}^{(p)}(S)^2 &\geq \lambda(1-\lambda) \cdot \Inf[(p)]{f, i} - \sum_{|S| > d} \hat{f}^{(p)}(S)^2 \\
            &\geq \lambda(1-\lambda) \left(\Inf[(p)]{f, i} - \delta/4 \right).
    \end{align*}

    By \cref{shapley:claim:connection_between_shapley_and_biased} we have $\int_{\lambda}^{1-\lambda} \Inf[(p)]{f, i} \, dp \geq \delta/2$. Hence,
    \begin{align*}
        \int_{\lambda}^{1-\lambda} \sum_{\substack{S \ni i \\ |S| \leq d}} \hat{f}^{(p)}(S)^2 \, dp &\geq \lambda(1-\lambda) \left( \int_\lambda^{1-\lambda} \Inf[(p)]{f, i} \, dp - \delta/4 \right),
    \end{align*}
    which is at least a constant $\Delta(\delta) = \lambda(1-\lambda)\cdot(\delta/4) > 0$. Summing over the coordinates in $A$, we obtain the final inequality:
    \begin{align*}
        \Delta \cdot |A| &\leq \int_\lambda^{1-\lambda} \sum_{i \in A} \sum_{\substack{S \ni i \\ |S| \leq d}} \hat{f}^{(p)}(S)^2 \, dp  \\
        &\leq \int_\lambda^{1-\lambda} \sum_{|S| \leq d} |S| \cdot \hat{f}^{(p)}(S)^2 \, dp \leq d. \qedhere
    \end{align*}
\end{proof}

\section{Monotonicity}\label{sec:monotonicity}
Understanding the structure of monotone \kl{polymorphisms} is, undoubtedly, a necessary stepping stone towards classifying arbitrary \kl(PCSP){Boolean} \PCSP s.
The starting point of this section is the conditional dichotomy for the \kl{Ordered} \PCSP s of~\cite{brakensiek2021conditional}.
\AP In algebraic terms, these \PCSP s are exactly those whose \kl{polymorphisms} are monotonically \intro(func){increasing} functions, that is, $(\forall{i}:x_i \le y_i) \implies f(x_1, \dots, x_n) \le f(y_1, \dots, y_n)$.
Note that the \kl{polymorphisms} of any \kl(PCSP){Boolean} \PCSP\ \kl(PCSP){template} can be restricted to \kl(func){increasing} functions by expanding the template with the pair $(\le, \le)$.\footnote{Here $\le$ denotes the binary relation $\{(x,y) \in \Bool^2 \mid x \le y\}$.}
\begin{example}
    A simple example of an \kl{Ordered PCSP} is defined by the following \kl(PCSP){template}
    \begin{align*}
        \relstr T^{\le} &= \left( \Bool; \leq, ~\mathsf{1in3} \right) \\
        \relstr H_2^{\le} &= \left( \Bool; \leq, ~\mathsf{NAE}_3 \right)
    \end{align*}
    By the dichotomy of \cite{brakensiek2021conditional}, it is tractable.
    An easy way to understand why is to recall from the~\nameref{sec:intro} that one sufficient condition for tractability is having infinitely many symmetric \kl{polymorphisms}~\cite[Theorem 3.2]{blp+aip}, which the above \PCSP\ does, for instance $\genthr{1/3}{3n+1}$ for every $n \ge 1$.
\end{example}

Our aim is to first investigate how the original tractability proof for \kl{Ordered PCSPs} takes advantage of the fundamental concept of sharp thresholds, and then to argue that this approach extends beyond this class of \PCSP s.
Namely, the main result of this section is a tractability condition for \kl(PCSP){Unate} \PCSP s; they retain monotone properties in a certain weak sense, and in particular capture all \kl{Ordered PCSPs}.
We illustrate that a \kl{Unate PCSP} can be solved efficiently when its \kl{polymorphism} \kl{minion} exhibits sharp threshold features (the formal definition is deferred to \cref{sec:unate}).

As an intermediate step, we will also present how the dichotomy proof for \kl{Ordered PCSPs} of \cite{brakensiek2021conditional} can be phrased in the language of our new framework.

\subsection{Ordered PCSPs}\label{sec:ordered}
Following the algebraic approach to \PCSP s, we shall first investigate the analytical properties of \kl{polymorphisms} of \kl{Ordered} \PCSP s.
Let $f \colon \Bool^n \to \Bool$ be an \kl(func){increasing} function.
Observe that $\E[p]{f}$, the expectation of $f$ under the $p$-\kl{biased distribution}, increases monotonically from 0 to 1 when $p$ goes from 0 to 1.
Furthermore, unless $f$ is a constant function, the map $[0,1] \ni p \mapsto \E[p]{f}$ is strictly increasing.
\begin{example}
    Consider the \kl{projection} function on $n$ variables $\mathrm{proj}^{(n)}_i(x_1, \dots, x_n) = x_i$.
    Clearly it is an \kl(func){increasing} function, and $\E[p]{\mathrm{proj}^{(n)}_i} = p$. The expectation grows at a steady rate as $p$ increases. The \kl{total influence} is $\I[(p)]{\mathrm{proj}^{(n)}_i} = 1$.
\end{example}
At the same time, there are functions that behave rather differently.
\begin{example}\AP
    Consider the majority function on $2n+1$ variables%
    \[
    \intro*\gen{maj}{n}(x_1, \dots, x_{2n+1}) = \begin{cases}
        1, &\quad\text{if } \sum_i x_i > n \\
        0, &\quad\text{otherwise}
    \end{cases}
    \]
    It is easy to see that $\E[p]{\gen{maj}{n}} = \Pr_{\tuple x \sim \cube{p}}\left[\ham{\tuple x} > n\right] = p$ for $p \in \{0, 1/2, 1\}$.
    Furthermore, it can be shown that, for a fixed $p > 1/2$, the expectation $\E[p]{\gen{maj}{n}} \to 1$ as $n \to \infty$.
    It is often convenient to draw the graph of $\E[p]{f}$ in terms of $p \in [0,1]$.
    For the majority function, it looks more or less as follows:
    \begin{center}
        \vspace{0.5em}
        \includegraphics[width=0.7\linewidth]{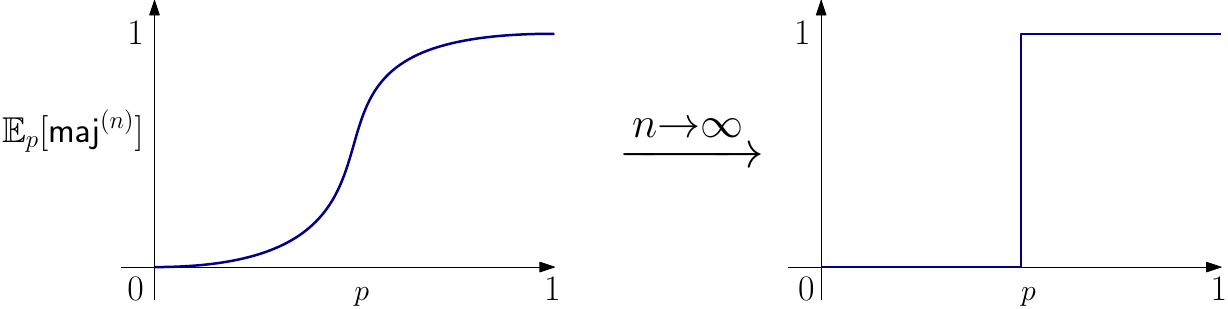}
    \end{center}
    Observe that this curve gets steeper and steeper at $p = 1/2$ as $n$ increases.
    In other words, the derivative
    \[
    \frac{d}{dp} \E[\frac{1}{2}]{\gen{maj}{n}} \to +\infty \quad\text{ as } n \to \infty.
    \]
\end{example}

\AP
These two examples highlight a crucial connection between $\E[p]{f}$ and $\I[(p)]{f}$ for \kl(func){increasing} functions \cite{margulis,russo}:
\begin{equation}
    \qquad\qquad\qquad
    \frac{d}{dp} \, \E[p]{f} = \I[(p)]{f}.
    \tag{\text{\intro{Margulis-Russo formula}}}
\end{equation}
Under the uniform distribution, the \kl{total influence} of $\gen{maj}{n}$ is unbounded --- the function becomes more sensitive as its arity increases, and therefore the graph of its expected value gets steeper.

Informally, we say that a sequence of \kl(func){increasing} functions $\{f_n\}_n$ exhibits a sharp threshold behavior if $\E[p]{f_n}$ goes from almost zero to almost one in ever shorter intervals, as $n \to \infty$. The running example of the majority function illustrates it very well; in general, however, the point $p=1/2$ does not necessarily have to lie within the ``threshold interval''.

Kalai~\cite[Theorem 1.3]{kalai2004shapley} proved that a sharp threshold in a sequence of \kl(func){increasing} functions is equivalent to diminishing maximum \kl{Shapley values} as $n \to \infty$.
The authors of \cite{brakensiek2021conditional} utilized this property for the tractability part of their dichotomy for \kl{Ordered PCSPs}.
\kl{Shapley values} also played a central role in the hardness proof thereof, but in a rather combinatorial way.
Unfortunately, the combinatorial properties of \kl{Shapley values} do not scale well to a more general setting, which motivates investigating them through an analytical lens.

In this subsection, we demonstrate how the proof of the dichotomy result for \kl{Ordered PCSPs} by \cite{brakensiek2021conditional} can be phrased within our framework.
Namely, we make use of our new hardness tool that reasons about the preservation of classical \kl{influence} under random \kl{2-to-1} \kl{minors}, while \cite{brakensiek2021conditional} deployed an ad-hoc preservation of \kl{Shapley values}.
We believe that our version of the proof explains the underlying structure of \kl{Ordered PCSPs} in a more transparent way.

\subsubsection{Tractability}

As declared, we prove that an \kl{Ordered PCSP} is tractable as long as its \kl{polymorphism} \kl{minion} exhibits a sharp threshold.
This follows from the following result and \cref{th:tractablelist}.
\begin{theorem}[Tractable Ordered PCSPs]\label{th:ordered-tractability}
    Let $\minion M$ be a \kl{minion} of Boolean \kl(func){increasing} functions.
    If\hspace{.1em}\footnote{We implicitly use the (1-dimensional) Lebesgue measure throughout this section.}
    \[
    \forall{\, \varepsilon > 0}~\exists{\, f \in \minion M} : \left\lvert \left\{ p : \E[p]{f} \in (\varepsilon, 1-\varepsilon)\right\} \right\rvert < \varepsilon
    \]
    then $\minion M$ contains a constant function or includes one of the minions $\MAX, \MIN, \THR{t}$ for some $t \in (0,1)$.
\end{theorem}
Our proof is a minor extension of the tractability proof in \cite[Lemma 3.4]{brakensiek2021conditional}.
In particular, we demonstrate how to solve the search version of \kl{Ordered PCSPs}, while the authors of \cite{brakensiek2021conditional} were content with the decision version.
\begin{proof}
    Assume from now on that $\minion M$ does not contain a constant function. Thus every $f \in \minion M$ is \kl{idempotent}.

    For each $n \in \mathbb N$, let $f_n$ be the function obtained from the assumption for $\varepsilon = 1/n$. In other words,
    \[
    \left\lvert \left\{ p : \E[p]{f_n} \in \left(\frac{1}{n},\, 1-\frac{1}{n}\right) \right\} \right\rvert < \frac{1}{n}.
    \]
    Let $p_n \in (0,1)$ be the unique value such that $\E[p_n]{f_n} = 1/2$.
    Since the infinite sequence $\{p_n\}_n$ is bounded, it must include a subsequence that converges to some $t \in [0,1]$.
    For convenience, we shall replace the original functions $f_n$ with the corresponding subsequence, and assume that $p_n$ converges to $t$.
    We will argue that
    \[
    \minion M \supseteq \begin{cases}
        \MAX, &\quad\text{if } t = 0 \\
        \THR{t}, &\quad\text{if } t \in (0,1) \\
        \MIN, &\quad\text{if } t = 1
    \end{cases}
    \]

    Fix $m \in \mathbb N$ such that $tm \notin \{1, \dots, m-1\}$. We claim that the $m$-ary $\gen{max}{m}$/$\genthr{t}{m}$/$\gen{min}{m}$ is a \kl{minor} of $f_n$ for some $n$.
    Indeed, pick $n$ large enough so that $p_n$ is sufficiently close to $t$. Let $\pi \colon [\ar(f_n)] \to [m]$ be a random \kl{minor map} where $\pi(i)$'s are independent random variables, uniformly distributed in $[m]$.
    Denote by $g$ the (random) \kl{minor} of $f_n$ given by $\pi$.
    Fix a tuple $\tuple x \in \Bool^m$.
    Observe that the expected value of $g(\tuple x)$ is the same as the expected value of $f(\tuple y)$ where $\tuple y \sim \cube{p}$ for $p = \ham{\tuple x}/m$.
    Consequently, whenever $|\tuple x| < tm$, then $\E{g(\tuple x)} < 1/n$; and whenever $|\tuple x| > tm$, then $\E{g(\tuple x)} > 1-1/n$. Moreover $g$ is \kl{idempotent}. By union bound, the probability that $g$ is not the desired $m$-ary function is at most
    \[
    \frac{2^m}{n} < 1,
    \]
    provided $n$ was chosen sufficiently large.
\end{proof}

\subsubsection{Hardness}\label{sec:monotone-hardness}

We are going to appeal to \cref{prelims:theorem:random_2_to_1_condition_implies_hardness} for the hardness part of the dichotomy.
\begin{theorem}[\NP-hard Ordered \PCSP s]\label{ordered:theorem:hardness_from_cube_influence_preservation}
    Let $\minion M$ be a \kl{minion} of Boolean \kl(func){increasing} functions.
    If
    \begin{equation}\label{eq:ordered-coarce-threshold}
    \exists\ {\varepsilon > 0}~\forall \, {f \in \minion M}~: \left\lvert \left\{ p : \E[p]{f} \in (\varepsilon, 1-\varepsilon)\right\} \right\rvert \ge \varepsilon
    \end{equation}
    then $\minion M$ satisfies the \kl{random 2-to-1 condition}.
\end{theorem}
The goal is to construct a mapping $\Sel$ that satisfies all assumptions of the \kl{random 2-to-1 condition}; the Influence Preservation Lemma (\cref{lemma:influence_preservation}), introduced earlier, suggests that the mapping should select influential coordinates.

We first argue that the lack of sharp threshold~\eqref{eq:ordered-coarce-threshold} ensures that influential coordinates do exist:
\begin{restatable}[Based on {\cite[Theorem 3.5]{kalai2004shapley}}]{lemma}{kalai}\label{ordered:lem:kalai}
    Let $\varepsilon > 0$. There exists $\delta = \delta(\varepsilon) > 0$ such that any \kl(func){increasing} function $f \colon \{0,1\}^n \to \{0,1\}$ satisfying
    \[
    \left\lvert \left\{ p : \E[p]{f} \in (\varepsilon, 1-\varepsilon)\right\} \right\rvert \ge \varepsilon
    \]
    has a coordinate $j \in [n]$ with $\Inf[(p)]{f,j} \ge \delta$ on a set of $p$'s of measure at least $\delta$.
\end{restatable}
\begin{proof}
    Fix $\varepsilon > 0$ and $f$. Let $(L, R)$ be the interval
    \[
    \left\{p \in [0,1] : \E[p]{f} \in (\varepsilon, 1-\varepsilon)\right\} \cap \left(\frac{\varepsilon}{2}, 1-\frac{\varepsilon}{2}\right).
    \]
    Observe that $R - L \ge \varepsilon/2$.
    By the mean value theorem, there exists $p \in (L, R)$ such that the derivative
    \[
    \frac{d}{dp} \E[p]{f} = \frac{\E[R]{f} - \E[L]{f}}{R - L} \le \frac{2}{\varepsilon}.
    \]
    \noindent
    We conclude by the \kl{Margulis-Russo formula} that $\I[(p)]{f} \le 2/\varepsilon$.
    It is known that functions with low \kl{average sensitivity} behave like \emph{juntas}, i.e. functions that depend on few coordinates. Formally, a function $h \colon \Bool^n \to \Bool$ is called a \emph{$J$-junta} for $J \subseteq [n]$ if $h$ only depends on coordinates within $J$.

    Friedgut's Theorem~\cite{friedguttheorem} is a classic result formalizing the connection between low \kl{average sensitivity} and juntas, described above.
    Although his original work only dealt with the uniform distribution, we state it in a more general, $p$-\kl{biased}, setting.
    \newtheorem*{friedgutjunta}{Friedgut's Junta Theorem}
    \begin{friedgutjunta}[See {\cite[Section 10.3]{odonnell}}]
        For every $\lambda, B, \nu > 0$, there exists $M = M(B, \nu, \lambda) > 0$ such that the following holds. Suppose $p \in [\lambda, 1-\lambda]$ and $f \colon \Bool^n \to \Bool$ is a function such that $\I[(p)]{f} \leq B$. Then there exists $J \subseteq [n]$ with $|J| \leq M$ and a $J$-junta $h \colon \Bool^n \to \Bool$ such that $\Pr_{\, \tuple x \sim \cube p \,}[f(\tuple x) \neq h(\tuple x)] \leq \nu$.
    \end{friedgutjunta}
    We apply this theorem to our function $f$ with $B = 2/\varepsilon$ and $\nu = \varepsilon^2/4$, and obtain a junta $h$ that approximates $f$. Note that $M$, the maximum possible number of \kl{essential} coordinates of $h$, is a constant depending only on $\varepsilon$.

    It is expected that, if a function is ``close'' to a junta under some distribution, then it has a coordinate with a significant \kl{influence} under that distribution~\cite{KKL, friedgut1996sharp}.
    However, this is not sufficient for our purposes, as every function $f$ might have its ``own'' $p$ witnessing an influential coordinate; then the \kl{Cube Influence Preservation} cannot be applied.
    To avoid this pitfall, we prove that our functions must have influential coordinates in many \kl{biased} distributions:
    \begin{restatable}{claim}{atwo}\label{claim:a2}
        For any $q \in (L, R)$, there exists a coordinate $j \in J$ such that $\Inf[(q)]{f,j} \ge (\varepsilon/4)^M/M$.
    \end{restatable}
    \noindent We defer the full proof to \cref{appendix:ordered}, and only give the high-level idea here.
    \begin{proof}[Proof idea]
        Monotonicity forces $f$ to output $1$ on almost all tuples where all coordinates in $J$ are fixed to $1$.
        On the other hand, at $q > p$ (assumed without loss of generality) $f$ still outputs $0$ with constant probability.
        Thus, with a constant probability, for $\tuple x \sim \cube q$, $f(\tuple x) = 0$ but setting $x_j = 1$ for all $j \in J$ gives $1$.
        Monotonicity forces one of the $|J|$ coordinates to be influential.
    \end{proof}

    \cref{claim:a2} gives an influential coordinate at every $q \in (L, R)$.
    By the pigeonhole principle, some $j \in J$ satisfies $\Inf[(q)]{f,j} \ge (\varepsilon/4)^M/M$ for at least $(R-L)/|J|\ge \varepsilon/(2M)$-many values of $q$.
    Taking $\delta \coloneqq \min\{(\varepsilon/4)^M/M, \varepsilon/(2M)\}$ concludes the proof.
\end{proof}

The conclusion of \cref{ordered:lem:kalai} indicates how the mapping $\Sel$ should be constructed --- it should distinguish coordinates which are influential at ``many'' biases $p$.

\begin{proof}[Proof of \cref{ordered:theorem:hardness_from_cube_influence_preservation}]
    Fix $\minion M$ and $\varepsilon > 0$ that satisfy~\eqref{eq:ordered-coarce-threshold}.
    Let $\delta = \delta(\varepsilon)$ be the value from~\cref{ordered:lem:kalai}, and put $\lambda \coloneqq \delta/4$.
    Let $\tau = \tau(\lambda, \delta)$ be the value from~\kl{Cube Influence Preservation}; note that $\tau$ may be reduced without compromising the conclusion of~\kl{Cube Influence Preservation}, so we will assume $\tau \le \delta$.

    For each $f \in \minion M$, select the coordinates
    \[
    \Sel(f) = \left\{i : \left\lvert\left\{p \in [0,1] : \Inf[(p)]{f, i} \ge \tau\right\}\right\rvert \ge \lambda\tau \right\}.\footnotemark
    \]%
    \footnotetext{A fully rigorous selection would split into the low- and high-arity cases, like in the proof of~\cref{ptfs:proposition:ptfs_minions_satisfy_random_2_to_1_condition}, selecting \emph{all} coordinates for low-arity functions. For conciseness, we omit this detail in the present proof.}%
    In other words, we choose those coordinates which are influential under ``many'' different $p$-\kl{biased distributions}.

    Immediately by~\cref{ordered:lem:kalai}, the selection is non-empty for all $f \in \minion M$.
    On the other hand, it also cannot include arbitrarily many coordinates:
    each $i \in \Sel(f)$ contributes at least $\lambda\tau^2$ to the sum
    \begin{align*}
    \sum_i \int_0^1 \Inf[(p)]{f,i} dp &= \int_0^1 \sum_i \Inf[(p)]{f,i} dp &&\text{(integrability of $\Inf[(p)]{f,i}$)}\\
    &= \int_0^1 \I[(p)]{f} dp &&\text{(\kl{Margulis-Russo formula})}\\
    &= \int_0^1 \frac{d}{dp} \E[p]{f} dp &&\\
    &= \E[1]{f} - \E[0]{f} = 1.&&
    \end{align*}

    Finally, we argue that $\Sel(f)$ is preserved under random \kl{2-to-1} \kl{minors}, which shall follow from \kl{Cube Influence Preservation}.
    we will demonstrate later (\cref{unate:lemma:total_influence_is_sqrt}) that $\I[(p)]{f} \in \mathcal{O}_\lambda(\sqrt{\ar(f)})$ for all $\lambda \le p \le 1-\lambda$.
    From \kl{Cube Influence Preservation}, we obtain that, for each coordinate $i \in \Sel(f)$ and every $p \in [\lambda, 1-\lambda]$,
    \[
    \Inf[(p)]{f,i} \ge \delta \implies \Pr_{\pi}\left[\Inf[(p)]{\minor{f}{\pi},\pi(i)} \ge \tau \right] \ge \tau.
    \]
    Therefore, for the coordinate $j$ obtained from~\cref{ordered:lem:kalai} applied to $f$, which belongs to $\Sel(f)$,
    \[
    \E[\pi]{ \left\lvert\left\{ p : \Inf[(p)]{\minor{f}{\pi},\pi(j)} \ge \tau \right\}\right\rvert } \ge \left( \delta - 2\lambda \right) \cdot \tau = 2\lambda\tau
    \]
    where we used Fubini's theorem to switch the order of integration over $\pi$ and $p$.
    Finally, a simple averaging argument gives
    \[
    \Pr_\pi\left[ \left\lvert\left\{ p : \Inf[(p)]{\minor{f}{\pi},\pi(j)} \ge \tau \right\}\right\rvert \ge \lambda\tau \right] \ge \lambda\tau
    \]
    which puts $\pi(j)$ in $\Sel(\minor{f}{\pi})$ with a constant probability.

    This concludes the proof, as we verified that the constructed $\Sel(f)$ satisfies all the assumptions of the~\kl{random 2-to-1 hardness condition}.

    We shall briefly explain how to recover the original hardness argument of~\cite{brakensiek2021conditional}.
    By~\cref{ordered:lem:kalai}, every $f \in \minion M$ has a coordinate $j$ whose \kl{Shapley value} is
    \[
    \Inf[\Phi]{f,j} = \int_0^1 \Inf[(p)]{f,j} dp \ge \lambda \tau^2.
    \]
    It is easy to verify that the mapping $\Sel(f)$ that distinguishes coordinates influential under the \kl{Shapley} distribution also satisfies the assumptions of the \kl{random 2-to-1 hardness condition} through \kl{Shapley Influence Preservation}.

\end{proof}

To summarize, every Boolean \kl{minion} of \kl(func){increasing} functions either has a sharp threshold (\cref{th:ordered-tractability}), or does not have it (\cref{ordered:theorem:hardness_from_cube_influence_preservation}), recovering the conditional dichotomy of~\cite{brakensiek2021conditional}.

\subsection{Unate PCSPs}\label{sec:unate}

In the previous subsection, we illustrated how a sharp threshold in a \kl{minion} of \kl(func){increasing} functions leads to tractability.
Sharp thresholds are phenomena that appear in many different areas of mathematics --- they are not exclusive to \kl(func){increasing} functions, so it is natural to expect that ``sharp threshold implies tractability'' holds in more general \kl{minions} that exhibit some relaxed notion of monotonicity.
We propose \emph{unateness} as a strong candidate for this generalization.

\AP
Recall that a function is \reintro{unate} if every its coordinate is \kl{increasing} or \kl{decreasing}.
Note that a coordinate of a \kl{unate} function is \kl{non-essential} if and only if it is both \kl{increasing} and \kl{decreasing}.
A simple example of a \kl{unate} function is $(x, y) \mapsto x(1-y)$.
More generally, \kl{unate} functions include all \kl(func){increasing} functions and all \kl{linear threshold functions}.

It is worth noting, however, that \kl{unate} functions, in general, may have non-\kl{unate} \kl{minors}.
For instance, consider a function $(x_1, x_2, x_3, x_4) \mapsto \max\{x_1(1-x_2), x_3(1-x_4)\}$ which is clearly \kl{unate}.
The \kl{minor} obtained by identifying $x_2 = x_3$ is a ternary function $f(x, y, z) = \max\{x(1-y), y(1-z)\}$.
Observe that $y$ is neither \kl{increasing} nor \kl{decreasing}: $0 = f(0, 0, 0) < f(0, 1, 0) = 1$ and $1 = f(1, 0, 1) > f(1, 1, 1) = 0$.
One notable exception to this are \kl(func){increasing} functions, whose every \kl{minor} is also \kl(func){increasing}.

Our focus in this section is on \PCSP s whose \kl{polymorphisms} are \kl{unate} functions, which we call \reintro{Unate PCSPs}.
Note that they subsume both \kl{Ordered} and \kl(PCSP){Linear Threshold} \PCSP s.

We can structurally capture \kl{Unate PCSPs} in the same way that \kl{Ordered PCSPs} are captured by the binary inequality relation.
Namely, the \kl{polymorphisms} of any \kl(PCSP){Boolean} \PCSP\ \kl(PCSP){template} can be restricted to \kl{unate} functions by expanding the template with the following pair of relations\footnote{In this drawing, every column corresponds to a tuple in a relation.}
\begin{equation*}
    (\relstr U_L, \relstr U_R) =
    \left(\,
    \relation{| B B B B B |}{
      \hline
      0 & 0 & 1 & 1 & 0\\
      0 & 0 & 1 & 1 & 1\\
      0 & 1 & 0 & 1 & 0\\
      0 & 1 & 0 & 1 & 1\\
      \hline
    }\,,\,
    \Bool^4 \setminus \relation{| B |}{
      \hline
      0\\
      1\\
      1\\
      0\\
      \hline
    }\,\right)
\end{equation*}

\begin{lemma}\label{lem:unatestructural}
    All members of a Boolean \kl{minion} $\minion M$ are \kl{unate} if and only if $\minion M \subseteq \Pol(\relstr U_L, \relstr U_R)$.
\end{lemma}
\begin{proof}
    Fix a Boolean \kl{minion} $\minion M$.
    Suppose first that it is a subset of $\Pol(\relstr U_L, \relstr U_R)$.
    Fix any member $f \colon \Bool^n \to \Bool$ of $\minion M$; the goal is to show that it is \kl{unate}.
    We proceed by contradiction.
    Suppose a coordinate, without loss of generality, 1, is neither \kl{increasing} nor \kl{decreasing}.
    Then there exist tuples $\tuple x, \tuple y \in \Bool^{n-1}$ such that $0 = f(0\tuple x) < f(1\tuple x) = 1$ and $1 = f(0\tuple y) > f(1\tuple y) = 0$.
    Observe that $f$ applied to $(0101), (x_1x_1y_1y_1), \dots, (x_nx_ny_ny_n)$ produces
    \begin{equation*}
        f~
        \relation{| E E E E |}{
            0 & $x_1$ & $\dots$ & $x_n$\\
            1 & $x_1$ & $\dots$ & $x_n$\\
            0 & $y_1$ & $\dots$ & $y_n$\\
            1 & $y_1$ & $\dots$ & $y_n$\\
        }
        =
        \relation{| E |}{
            0\\
            1\\
            1\\
            0\\
        }
    \end{equation*}
    which contradicts \kl{compatibility with} the relation pair.

    For the opposite direction, suppose that $\minion M$ consists of \kl{unate} functions. We verify \kl{compatibility with} $(\relstr U_L, \relstr U_R)$.
    Let $f \in \minion M$ of arity $n$.
    Enumerate the tuples in $\relstr U_L$, and fix any $\pi \colon [n] \to [5]$.
    Consider the $4\times n$ matrix whose $i$-th column is the $\pi(i)$-th tuple in $\relstr U_L$ — we want to prove that $f$ applied row-wise does not output $(0,1,1,0)$.
    Equivalently, it suffices to show that
    \[
    f^\pi~
        \relation{| E E E E E |}{
            0 & 0 & 1 & 1 & 0\\
            0 & 0 & 1 & 1 & 1\\
            0 & 1 & 0 & 1 & 0\\
            0 & 1 & 0 & 1 & 1\\
        }
        \neq
        \relation{| E |}{
            0\\
            1\\
            1\\
            0\\
        }
    \]
    Note that $f^\pi \in \minion M$, so it is also \kl{unate}. In particular, its last coordinate is \kl{increasing} or \kl{decreasing}, which concludes the proof.
\end{proof}

\begin{example}
    For a more practical example of a \kl{Unate PCSP}, consider the following \kl(PCSP){template}
    \begin{align*}
        \relstr T^{\neq} &= \left( \Bool; \neq, ~\mathsf{1in3} \right) \\
        \relstr H_2^{\neq} &= \left( \Bool; \neq, ~\mathsf{NAE}_3 \right)
    \end{align*}
    To see why it defines a \kl{Unate PCSP}, let $f \colon \Bool^n \to \Bool$ be an arbitrary \kl{polymorphism} of $(\relstr T^{\neq}, \relstr H_2^{\neq})$, and $i \in [n]$ be its arbitrary coordinate.
    Denote by $\tuple 1_i \in \Bool^n$ the tuple that has zeros on all positions except $i$-th.
    We will argue that whenever $f(\tuple 1_i) = 1$, the coordinate $i$ is \kl{increasing}; and whenever $f(\tuple 1_i) = 0$, it is \kl{decreasing}.

    Suppose first that $f(\tuple 1_i) = 1$. Then, for any $\tuple x \in \Bool^n$ such that $x_i = 0$ and $f(\tuple x) = 1$, we have $(f(\tuple x), f(\tuple 1_i), f(\tuple x \oplus ([n] \setminus i)) \in \mathsf{NAE}_3$, so $f(\tuple x \oplus ([n] \setminus i)) = 0$. Moreover $f(\tuple x \oplus ([n] \setminus i)) \neq f(\tuple x \oplus i)$ as $f \in \Pol(\neq, \neq)$, therefore $f(\tuple x \oplus i) = 1$, as required.
    The argument for the case $f(\tuple 1_i) = 0$ is similar, so we omit it.
\end{example}

On the analytical front, we will argue that this relaxed notion of monotonicity is sufficient to (1) formally define sharp threshold property for \kl{unate} \kl{minions}, and (2) prove that \kl{minions} satisfying that property are tractable.
After that, some ideas towards hardness for \kl{Unate PCSPs} are discussed, although they do not result in a dichotomy.

\subsubsection{Tractability}

For \kl(func){increasing} functions, the expectation $\E[p]{f}$ grows as $p$ increases, which allowed us to define sharp thresholds in the obvious way.
For \kl{unate} functions, however, this is no longer the case.
\begin{example}\AP
    Consider the functions $\gen{an}{n}$ defined as
    \[
    \intro*\gen{an}{n}(x_1, \dots, x_n) = \begin{cases}
        x_1, &\quad\text{if } x_1 = \dots = x_n \\
        1-x_1, &\quad\text{otherwise}
    \end{cases}
    \]
    Under the name of \emph{Almost Negation}, their combinatorial properties were studied in~\cite[Section A.1]{ficak2019symmetric}.
    These functions are in fact \kl{unate}: the first coordinate is \kl{decreasing}, while the rest are \kl{increasing}.
    Observe that
    \[
    \E[p]{\gen{an}{n}} = p^n + (1-p)\cdot \left(1-(1-p)^{n-1}\right).
    \]
    It is easy to verify that this expression is not monotone in $p \in [0,1]$: for instance, its derivative at $p = 1/2$ is $\frac{n}{2^{n-2}}-1 < 0$ for large $n$.
\end{example}
\noindent Therefore, it is not immediately clear how a sharp threshold should be defined for \kl{unate} functions. The first idea would be to simply invert the \kl{decreasing} coordinates, and consider the resulting \kl(func){increasing} function.
Unfortunately, this definition of sharp threshold is insufficient for our purposes; for example, it does not distinguish the original \kl{increasing} coordinates from the new, inverted ones.
We fix this issue in a rather pedestrian manner: by making the ``density'' parameter of this random system two-dimensional. Fix a \kl{unate} function $f \colon \Bool^n \to \Bool$ and numbers $p, q \in [0,1]$. Let $\tuple r \in [0,1]^n$ be the tuple such that $r_i = p$ if $i$ is an \kl{increasing} coordinate, and $r_i = q$ otherwise. We will use $\E[p,q]{f}$ as a shortcut for $\E[\mu_{\tuple r}]{f}$. In other words, $\E[p,q]{f}$ is the expected value of $f(\tuple x)$ when each $x_i$ is sampled independently from $\mu_p$ if $i$ is \kl{increasing}, and from $\mu_q$ if $i$ is \kl{decreasing} (note that a coordinate can be both \kl{increasing} and \kl{decreasing}, in which case any bias may be used).
Similarly, we extend our notation for \kl{influence} to $\I[(p,q)]{f}$ and $\Inf[(p,q)]{f,i}$.

We are ready to state the sharp threshold condition that, combined with \cref{th:tractablelist}, implies \PCSP\ tractability.
\begin{restatable}{theorem}{maintractability}\label{th:tractability}
    Let $\minion M$ be a Boolean \kl{minion} of \kl{unate} functions. If\hspace{.1em}\footnote{We implicitly use the (2-dimensional) Lebesgue measure throughout this section.}
    \begin{equation*}
        \forall{\,\varepsilon > 0}~\exists{\, f \in \minion M} : \left\lvert\left\{(p, q) : \E[p,q]{f} \in (\varepsilon, 1-\varepsilon) \right\}\right\rvert < \varepsilon
    \end{equation*}
    then $\minion M$ contains a constant function or includes one of $\MAX$, $\MIN$, $\THR{t}$ for some $t$, $\AT$, or their \kl{negated} counterparts.
\end{restatable}

The intuition behind this sharp threshold condition is that the graph of $\E[p,q]{f}$ exhibits an abrupt 0-1 transition when $p$ increases and $q$ decreases (see \cref{fig:intro:at_graph}).

\begin{figure}[!htb]
    \centering
    \includegraphics[scale=0.5,trim=0 0.5cm -0.5cm 2.75cm,clip]{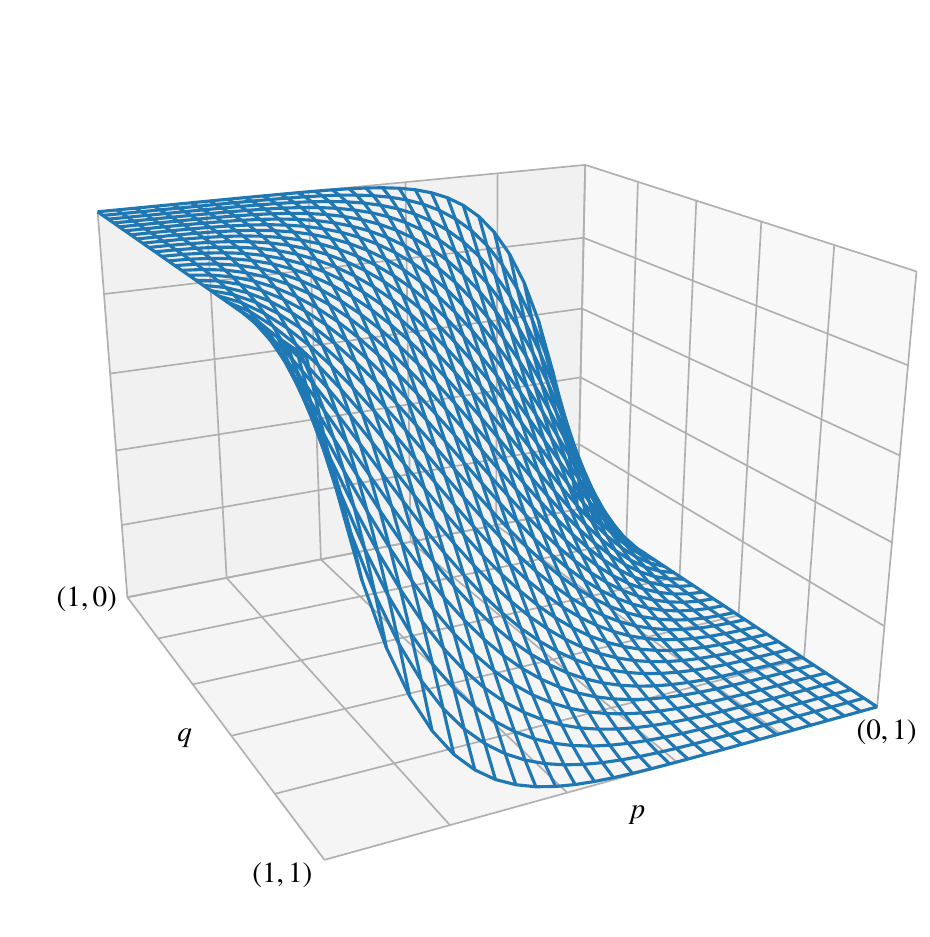}
    \caption{The graph of $\E[p,q]{\gen{at}{3}}$ with respect to $(p,q) \in [0,1]^2$. We can see that the expected value admits a sharp 0-1 transition as $p$ increases and $q$ decreases simultaneously.}\label{fig:intro:at_graph}
\end{figure}

We now proceed to outline the high-level proof strategy; the full proof is deferred to \cref{appendix:unate}.
Before starting the discussion, we need to introduce new notation in order to simplify presentation.
Due to the geometric nature of the proof, the normal font will be used to denote points in $\mathbb R^2$ (e.g. $P, Q$) and the bold font to denote curves (e.g. $\mathbf C$).

Firstly, the set of measure less than $\varepsilon$ in the statement is a connected region in $\mathbb R^2$ enclosed by two curves and the boundary of the unit square $[0,1]^2$.
The assumption states that we can find an infinite sequence of functions in $\minion M$ whose regions converge to a measure zero set.
We will deduce from the fundamental Arzelà-Ascoli theorem~\cite[Theorem 2.5.14]{burago2001metricgeom} that this measure zero set is a \emph{monotonic} curve from one edge of the unit square to the opposite one, call it $\mathbf T$.

Note that $\mathbf T$ splits the unit square into two connected (possibly empty) regions. For a point $P \in [0,1]^2 \setminus \mathbf T$, we will write $P < \mathbf T$ if $P$ lies in the same region as $(0,1)$, and $P > \mathbf T$ if $P$ lies in the same region as $(1,0)$. Note that when $(0, 1) \in \mathbf T$ or $(1, 0) \in \mathbf T$, one of the regions is empty. Nevertheless, it is always the case that either $P \in \mathbf T$, or $P < \mathbf T$, or $P > \mathbf T$ for every $P \in [0,1]^2$.

To summarize, we obtain a sequence of functions for which this monotonic curve acts as a sharp threshold:
\begin{lemma}\label{lem:regions-to-curve}
    Let $\minion M$ be an \kl{idempotent} \kl{minion} that satisfies the assumptions of \cref{th:tractability}.
    There exist a sequence of functions $\{f_n\}_{n \in \mathbb N} \subseteq \minion M$ and a monotonic curve $\mathbf T$ from $(0,0)$ to $(1,1)$ such that%
    \[
    \forall{P \in [0,1]^2\setminus \mathbf T}:~ \lim_{n\to\infty} \E[P]{f_n} = \begin{cases}
        1, &\quad\text{if } P > \mathbf T \\
        0, &\quad\text{if } P < \mathbf T
    \end{cases}
    \]
\end{lemma}

\begin{center}
    \includegraphics[width=0.5\linewidth]{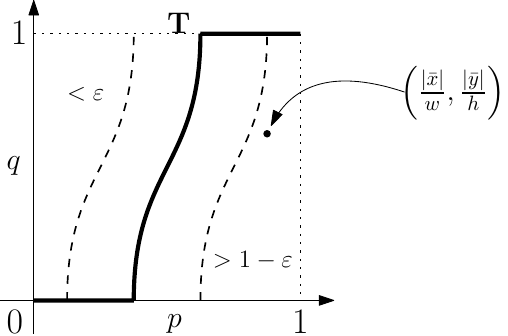}
\end{center}

Next, we employ the probabilistic method. Similarly to how random \kl{minors} produced symmetric functions for \kl{Ordered PCSPs}, we obtain:
\begin{claim}
    For every $w, h \in \mathbb N_+$ there exists a function $g_{w,h} \in \minion M$ of arity $w+h$ that satisfies the following: For any $\tuple x\in\Bool^w, \tuple y \in \Bool^h$ such that $P := (\frac{\ham{\tuple x}}{w}, \frac{\ham{\tuple y}}{h}) \notin \mathbf T$ holds
    \[
    g_{w,h}(\tuple x \tuple y) = \begin{cases}
        1, &\quad\text{if } P > \mathbf T \\
        0, &\quad\text{if } P < \mathbf T
    \end{cases}
    \]
\end{claim}

The properties of \kl{unate} functions ensure that the curve $\mathbf T' = \mathbf T \cap (0,1)^2$ is a straight line segment (otherwise, we obtain an identification \kl{minor} with a coordinate which is neither \kl{increasing} nor \kl{decreasing}).
Finally, depending on the configuration of this line, we show that $\minion M$ must include one of the \kl{minions} on the list:
\begin{itemize}[leftmargin=*]
    \item $\mathbf T'$ above the main diagonal $(0,0) - (1,1)$ $\implies \MAX \subseteq \minion M$;
    \item $\mathbf T'$ below the main diagonal $\implies \MIN \subseteq \minion M$;
    \item $\mathbf T'$ coincides with the main diagonal $\implies \AT \subseteq \minion M$;
    \item $\mathbf T'$ intersects the main diagonal once $\implies \THR{t} \subseteq \minion M$;
\end{itemize}
where we assumed, for simplicity, that $\minion M$ is \kl{idempotent}.

\subsubsection{Towards hardness}\label{sec:unate-hardness}

Consider a Boolean \kl{minion} $\minion M$ of \kl{unate} functions that fails the assumptions of \cref{th:tractability}.
Following the proof of \cref{ordered:theorem:hardness_from_cube_influence_preservation}, it is possible to obtain, for each $f \in \minion M$, a coordinate $j$ such that $\Inf[(p,q)]{f,j} \ge \delta$ for at least a constant fraction of values $(p, q) \in [0,1]^2$.
There are 2 main obstacles to obtaining hardness from this: (1) the analog of \cref{claim:a2} is too weak, and (2) preservation of $\Inf[(p,q)]{f,j}$ under random \kl{2-to-1} \kl{minors} is not well-defined, because the underlying distribution $\mu_{\tuple r}$ depends on the semantics of the function $f$, and not just on its arity.
So far we have not been able to overcome these issues; therefore, the dichotomy for \kl{Unate PCSPs} remains open.

On the other hand, if one can ensure that these influential coordinates appear in $(p,q)$-distributions such that $p = q$, these issues disappear.
\begin{theorem}\label{main:theorem:hardness_for_unate_over_cube}
    Suppose that $(\A, \B)$ is a \kl{Unate PCSP} \kl(PCSP){template}. If
    \[
        \exists \, p,\delta > 0 : \forall \, f \in \Pol(\A, \B) : \max_i \, \Inf[(p)]{f, i} \geq \delta,
    \]
    then $\PCSP(\A, \B)$ is \NP-hard, assuming the \richconj.
\end{theorem}

The proof follows the already familiar pattern: upper-bound the \kl{total influence}, upper-bound the number of influential coordinates, and apply \Cref{lemma:influence_preservation}.

\begin{lemma}\label{unate:lemma:inf_is_scaled_coeff}
    Let $p \in (0, 1)$. If $f \colon \{0,1\}^n \to \{0,1\}$ is \kl{unate}, then
    \[
        \Inf[(p)]{f, i} = \frac{1}{\sqrt{p(1-p)}} \cdot \left\lvert \hat{f}(\{i\}) \right\rvert.
    \]
\end{lemma}
\begin{proof}
    Fix $p$ and $f$ as in the statement. Let $i \in [n]$. We have
    \begin{equation}\label{prelims:eq:influence_for_unate}
        \Inf[(p)]{f, i} = \E[p]{ \left| f(\tuple x^{i \to 1}) - f(\tuple x^{i \to 0}) \right|}
            = \left| \E[p]{f(\tuple x^{i \to 1}) - f(\tuple x^{i \to 0}) } \right|
    \end{equation}
    where the first equality follows from the fact that $f$ has Boolean values, and the second is a consequence of $f$ being \kl{unate}: observe that $f(\tuple x^{i \to 1}) - f(\tuple x^{i \to 0})$ is either always non-negative, or always non-positive, depending on whether $i$ is \kl{increasing} or \kl{decreasing}.
    Next, observe that $\E[p]{\chi_j(\tuple x^{i\to a})} = 0$ for every $a \in \Bool$ and every $j \neq i$. Thus
    \[
        \E[p]{f(\tuple x^{i \to a})} = \sum_{S \subseteq [n]} \hat{f}(S) \cdot \E[p]{\chi_S(\tuple x^{i \to a})} = \hat{f}(\emptyset) + \hat{f}\left(\{ i \} \right) \cdot \chi_i(a).
    \]
    We plug this into \eqref{prelims:eq:influence_for_unate} using the linearity of expectation:
    \[
        \Inf[(p)]{f, i} = \left| \hat{f}(\{i\}) \cdot (\chi_i(1) - \chi_i(0))  \right| = \frac{1}{\sqrt{p(1-p)}} \cdot \left|\hat{f}(\{i\}) \right|.
    \]
\end{proof}

\begin{corollary}\label{unate:lemma:total_influence_is_sqrt}
    Let $p \in (0, 1)$. If $f \colon \{0,1\}^n \to \{0,1\}$ is \kl{unate}, then
    \[
        \I[(p)]{f} \leq \sqrt{\frac{n}{p(1-p)}}
    \]
\end{corollary}
\begin{proof}
    Fix $p$ and $f$ as in the statement.
    We combine the equality given by~\cref{unate:lemma:inf_is_scaled_coeff} with Cauchy-Schwarz and \kl{Parseval identity}:
    \begin{align*}
        \I[(p)]{f} &= \sum_{i \in [n]} \Inf[(p)]{f, i} = \frac{1}{\sqrt{p(1-p)}} \cdot \sum_{i \in [n]} \left\lvert\hat{f}(\{i\})\right\rvert \\
        &\leq \frac{\sqrt{n}}{\sqrt{p(1-p)}} \cdot \sqrt{\sum_{i \in [n]} \hat{f}(\{i\})^2} = \sqrt{\frac{n}{p(1-p)}} \cdot \norm{f} \le \sqrt{\frac{n}{p(1-p)}}.\qedhere
    \end{align*}
\end{proof}

\begin{remark}
    \cref{unate:lemma:total_influence_is_sqrt} can be extended to $\I[(p,q)]{f}$. We do not use it here, though.
\end{remark}

By \cref{unate:lemma:total_influence_is_sqrt}, we immediately deduce that \kl{unate} functions admit influence preservation over $\cube{p}$.
Again, in order to obtain a hardness reduction based on the \kl{random 2-to-1 condition}, the number of influential coordinates in a \kl{unate} function must be upper-bounded by a constant.

\begin{proposition}\label{unate:proposition:unate_have_constant_influential_coordinates_biased}
    Suppose that $\lambda, \delta > 0$ and $p \in (\lambda, 1-\lambda)$. There exists a constant $\mathcal{C} = \mathcal{C}(\lambda, \delta)$ such that the following holds. If $f : \{0,1\}^n \to \{0,1\}$ is \kl{unate}, then
    \[
        \Big| \big\{ i \in [n] : \Inf[(p)]{f, i} \geq \delta \big \} \Big| \leq \mathcal{C}.
    \]
\end{proposition}
\begin{proof}
    Fix $\lambda, \delta, p$ and $f$ as in the statement. Let $A$ be the set of coordinates $i \in [n]$ with \kl{influence} at least $\delta$. Recall that $\Inf[(p)]{f, i} = \left|\hat{f}(\{i\}) \right| / \sqrt{p(1-p)}$ for every $i \in [n]$, by~\cref{unate:lemma:inf_is_scaled_coeff}. Therefore, by \kl{Parseval identity}, we have
    \[
        1 \geq \norm{f}^2 \geq \sum_{i \in A} \hat{f}\big(\{i \}\big)^2 \geq |A| \cdot p(1-p) \delta^2 \geq |A| \cdot (\lambda(1-\lambda)  \delta)^2. \qedhere
    \]
\end{proof}

By combining \cref{unate:lemma:total_influence_is_sqrt}, \cref{unate:proposition:unate_have_constant_influential_coordinates_biased}, and \cref{prelims:theorem:random_2_to_1_condition_implies_hardness}, we obtain the hardness result for \kl{Unate PCSPs} whose \kl{polymorphisms} possess influential coordinates over the $p$-\kl{biased distribution}, that is, \cref{main:theorem:hardness_for_unate_over_cube}.

Is it possible to obtain an analogue of \cref{main:theorem:hardness_for_unate_over_cube} for the \kl{Shapley distribution}? The answer is positive. We follow the same pattern as before. First, we show that \kl{unate} functions have sublinear \kl{total influence} under $\shap$, which yields influence preservation:

\begin{restatable}{proposition}{unatesublineartotalshap}\label{unate:proposition:unate_have_sublinear_total_shapley}
    If $f : \{0,1\}^n \to \{0,1\}$ is \kl{unate}, then $\I[\shap]{f} \leq o(n)$.
\end{restatable}

\begin{proof}
    Fix a \kl{unate} $f$. By \cref{shapley:claim:connection_between_shapley_and_biased} and \cref{unate:lemma:total_influence_is_sqrt}, for every $\varepsilon > 0$ we have
    \begin{align*}
        \I[\shap]{f} &= \int_0^{\varepsilon/4} \I[(p)]{f} \, dp +  \int_{\varepsilon/4}^{1-\varepsilon/4} \I[(p)]{f} \, dp + \int_{1-\varepsilon/4}^1 \I[(p)]{f} \, dp \\
        &\leq (\varepsilon n/2) + \int_{\varepsilon/4}^{1-\varepsilon/4} \mathcal{O}_{\varepsilon}\big(\sqrt{n} \big) \, dp \leq \varepsilon \cdot n,
    \end{align*}
    where the last inequality holds for sufficiently large $n$, using $\I[(p)]{f} \leq n$ on the outer intervals and \cref{unate:lemma:total_influence_is_sqrt} on the middle one. As $\varepsilon > 0$ can be arbitrarily small, we conclude that $\I[\shap]{f} \leq o(n)$.
\end{proof}

Now we have to bound the number of influential coordinates under $\shap$. The methods used so far do not apply here. For \kl{polynomial threshold functions}, we relied on \kl{low-degree concentration}, but \kl{unate} functions --- and even \kl(func){increasing} ones --- are not \kl{low-degree concentrated}\footnote{The fact that even \kl{increasing} functions are not \kl{low-degree concentrated} is commonly known; we refer to \cite{noise_sensitivity_of_boolean_functions_and_percolation} for a comprehensive study of \kl{noise sensitivity} and spectral concentration of \kl(func){increasing} \kl{Boolean functions}.} (\cref{proposition:monotone_not_low_degree_concentrated}). For \kl(func){increasing} functions, the bound followed from the fact that their \kl{total influence} under $\shap$ is at most $1$, whereas \kl{unate} functions can achieve $\Omega(\sqrt{n})$ \kl{total influence} under $\shap$ (\cref{proposition:at_has_large_shapey_value}); we prove both propositions in \cref{appendix:unate}.
Instead, we bound the sum of squares of the \kl{Shapley values} of a \kl{unate} function by combining the representation of \kl{Shapley values} as integrals of \kl{influences} over $p$-\kl{biased distributions} (\cref{shapley:claim:connection_between_shapley_and_biased}) with Minkowski's integral inequality.\footnote{This argument was suggested by ChatGPT (OpenAI), and communicated to us by Antoine Mottet.}

\begin{lemma}\label{unate:lemma:shapley_values_l2_bound}
    If $f \colon \{0,1\}^n \to \{0,1\}$ is \kl{unate}, then
    \[
        \sum_{i \in [n]} \Inf[\shap]{f, i}^2 \leq \frac{\pi^2}{4}.
    \]
\end{lemma}

\begin{proof}
    Fix a \kl{unate} $f$ of arity $n$. For $p \in (0,1)$, let $\tuple v(p) = \left( \Inf[(p)]{f, 1}, \dots, \Inf[(p)]{f, n} \right) \in \mathbb{R}^n$. Every coordinate of $\tuple v(p)$ is a polynomial in $p$, so $\tuple v$ is continuous. By \cref{shapley:claim:connection_between_shapley_and_biased}, the vector of \kl{Shapley values} of $f$ is the coordinatewise integral $\int_0^1 \tuple v(p) \, dp$. Minkowski's integral inequality (equivalently, the triangle inequality for the Bochner integral of $\tuple v$ with respect to the Euclidean norm $\norm{\cdot}_2$ on $\mathbb{R}^n$) gives
    \[
        \left( \sum_{i \in [n]} \Inf[\shap]{f, i}^2 \right)^{1/2} = \norm*{ \int_0^1 \tuple v(p) \, dp }_2 \leq \int_0^1 \norm{\tuple v(p)}_2 \, dp.
    \]
    Recall that $\Inf[(p)]{f, i} = \left|\hat{f}^{(p)}(\{i\}) \right| / \sqrt{p(1-p)}$ for every $p \in (0,1)$ and $i \in [n]$, by~\cref{unate:lemma:inf_is_scaled_coeff}. Since $f$ is Boolean, we have $\norm{f}^2 = \E[p]{f^2} = \E[p]{f}$, and $\hat{f}^{(p)}(\emptyset) = \E[p]{f}$. Hence, by the \kl{Parseval identity},
    \[
        \norm{\tuple v(p)}_2^2 = \frac{1}{p(1-p)} \cdot \sum_{i \in [n]} \hat{f}^{(p)}(\{i\})^2 \leq \frac{\norm{f}^2 - \hat{f}^{(p)}(\emptyset)^2}{p(1-p)} = \frac{\E[p]{f} \cdot \left(1 - \E[p]{f}\right)}{p(1-p)} \leq \frac{1}{4p(1-p)}.
    \]
    We conclude using the standard identity $\int_0^1 \frac{dp}{\sqrt{p(1-p)}} = \pi$:
    \[
        \left( \sum_{i \in [n]} \Inf[\shap]{f, i}^2 \right)^{1/2} \leq \int_0^1 \frac{dp}{2\sqrt{p(1-p)}} = \frac{\pi}{2}. \qedhere
    \]
\end{proof}

\begin{corollary}\label{unate:corollary:unate_have_constant_influential_coordinates_shapley}
    Suppose that $\delta > 0$. If $f \colon \{0,1\}^n \to \{0,1\}$ is \kl{unate}, then
    \[
        \Big| \big\{ i \in [n] : \Inf[\shap]{f, i} \geq \delta \big \} \Big| \leq \frac{\pi^2}{4\delta^2}.
    \]
\end{corollary}

\begin{proof}
    Every coordinate in the set on the left-hand side contributes at least $\delta^2$ to the sum in \cref{unate:lemma:shapley_values_l2_bound}.
\end{proof}

Combining \cref{unate:proposition:unate_have_sublinear_total_shapley} and \cref{unate:corollary:unate_have_constant_influential_coordinates_shapley} with the \kl{Shapley Influence Preservation}, we derive the following analogue of \cref{main:theorem:hardness_for_unate_over_cube} in exactly the same way as our previous hardness results (\cref{main:theorem:hardness_for_ptfs_over_cube,main:theorem:hardness_for_ptfs_over_shapley}).

\begin{theorem}\label{main:theorem:hardness_for_unate_over_shapley}
    Suppose that $(\A, \B)$ is a \kl{Unate PCSP} \kl(PCSP){template}. If
    \[
        \exists \, \delta > 0 : \forall \, f \in \Pol(\A, \B) : \max_i \, \Inf[\shap]{f, i} \geq \delta,
    \]
    then $\PCSP(\A, \B)$ is \NP-hard, assuming the \richconj.
\end{theorem}

\section{Conclusions}\label{sec:conclusions}
We initiated an analytical framework for studying Boolean \kl{polymorphisms}, building on tools from the Fourier analysis of \kl{Boolean functions}. Substantially generalizing \cite[Lemma 4.4]{brakensiek2021conditional}, we proved the Influence Preservation Lemma (\Cref{lemma:influence_preservation}), which yields a family of \PCSP\ hardness reductions from the \kl{Rich 2-to-1 Gap Label Cover} problem, conjectured to be \NP-hard \cite[Conjecture 7]{onrich2to1}, to \kl(PCSP){Unate} and \kl(PCSP){Polynomial Threshold} \PCSP s whose polymorphisms do not admit vanishing maximal coordinate \kl{influences} (\Cref{main:theorem:hardness_for_ptfs_and_unate_biased_distribution}).
Furthermore, we introduced the notion of a sharp threshold for \kl(func){unate} functions and investigated its relation to tractable Unate \PCSP s (\Cref{th:tractability}), directly extending \cite[Lemma 3.4]{brakensiek2021conditional}.

We view our main contribution as the establishment of a general analytic framework for \kl(PCSP){Boolean} \PCSP s, which will hopefully be further developed to amplify the algebraic approach \cite{Bible} and advance the understanding of the largely unexplored Boolean \PCSP\ landscape. Several questions that arose during the investigation of particular subclasses of Boolean \PCSP s remain open.
For \kl(PCSP){Unate} \PCSP s, we gave sufficient conditions for tractability in \cref{th:tractability} and hardness in \cref{main:theorem:hardness_for_unate_over_cube}, but we were unable to prove that these conditions complement each other, leaving the following question open:

\begin{question}
    Do \cref{th:tractability} and \cref{main:theorem:hardness_for_unate_over_cube} fully characterize the complexity of all \kl(PCSP){Unate} \PCSP s?
\end{question}

\noindent An analogous question arises for \kl(PCSP){Polynomial Threshold} \PCSP s:
\begin{question}
    Do all \NP-hard \kl(PCSP){Polynomial Threshold} \PCSP s satisfy the assumptions of \cref{main:theorem:hardness_for_ptfs_over_cube}?
    More generally, does this class of \PCSP s admit a computational dichotomy?
\end{question}

\renewcommand{\arraystretch}{1.5}
\begin{table}[h]
    \centering
    \captionsetup{width=.85\textwidth}
    \caption{A summary of known and new Rich 2-to-1 hardness reductions for Boolean $\PCSP$s based on the preservation of influence over $\cube{p}$ and $\shap$.}
    \begin{tabular}{|c|c|c|c|}
    \hline
    & \kl{Ordered} & \kl(PCSP){Unate} & \kl(PCSP){Polynomial Threshold} \\
    \hline
    $\cube{p}$  & \cref{ordered:theorem:hardness_from_cube_influence_preservation} & \cref{main:theorem:hardness_for_unate_over_cube} & \cref{main:theorem:hardness_for_ptfs_over_cube} \\
    \hline
    $\shap$  & \cite[Theorem 4.8]{brakensiek2021conditional} & \cref{main:theorem:hardness_for_unate_over_shapley} &  \cref{main:theorem:hardness_for_ptfs_over_shapley} \\
    \hline
    \end{tabular}
\end{table}

Another potential direction for further research is to explain the existing results on \kl(PCSP){Boolean} \PCSP s within our analytic framework, unifying various techniques under the umbrella of Fourier analysis.

Finally, we note that the influence preservation result of \cite[Lemma 6.21]{multislices} closely resembles our \cref{lemma:influence_preservation}, but applies only to distributions defined on multi-slice domains (which are not \kl{reasonable}) and relies on a different, ad-hoc notion of \kl{influence} because of that. Does this point towards some more general theory, of which both ours and their result are special cases?

\section*{Acknowledgements}
We thank Marcin Kozik for insightful conversations on analytical approaches to Boolean PCSPs.
We thank Antoine Mottet for communicating the technique behind the bound on the sum of squares of the Shapley values of unate functions (\cref{unate:lemma:shapley_values_l2_bound}).
We are also grateful to the anonymous reviewers for the many comments that allowed us to improve the manuscript.

\section*{AI use statement}
With the exception of~\cref{unate:lemma:shapley_values_l2_bound}, which was discovered by ChatGPT (OpenAI), all ideas, results, and proofs in this paper are solely the work of the authors.
In preparing the final revision, the authors used Claude Code (Anthropic) purely for editorial support: revising the exposition in response to reviewers' comments and checking the consistency of notation and cross-references.
The authors reviewed all AI-assisted edits and take full responsibility for the paper's content.

\printbibliography

\appendix
\crefalias{section}{appendix}

\section{A proof of the Influence Preservation Lemma}\label{appendix:sec:influence_preservation}
In this section, we prove \Cref{lemma:influence_preservation}.

\begin{definition}[\intro{Reasonable parameters}]\AP
    Suppose that $\Omega$ is a \kl{reasonable} family of distributions. For every $\varepsilon > 0$, by $\varepsilon$-\emph{reasonable parameters} we call the tuple $(\alpha_\varepsilon, \beta_\varepsilon, \lambda_\varepsilon, N_\varepsilon)$ satisfying the point (3) (for given $\varepsilon$) of the definition of \kl{reasonable distributions}.
\end{definition}

We start with two simple claims about \kl{reasonable distributions}.

\begin{claim}\label{influence:lemma:reasonable_distributions_restrictions}
    Suppose that $\Omega$ is a \kl{reasonable} family of distributions, $\varepsilon > 0$ and $(\alpha, \beta, \lambda, N)$ are $\varepsilon$-\kl{reasonable parameters} of $\Omega$. Let $n \geq N + 1$, $\tuple x \in \{0,1\}^n$ and $\tuple y \in \{0,1\}^{n-1}$ with $x_i = y_i$ for every $i \in [n-1]$. If $|\tuple x| \in [\alpha n, \beta n]$ or $|\tuple y| \in [\alpha(n-1), \beta(n-1)]$, then
    \[
        \frac{\Omega_{n-1}(\tuple y)}{\Omega_n(\tuple x)} \in \left(\lambda^2, 1/\lambda^2 \right).
    \]
\end{claim}

\begin{proof}
    We consider two cases based on the value of $x_n$. If $x_n = 0$, we have $|\tuple y| = |\tuple x|$ and hence
    \begin{align*}
        \Omega_{n-1}(|\tuple y|) = \Omega_{n-1}(|\tuple x|) \overset{\text{(\kl{consistency})}}{\in} (\lambda, 1/\lambda) \cdot \Omega_n(|\tuple x|).
    \end{align*}
    If $x_n = 1$, we have $|\tuple y| + 1 = |\tuple x|$. We consider two sub-cases. If $|\tuple x| \in [\alpha n, \beta n]$, then $|\tuple y| \in [\alpha n -1, \beta n -1]$. This implies $\{ |\tuple y|, |\tuple y| + 1 \} \cap [\alpha(n-1), \beta(n-1) ] \neq \emptyset$. We obtain
    \begin{align*}
        \Omega_{n-1}(|\tuple y|) &\overset{\text{(\kl{smoothness})}}{\in} (\lambda, 1/\lambda) \cdot \Omega_{n-1}(|\tuple y| + 1) \\
        &\overset{\text{(\kl{consistency})}}{\in} \left(\lambda^2, 1/\lambda^2 \right) \cdot \Omega_n(|\tuple x|).
    \end{align*}
    In the remaining case, we assume that $|\tuple y| \in [\alpha(n-1), \beta(n-1)]$. This implies $|\tuple x| \in [\alpha(n-1) + 1, \beta(n-1) + 1]$, and in turn $\{|\tuple x| - 1, |\tuple x| \} \cap [\alpha n, \beta n] \neq \emptyset$. This yields
    \begin{align*}
        \Omega_{n-1}(|\tuple y|) &\overset{\text{(\kl{consistency})}}{\in} (\lambda, 1/\lambda) \cdot \Omega_{n}(|\tuple x| - 1) \\
        &\overset{\text{(\kl{smoothness})}}{\in} \left(\lambda^2, 1/\lambda^2 \right) \cdot \Omega_n(|\tuple x|). \qedhere
    \end{align*}
\end{proof}

\begin{claim}\label{influence:lemma:reasonable_set_extensions}
    Suppose that $\Omega$ is a \kl{reasonable} family of distributions and $\varepsilon > 0$. There are constants $N = N(\Omega, \varepsilon) > 0$ and $\Delta = \Delta(\Omega, \varepsilon) > 0$ such that for every $n \geq N + 1$ and set $A \subseteq \{0,1\}^{n-1}$, we have
    \begin{enumerate}[leftmargin=*]
        \item $\Omega_{n-1}(A) \leq \Delta \cdot \Omega_n\left( A \times \{0,1\} \right) + \varepsilon$, and
        \item $\Omega_n \left( A \times \{0,1\} \right) \leq \Delta \cdot \Omega_{n-1}(A) + \varepsilon$.
    \end{enumerate}
\end{claim}

\begin{proof}
    \newcommand{\V}[0]{\mathbf{V}}
    \newcommand{\W}[0]{\mathbf{W}}
    Let $(\alpha, \beta, \lambda, N)$ be the $\varepsilon$-\kl{reasonable parameters} of $\Omega$. We distinguish the following sets:
    \begin{gather*}
        \V = \left\{ \tuple x \in \{0,1\}^{n-1} : |\tuple x| \in [\alpha (n-1), \beta (n-1)] \right\} \\
        \W = \left\{ \tuple x \in \{0,1\}^n : |\tuple x| \in [\alpha n, \beta n] \right\}
    \end{gather*}
    Let $B = A \times \{0,1\}$. \cref{influence:lemma:reasonable_distributions_restrictions} and the fact that $\Omega_{n-1}(\V) \geq 1 - \varepsilon$ imply (1):
    \[
        \Omega_{n-1}(A) = \Omega_{n-1} \left( A \cap \V \right) + \Omega_{n-1} \left( A \setminus \V \right) \leq (1/\lambda^2) \cdot \Omega_n(B) + \varepsilon.
    \]
    Similarly, \cref{influence:lemma:reasonable_distributions_restrictions} and $\Omega_n(\W) \geq 1-\varepsilon$ yield (2):
    \[
        \Omega_n(B) \leq \Omega_n\left( B \cap \W \right) + \Omega_n\left( B \setminus \W \right) \leq 2 \cdot (1/\lambda^2) \cdot \Omega_{n-1}(A) + \varepsilon. \qedhere
    \]
\end{proof}

The next three lemmas capture all the technical difficulties of the proof of \Cref{lemma:influence_preservation}; the rest will be to simply combine them.

\begin{lemma}\label{influence:lemma:first_step}
    Suppose that $\Omega$ is a \kl{reasonable} family of distributions and $\delta > 0$. There are constants $\alpha = \alpha(\Omega, \delta) > 0$ and $N = N(\Omega, \delta) > 0$ such that the following holds. Suppose that $f : \{0,1\}^m \to \{0,1\}$ is a \kl{Boolean function} with $m \geq N$ such that $\Inf[\Omega]{f, m-1} \geq \delta$. Let $\pi : [m] \to [m-1]$ be a \kl{minor map} such that $\pi(m-1) = \pi(m) = m-1$ and $\pi(i) = i$ for every $i \in [m-2]$. Then
    \[
        \Inf[\Omega]{\minor{f}{\pi}, m-1} \geq (\alpha \cdot \delta) - \frac{1}{\alpha} \cdot \Inf[\Omega]{f, m}.
    \]
\end{lemma}

\begin{proof}
    \newcommand{\V}[0]{\mathbf{V}}
    \newcommand{\W}[0]{\mathbf{W}}
    \newcommand{\X}[0]{\mathbf{X}}
    \newcommand{\Y}[0]{\mathbf{Y}}
    \newcommand{\Piv}[0]{\mathbf{Piv}}
    \newcommand{\Same}[0]{\mathbf{Same}}
    \newcommand{\Diff}[0]{\mathbf{Diff}}
    Fix $\Omega, \delta, f$ and $\pi$ as in the statement. Furthermore, let $(\alpha_0, \beta_0, \lambda_0, N_0)$ be the $(\delta/16)$-\kl{reasonable parameters} of $\Omega$. We distinguish several subsets of the input space:
    \begin{gather*}
        \V = \left\{ \tuple x \in \{0,1\}^m : |\tuple x| \in [\alpha_0 m, \beta_0 m] \right\} \\
        \W = \left\{ \tuple x \in \{0,1\}^m : |\tuple x| \in [\alpha_0 m + 1, \beta_0 m - 1] \right\} \\
        \forall \, i \in \{m-1, m \} : \Piv(i) = \left\{ \tuple x \in \{0,1\}^m : f(\tuple x) \neq f(\tuple x \oplus i) \right\} \\
        \Same = \left\{ \tuple x \in \{0,1\}^m: x_{m-1} = x_m \right\} \\
        \Diff = \{0,1\}^m \setminus \Same
    \end{gather*}
    Directly from the definition of \kl{reasonable parameters}, we have $\Omega_m(\V) \geq 1-(\delta/8)$. Moreover, we observe that $\W \subseteq \V$ is essentially the set $\V$ without at most two layers of $\{0,1\}^m$. From \kl{flatness} of $\Omega$, we know that these two layers have measure at most $\delta/16$ each. Therefore, we have $\Omega(\V \setminus \W) \leq \delta/8$ and $\Omega_m(\W) \geq 1 - (\delta/4)$.

    Observe that we can bound $\Inf[\Omega]{\minor{f}{\pi}, m-1}$ as follows:
    \begin{align*}
        \Inf[\Omega]{\minor{f}{\pi}, m-1} &\geq \Pr_{\tuple x \sim \Omega_{m-1}} \left[ \pi^{-1}(\tuple x) \in \Piv(m-1) \right] \\
        &- \Pr_{\tuple x \sim \Omega_{m-1}} \left[ (\pi^{-1}(\tuple x) \oplus (m-1)) \in \Piv(m) \right]. \tag{{\color{ao(english)}$1$}}\label{influence:eq:how_to_bound_inf_m_1}
    \end{align*}
    We will bound both terms in \eqref{influence:eq:how_to_bound_inf_m_1} separately. First, we argue that $\Omega_m(\V \cap \Same \cap \Piv(m-1))$ has significant measure. We have $\Omega_m(\Piv(m-1)) = \Inf[\Omega]{f, m-1} \geq \delta$, which implies $\Omega_m(\V \cap \Piv(m-1)) \geq (7/8) \cdot \delta$. Since $\Same \sqcup \Diff = \{0,1\}^m$, we have two cases: either $\Omega_m(\V \cap \Same \cap \Piv(m-1)) \geq (7/16) \cdot \delta$, or $\Omega_m(\V \cap \Diff \cap \Piv(m-1)) \geq (7/16) \cdot \delta$. Assume the latter case, in which we will obtain a weaker bound. Observe that
    \begin{align*}
        \forall \, \tuple x \in \{0,1\}^m &: \tuple x \in \left( \W \cap \Diff \cap \Piv(m-1) \right) \\
        &\implies (\tuple x \oplus (m-1)) \in \left( \V \cap \Same \cap \Piv(m-1) \right).
    \end{align*}
    We combine this observation with \kl{smoothness} of $\Omega$ to bound:
    \begin{align*}
        \Omega_m\left(\V \cap \Same \cap \Piv(m-1) \right) &\geq \lambda_0 \cdot \Omega_m \left( \W \cap \Diff \cap \Piv(m-1) \right) \\
        &\geq \lambda_0 \cdot (5/16) \cdot \delta.
    \end{align*}
    The inequality above combined with \cref{influence:lemma:reasonable_distributions_restrictions} implies the following bound on the first term of \eqref{influence:eq:how_to_bound_inf_m_1}:
    \begin{align*}
        \Pr_{\tuple x \sim \Omega_{m-1}} &\left[ \pi^{-1}(\tuple x) \in \Piv(m-1) \right] \\ &\geq \Omega_{m-1} \left( \left\{ \tuple x : \pi^{-1}(\tuple x) \in \V \cap \Same \cap \Piv(m-1) \right\}\right) \\
        &\geq \lambda_0^2 \cdot \Omega_m \left( \V \cap \Same \cap \Piv(m-1) \right) \\
        &\geq \lambda_0^3 \cdot (5/16) \cdot \delta. \tag{{\color{ao(english)}$2$}}\label{influence:eq:first_term_is_large}
    \end{align*}
    Moving forward, let $(\alpha_1, \beta_1, \lambda_1, N_1)$ be the $(\lambda_0^3 \cdot (1/16) \cdot \delta)$-\kl{reasonable parameters} of $\Omega$. Note that $\lambda_0$ depends only on $\Omega$ and $\delta$, and therefore the same is true for all these parameters. Let $\X \subseteq \{0,1\}^{m-1}$ be the set of points $\tuple x$ with $|\tuple x| \in [\alpha_1 (m-1), \beta_1 (m-1)]$. We have $\Omega_{m-1}\left(\{0,1\}^{m-1} \setminus \X \right) \leq (\lambda_0^3 \cdot (1/16) \cdot \delta)$. Suppose that $\tuple x \in \X$ is equal to $(\tuple y \cdot s)$ for $\tuple y \in \{0,1\}^{m-2}$ and $s \in \{0,1\}$. Then
    \[
        |\pi^{-1}(\tuple x) \oplus (m-1)| = |\pi^{-1}(\tuple y \cdot s) \oplus (m-1)| = |\tuple y \cdot s \cdot (\neg s)| = |\tuple x \cdot (\neg s)|.
    \]
    In particular, \cref{influence:lemma:reasonable_distributions_restrictions} gives $\Omega_{m-1}(\tuple x) \leq (1/\lambda_1^2) \cdot \Omega_m(\pi^{-1}(\tuple x) \oplus (m-1))$. This yields the following bound of the second term of \eqref{influence:eq:how_to_bound_inf_m_1}:
    \begin{align*}
        &\Pr_{\tuple x \sim \Omega_{m-1}} \left[ (\pi^{-1}(\tuple x) \oplus (m-1)) \in \Piv(m) \right] \\
                                         &\leq \, \Omega_{m-1} \left( \left\{ \tuple x \in \X: (\pi^{-1}(\tuple x) \oplus (m-1)) \in \Piv(m) \right\} \right) \\
                                         &\quad + \Omega_{m-1} \left( \{0,1\}^{m-1} \setminus \X \right) \\
                                         &\leq \, (1/\lambda_1^2) \cdot \Omega_m \left( \Piv(m) \right) + (\lambda_0^3 \cdot (1/16) \cdot \delta) \\
                                         &= \, (1/\lambda_1^2) \cdot \Inf[\Omega]{f, m} + (\lambda_0^3 \cdot (1/16) \cdot \delta). \tag{{\color{ao(english)}$3$}}\label{influence:eq:second_term_is_small}
    \end{align*}
    Combining \eqref{influence:eq:how_to_bound_inf_m_1}, \eqref{influence:eq:first_term_is_large} and \eqref{influence:eq:second_term_is_small}, we finally obtain:
    \[
        \Inf[\Omega]{\minor{f}{\pi}, m-1} \geq (\lambda_0^3 / 4) \cdot \delta - (1/\lambda_1^2) \cdot \Inf[\Omega]{f, m}. \qedhere
    \]
\end{proof}

\begin{lemma}\label{influence:lemma:derivative_has_small_total_influence}
    Suppose that $\Omega$ is a \kl{reasonable} family of distributions and $\zeta > 0$. There are constants $\gamma =~\gamma(\Omega, \zeta)$ and $N = N(\Omega, \zeta)$ such that the following holds. Suppose that $f : \{0,1\}^m \to \{0,1\}$ is a \kl{Boolean function} with $m \geq N$ and $\I[\Omega]{f} \leq \gamma \cdot m$. Furthermore, let $\pi : [m] \to [m-1]$ be a \kl{minor map} such that $\pi(m-1) = \pi(m) = m-1$ and $\pi(i) = i$ for every $i \in [m-2]$. Finally, let $h : \Bool^{m-2} \to \Bool$ be defined as
    \[
        h(\tuple x) = \begin{cases}
            1 & \text{if } \minor{f}{\pi}(\tuple x 0) \neq \minor{f}{\pi}(\tuple x 1), \\
            0 & \text{otherwise}.
        \end{cases}
    \]
    Then
    \[
        \I[\Omega]{h} \leq \zeta \cdot (m-2).
    \]
\end{lemma}

\begin{proof}
    \newcommand{\V}[0]{\mathbf{V}}
    \newcommand{\W}[0]{\mathbf{W}}
    \newcommand{\Piv}[0]{\mathbf{Piv}}
    Fix $\Omega$ and $\zeta$ as in the lemma statement. Instead of providing explicit values for $\gamma$ and $N$, we will argue that they can be appropriately chosen during the proof.

    Our first goal is to show that the \kl{total influence} of $\minor{f}{\pi}$ cannot be too large --- that is, we can ensure that $\I[\Omega]{\minor{f}{\pi}} \leq \zeta_0 \cdot (m-1)$ for arbitrarily small $\zeta_0$.

    Let $\zeta_0 > 0$ and $(\alpha_0, \beta_0, \lambda_0, N_0)$ be the $(\zeta_0/2)$-\kl{reasonable parameters} of $\Omega$, and let
    \[
        \V = \{ \tuple z \in \{0,1\}^{m-1} : |\tuple z| \in [\alpha_0(m-1), \beta_0(m-1)] \}.
    \]
    Clearly, we have $\Inf[\Omega]{\minor{f}{\pi}, m-1} \leq 1$. Pick any $i \in [m-2]$. \cref{influence:lemma:reasonable_distributions_restrictions} implies that $\Omega_{m-1}(\tuple z) \leq (1/\lambda_0^2) \cdot \Omega_m(\pi^{-1}(\tuple z))$ for every $\tuple z \in \V$. This yields the following bound:
    \begin{align*}
        &\Inf[\Omega]{\minor{f}{\pi}, i} \\
        &\leq \Omega_{m-1}\left( \{0,1\}^{m-1} \setminus \V \right)  \\
        &\quad + \Omega_{m-1} \left( \V \cap \left\{ \tuple z \in \{0,1\}^{m-1}: f(\pi^{-1}(\tuple z)) \neq f(\pi^{-1}(\tuple z \oplus i)) \right\} \right) \\
        &\leq (\zeta_0/2) + (1/\lambda_0^2) \cdot \Omega_m \left( \left\{ \tuple x \in \{0,1\}^m : f(\tuple x) \neq f(\tuple x \oplus i) \right\} \right) \\
        &\leq (\zeta_0/2) + (1/\lambda_0^2) \cdot \Inf[\Omega]{f, i}.
    \end{align*}
    Summing over all coordinates from $[m-1]$, we obtain the following.
    \begin{align*}
        \I[\Omega]{\minor{f}{\pi}} &\leq 1 + \frac{\zeta_0 \cdot (m-2)}{2} + \frac{\I[\Omega]{f}}{\lambda_0^2} \\
        &\leq (m-1) \cdot \left( \frac{1}{m-1} + \frac{m-2}{m-1} \cdot \frac{\zeta_0}{2} + \frac{m}{m-1} \cdot \frac{\gamma}{\lambda_0^2}\right).
    \end{align*}
    It is now clear that the last factor in the inequality above can be strictly smaller than $\zeta_0$ for appropriately chosen $\gamma = \gamma(\Omega, \zeta_0)$ and $N = N(\Omega, \zeta_0)$.

    The second step is to bound $\I[\Omega]{h}$. Let $(\alpha, \beta, \lambda, N)$ be the $(\zeta/2)$-\kl{reasonable parameters} of $\Omega$ and $\W = \{ \tuple z \in \{0,1\}^{m-2} : |\tuple z| \in [\alpha(m-2), \beta(m-2)] \}$. Fix any coordinate $i \in [m-2]$ and let $\Piv(i) \subseteq \{0,1\}^{m-2}$ be the set of points $\tuple x$ such that $h(\tuple x) \neq h(\tuple x \oplus i)$. Observe that given $\tuple x \in \Piv(i)$, we must have $\minor{f}{\pi}(\tuple x 0) \neq \minor{f}{\pi}(\tuple x 0 \oplus i)$ or $\minor{f}{\pi}(\tuple x 1) \neq \minor{f}{\pi}(\tuple x 1 \oplus i)$, as otherwise we would have:
    \begin{align*}
        h(\tuple x) = 1 &\iff \minor{f}{\pi}(\tuple x 0) \neq \minor{f}{\pi}(\tuple x 1) \\
        &\iff \minor{f}{\pi}(\tuple x 0 \oplus i) \neq \minor{f}{\pi}(\tuple x 1 \oplus i) \iff h(\tuple x \oplus i) = 1,
    \end{align*}
    which is in contradiction with $\tuple x \in \Piv(i)$. This observation along with \cref{influence:lemma:reasonable_distributions_restrictions} allows us to upper-bound $\Inf[\Omega]{h, i}$ in terms of $\Inf[\Omega]{\minor{f}{\pi}, i}$ as follows:
    \begin{align*}
        \Inf[\Omega]{h, i} &\leq \Omega_{m-2} \left( \{0,1\}^{m-2} \setminus \W \right) + \Omega_{m-2} \left( \W \cap \Piv(i) \right) \\
        &\leq \frac{\zeta}{2} + \frac{1}{\lambda^2} \cdot \sum \left\{ \min_{s \in \{0,1\}} \Omega_{m-1}(\tuple x s) : \tuple x \in \W \cap \Piv(i) \right\} \\
        &\leq \frac{\zeta}{2} + \frac{1}{\lambda^2} \cdot \Inf[\Omega]{\minor{f}{\pi}, i}.
    \end{align*}
    Summing over all coordinates in $[m-2]$, we obtain:
    \[
        \I[\Omega]{h} \leq \frac{\zeta \cdot (m-2)}{2} + \frac{\I[\Omega]{\minor{f}{\pi}}}{\lambda^2} \leq (m-2) \cdot \left( \frac{\zeta}{2} + \frac{m-1}{m-2} \cdot \frac{\zeta_0}{\lambda^2} \right).
    \]
    Again, we can see that the last factor can be smaller than $\zeta$ for sufficiently small $\zeta_0$ and sufficiently large $m$, which is ensured with a proper choice of $\gamma$ and $N$, which depends only on $\Omega$ and $\zeta$.
\end{proof}

\begin{lemma}\label{influence:lemma:expectance_in_pullback_does_not_vanish}
    Suppose that $\Omega$ is a \kl{reasonable} family of distributions and $\delta > 0$. There are constants $\gamma =~\gamma(\Omega, \delta) >0$, $\beta = \beta(\Omega, \delta) > 0$ and $N = N(\Omega, \delta) > 0$ such that the following holds. Suppose that $h : \{0,1\}^{2n} \to \{0,1\}$ is a \kl{Boolean function} with $2n \geq N$ such that $\E[\Omega]{h(\tuple x)} \geq \delta$ and $\I[\Omega]{h} \leq \gamma \cdot (2n)$. Then
    \[
        \E[\tuple z \sim \gpullback{\Omega}]{h(\tuple z)} \geq \beta.
    \]
\end{lemma}

\begin{proof}
    \newcommand{\V}[0]{\mathbf{V}}
    \newcommand{\1}[0]{\mathbf{1}}
    \newcommand{\Piv}[0]{\mathbf{Piv}}
    \newcommand{\Even}[0]{\mathbf{Even}}
    \newcommand{\Odd}[0]{\mathbf{Odd}}
    Fix $\Omega, \delta$ as in the lemma statement. Let $(\alpha, \beta', \lambda, N)$ be the $(\delta/4)$-\kl{reasonable parameters} of $\Omega$ and $\V = \{ \tuple x \in \{0,1\}^{2n} : |\tuple x| \in [\alpha \cdot 2n, \beta' \cdot 2n] \}$. Furthermore, let $C > 0$ be the universal constant from \kl{pull-back compatibility} of $\Omega$. We will show that $N$ together with the following constants satisfy the statement conditions:
    \[
        \gamma = \delta/8 \qquad \text{ and } \qquad \beta = C \cdot \lambda \cdot \delta/8.
    \]
    Suppose that $h : \{0,1\}^{2n} \to \{0,1\}$ is as in the statement. Let $\1 = h^{-1}(1)$ and let $\Even \sqcup \Odd = \{0,1\}^{2n}$ be the partition of input space into sets of points with an even and odd number of ones.

    We have $\Omega_{2n}(\1 \cap (\Even \sqcup \Odd)) \geq \delta$. We consider two cases. First, suppose that $\Omega_{2n}(\1 \cap \Even) \geq \delta/2$. In this case, we have
    \begin{align*}
        \E[\tuple z \sim \gpullback{\Omega}]{h(\tuple z)} &\geq \sum_{\tuple z \in \Even} \gpullback{\Omega}(\tuple z) \cdot h(\tuple z) = \gpullback{\Omega}\left( \1 \cap \Even \right) \\
        &\geq C \cdot \Omega\left( \1 \cap \Even \right) \geq C \cdot \delta/2 \geq \beta,
    \end{align*}
    which finishes the proof. In the remaining case, we have $\Omega(\1 \cap \Odd) \geq \delta/2$. Fix any coordinate $i \in [2n]$ such that $\Inf[\Omega]{h, i} \leq \gamma$, which must exist because of the bound on the \kl{total influence} of $h$. Let $\Piv(i) \subseteq~\{0,1\}^{2n}$ be the set of points $\tuple z$ such that $h(\tuple z) \neq h(\tuple z \oplus i)$. Clearly, we have $\Omega(\Piv(i)) = \Inf[\Omega]{h, i} \leq \gamma$.

    Moving forward, observe that $\Omega(\V \cap \1 \cap \Odd) \geq \delta/4$ and that for every $\tuple z \in (\V \cap \1 \cap \Odd) \setminus \Piv(i)$, the point $(\tuple z \oplus i)$ is in $\1 \cap \Even$ and $\Omega(\tuple z \oplus i) \geq \lambda \cdot \Omega(\tuple z)$. Therefore, we have
    \begin{align*}
        \Omega \left(\1 \cap \Even \right) &\geq \sum \left\{ \Omega(\tuple z \oplus i) : \tuple z \in (\V \cap \1 \cap \Odd) \setminus \Piv(i) \right\} \\
                                        &\geq \lambda \cdot \Omega\left( (\V \cap \1 \cap \Odd) \setminus \Piv(i) \right) \\
                                        &\geq \lambda \cdot (\delta/4 - \gamma),
    \end{align*}
    which is at least $\lambda \cdot \delta/8$ by the choice of $\gamma$. By an argument similar to the previous case, we get that $\E[\gpullback{\Omega}]{h} \geq C \cdot \lambda \cdot \delta/8$, which is at least $\beta$.
\end{proof}

With all the technical lemmas in hand, we are ready to combine them into the proof of \Cref{lemma:influence_preservation}.

\influencepreservation*

\begin{proof}
  \newcommand{\X}[0]{\mathbf{X}}
  \newcommand{\V}[0]{\mathbf{V}}
  \newcommand{\W}[0]{\mathbf{W}}
  \newcommand{\Piv}[0]{\mathbf{Piv}}
  Fix $\Omega$ and $\delta$ as in the statement. Instead of providing explicit values for the constants $\gamma, \tau$ and $N$, we will show that they can be chosen appropriately during the proof.

  Let $f$ be as in the theorem statement. Suppose that $i \in [2n]$ is such that $\Inf[\Omega]{f, i} \geq \delta$. Without loss of generality, we assume $i = 2n$. Our goal is to show that a significant portion of $\Inf[\Omega]{f, 2n}$ is preserved through a uniformly random $2$-to-$1$ \kl{minor map} $\pi$ with constant probability, as long as $\gamma$ is sufficiently small and $N$ is sufficiently large.

  To this end, we divide the drawing process of $\pi$ into two steps. First, we randomly choose $j \in [2n-1]$ and define a \kl{minor map} $\pi_0 : [2n] \to [2n-1]$ that identifies $2n$ with $j$ and does not perform any other identification. After that, we uniformly draw a $2$-to-$1$ minor map $\pi_1 : [2n - 1] \to [n]$ that pairs coordinates from $[2n-1] \setminus \{ \pi_0(2n) \}$. It is straightforward to see that the distribution of the composition $\pi_1 \circ \pi_0$ is uniform across all $2$-to-$1$ maps $\pi : [2n] \to [n]$. Therefore, it is sufficient to show that a constant portion of $\Inf[\Omega]{f, 2n}$ is preserved with constant probability after each of these two steps (see \cref{influence:fig:two_steps_minor}).

  \begin{figure}[ht]
    \centering
    \includegraphics[scale=0.65]{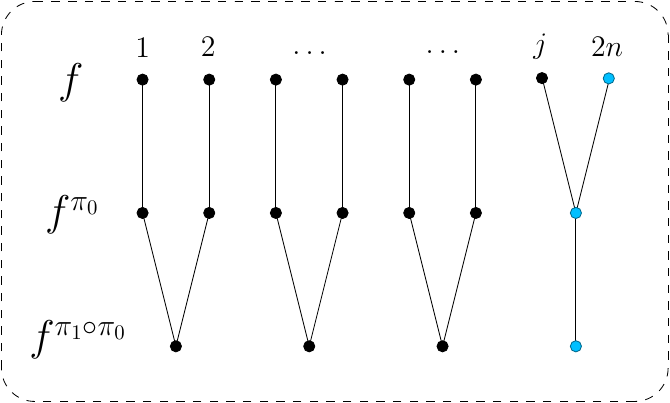}
    \caption{
        We split the random choice of $\pi$ into two steps: (1) choosing the coordinate $j$ with which the influential coordinate $2n$ should be identified, and (2) choosing the pairing of remaining coordinates. We show that after each step, the image of coordinate $2n$ remains influential with constant probability.
    }\label{influence:fig:two_steps_minor}
  \end{figure}

  \paragraph{Step I.} In the first step, we draw $\pi_0: [2n] \to [2n-1]$. Without loss of generality, assume that $\pi_0(2n) = 2n-1$. \cref{influence:lemma:first_step} implies that there exists a constant $\alpha = \alpha(\Omega, \delta) > 0$ such that
  \begin{align*}
    \E[\pi_0]{ \Inf[\Omega]{\minor{f}{\pi_0}, 2n-1} } &\geq (\alpha \cdot \delta) - (1/\alpha)\cdot \E[j \sim [2n-1]]{ \Inf[\Omega]{f, j} } \\
                                                      &\geq (\alpha \cdot \delta) - \frac{1}{\alpha \cdot (2n-1)} \cdot \I[\Omega]{f} \\
                                                      &\geq (\alpha \cdot \delta) - \frac{\gamma \cdot 2n}{\alpha \cdot (2n-1)} \\
                                                      &\geq \delta_0(\Omega, \delta),
  \end{align*}
  where $\delta_0 = \delta_0(\Omega, \delta) > 0$, provided $\gamma$ is sufficiently small and $N$ sufficiently large. Using the fact that every random variable $\X \in [0,1]$ admits $\Pr[\X \geq \E{\X}/2] \geq \E{\X}/2$, we obtain
  \[
    \Pr_{\pi_0} \left[ \Inf[\Omega]{\minor{f}{\pi_0}, 2n-1} \geq \delta_0/2 \right] \geq \delta_0/2. \tag{{\color{ao(english)}$1$}}\label{eq:influence:first_step_identification}
  \]

  \paragraph{Step II.} Fix any $\pi_0$, which witnesses the event in \eqref{eq:influence:first_step_identification}. Without loss of generality, we assume that $\pi_0$ identifies $2n$ with $2n-1$. We now proceed to the second step: choosing uniformly at random a \kl{minor map} $\pi_1 : [2n-1] \to [n]$, which pairs coordinates $\{1, \dots, 2n-2 \}$. We can assume that $\pi_1(2n-1) = n$. Let $\pi = \pi_1 \circ \pi_0$. We want to understand the expected \kl{influence} of coordinate $n$ in $\minor{f}{\pi}$. To this end, we analyze the function $h : \Bool^{2n-2} \to \Bool$ defined as $h(\tuple y) = 1 \iff \minor{f}{\pi_0}(\tuple y 0) \neq \minor{f}{\pi_0}(\tuple y 1)$. Let $\Piv(2n-1)$ be the set of points $\tuple x$ such that $\minor{f}{\pi_0}(\tuple x) \neq \minor{f}{\pi_0}(\tuple x \oplus (2n-1))$. Clearly, we have $\Omega_{2n-1}(\Piv(2n-1)) = \Inf[\Omega]{\minor{f}{\pi_0}, 2n-1} \geq \delta_0/2$. By \cref{influence:lemma:reasonable_set_extensions} applied with $\varepsilon_0 := \delta_0/4$, there is $\Delta_0 = \Delta_0(\Omega, \delta)$ such that
  \begin{align*}
    \E[\Omega]{h} &= \Omega_{2n-2} \left( \left\{ \tuple y : \minor{f}{\pi_0}(\tuple y  0) \neq \minor{f}{\pi_0}(\tuple y 1) \right\} \right) \\
    &\geq \frac{1}{\Delta_0} \cdot \left( \Omega_{2n-1} \left( \Piv(2n-1) \right) - \varepsilon_0 \right) \geq \delta_1(\Omega, \delta).
  \end{align*}
  This bound allows us to apply \cref{influence:lemma:expectance_in_pullback_does_not_vanish} to $h$, as long as $\I[\Omega]{h} \leq \zeta \cdot (2n-2)$ for some $\zeta = \zeta(\Omega, \delta)$. However, this can be ensured by choosing appropriate $\gamma$ and $N$, due to \cref{influence:lemma:derivative_has_small_total_influence}. As a consequence of this argument, we obtain a constant $\beta = \beta(\Omega, \delta) > 0$ such that $\E[\gpullback{\Omega}]{h} \geq \beta$.

  Let $\Piv(n)$ be the set of points $\tuple x \in \{0,1\}^n$ such that $\minor{f}{\pi}(\tuple x) \neq \minor{f}{\pi}(\tuple x \oplus n)$. Recall that $\pi_1$ is a uniformly random $2$-to-$1$ map $[2n-2] \to [n-1]$ and $\pi = \pi_1 \circ \pi_0$. Let $\Delta_1 = \Delta_1(\Omega, \delta)$ be the constant from \cref{influence:lemma:reasonable_set_extensions} for $\varepsilon_1 := \beta/2$. Directly from the definition of \kl{pull-back distribution}, we obtain:
  \begin{align*}
    \beta \leq \E[\tuple z \sim \gpullback{\Omega}]{ h(\tuple z) } &= \E[\pi_1]{ \E[\tuple z \sim \Omega_{n-1}]{ h\left( \pi_1^{-1}(\tuple z) \right) }  } \\
                                                                              &= \E[\pi_1]{ \Pr_{\tuple z \sim \Omega_{n-1}} \left[ \minor{f}{\pi_0}\left( \pi_1^{-1}(\tuple z) 0 \right) \neq \minor{f}{\pi_0} \left( \pi_1^{-1}(\tuple z)  1 \right) \right] } \\
                                                                              &= \E[\pi_1]{ \Pr_{\tuple z \sim \Omega_{n-1}} \left[ \minor{f}{\pi} ( \tuple z  0 ) \neq \minor{f}{\pi} ( \tuple z  1 ) \right] } \\
                                                                              &= \E[\pi_1]{ \Omega_{n-1} \left( \left\{ \tuple z : \minor{f}{\pi} ( \tuple z 0 ) \neq \minor{f}{\pi} ( \tuple z 1 ) \right\} \right) } \\
                                                                              &\leq \Delta_1 \cdot \E[\pi_1]{ \Omega_n \left( \Piv(n) \right) } + \varepsilon_1 \\
                                                                              &= \Delta_1 \cdot \E[\pi_1]{ \Inf[\Omega]{\minor{f}{\pi}, n} } + \varepsilon_1,
  \end{align*}
  and by the choice of $\varepsilon_1$, we get
  \[
      \E[\pi_1]{ \Inf[\Omega]{\minor{f}{\pi}, n} } \geq (1/\Delta_1) \cdot (\beta/2) \geq \delta_2(\Omega, \delta).
  \]
  Using $\Pr[\X \geq \E{\X}/2] \geq \E{\X}/2$ one more time and plugging $n = \pi(2n)$, we obtain
  \[
    \Pr_{\pi_1} \left[ \Inf[\Omega]{\minor{f}{\pi}, \pi(2n)} \geq \delta_2/2 \right] \geq \delta_2/2. \tag{{\color{ao(english)}$2$}}\label{eq:influence:second_step_identification}
  \]

  \paragraph{Conclusion.} Combining \eqref{eq:influence:first_step_identification} and \eqref{eq:influence:second_step_identification}, we obtain the final estimate:
  \[
    \Pr_{\pi} \left[ \Inf[\Omega]{\minor{f}{\pi}, \pi(2n)} \geq \delta_2/2 \right] \geq (\delta_0/2) \cdot (\delta_2/2),
  \]
  which clearly finishes the proof by the choice $\tau := (\delta_0 \cdot \delta_2)/4$.
\end{proof}

\section{Polynomial threshold functions over biased distributions are low-degree concentrated}\label{appendix:sec:noise_stability_of_ptfs}
In this section, we prove a \kl{low-degree concentration} result for bounded-degree $\PTF$s over the $p$-\kl{biased distribution}, which is an extension of an already known result over $\cube{1/2}$.

\ptfslowdegreeconcentration*

\cref{ptfs:proposition:ptfs_are_low_degree_concentrated} is a consequence of a bound on the \kl{noise sensitivity} of $\PTF$s over $\cube{p}$. We will prove the following result, which is an extension of \cite[Theorem 1.3]{Harsha2009BoundingTS} and \cite[Theorem 1.3]{bounding_ptfs_2} from $\cube{1/2}$ to the general $p$-\kl{biased distributions}.

\begin{restatable}[Extension of {\cite[Theorem 1.3]{Harsha2009BoundingTS}}]{theorem}{noisesensitivitymaintheorem}\label{noise_sensitivity:theorem:main_theorem}
  Suppose that $\lambda \in (0,1/2), p \in (\lambda, 1-\lambda)$ and $f \in \PTF_k$. Then
  \[
    \NS[(p)]{f, 1-\delta} \leq \mathcal{O}_{\lambda, k} \left( \delta^{1/(8k+2)} \right).
  \]
\end{restatable}

We note that our proof of \cref{noise_sensitivity:theorem:main_theorem} is almost identical to the original proof of \cite[Theorem 1.3]{Harsha2009BoundingTS}, with some minor adjustments arising from the more general setting; it is a standard pattern in Fourier analysis, that results obtained for the unbiased cube $\cube{1/2}$ usually extend to $\cube{p}$ as long as $p$ is bounded away from $0$ and $1$.

Before we proceed to the proof of the \kl{noise sensitivity} bound, we show how to derive \kl{low-degree concentration} from it. We use the standard connection between \kl{noise sensitivity}, \kl{noise operator} and \kl{Fourier coefficients} (see e.g. \cite[Proposition 3.3]{odonnell}).

\begin{proof}[\cref{noise_sensitivity:theorem:main_theorem} $\implies$ \cref{ptfs:proposition:ptfs_are_low_degree_concentrated}]
    Fix $\lambda, p, \varepsilon$, and $f$ as in the statement. Observe that \cref{noise_sensitivity:theorem:main_theorem} implies that $\NS[(p)]{f,1-\delta}$ approaches $0$ as $\delta \to 0$. Therefore, we can choose $\delta = \delta(\lambda, \varepsilon, k) > 0$ such that $\NS[(p)]{f,1-\delta} \leq \varepsilon$. By \cref{noise:claim:connection_between_noise_and_coefficients} it follows that $\norm{f}^2 - \ip{f}{\TN[(p)]{f,1-\delta}} \leq \varepsilon/2$. Using \cref{proposition:noise:coefficients_of_noise_operator} and the \kl{Parseval identity}, we obtain that
    \begin{align*}
        \varepsilon/2 &\geq \norm{f}^2 - \ip{f}{\TN[(p)]{f,1-\delta}} \\
        &= \sum_{S \subseteq [n]} \hat{f}(S)^2 - \sum_{S \subseteq [n]} (1-\delta)^{|S|} \cdot \hat{f}(S)^2 \\
        &= \sum_{S \subseteq [n]} \left(1-(1-\delta)^{|S|}\right) \cdot \hat{f}(S)^2.
    \end{align*}
    Let $d = d(\lambda, \varepsilon, k)$ be such that $(1-(1-\delta)^d) \ge 1/2$. We can check that $d$ satisfies the required condition:
    \[
        \sum_{|S| > d} \hat{f}(S)^2 \leq \frac{\varepsilon/2}{\left(1-(1-\delta)^d\right)} \leq \varepsilon. \qedhere
    \]
\end{proof}

\subsection{Towards a proof of the noise sensitivity bound}

\AP
The rest of this section is devoted to proving \cref{noise_sensitivity:theorem:main_theorem}. We start by recalling notions introduced in \cite{Harsha2009BoundingTS} and establishing the setting. We define the sign function $\sgn : \mathbb{R} \to \{0,1\}$ as $\sgn(x) = 1$ if and only if $x > 0$. Given a multilinear polynomial $Q$ in $n$ variables $x_1, \dots, x_n$, by $Q_S$ we denote its coefficient corresponding to monomial $\prod_{i \in S} x_i$. \AP Furthermore, we say that $Q$ is \intro(PTF){normalized} if $\sum_{|S| > 0} Q_S^2 = 1$.

Given a $\PTF$ $f = \sgn(P)$ for some multilinear polynomial $P : \{0,1\}^n \to \mathbb{R}$ of bounded degree, we will analyze $P$ under a $p$-\kl{biased distribution}. First, we argue that the degree of the \kl{Fourier decomposition} of $P$ over $\cube{p}$ (for any $p \in (0,1)$) is at most the degree of $P$. Indeed, if $P = \sum_{S \subseteq [n]} a_S \cdot \prod_{i \in S} x_i$ over $\{0,1\}^n$, then its \kl{Fourier decomposition} over $\cube{p}$ is given by applying a linear transformation $x_i = p + \sqrt{p(1-p)} \cdot \chi^{(p)}_i$ for every $i \in [n]$. After this operation, every monomial of $P$ with non-zero coefficient $P_S \neq 0$ becomes
\begin{align*}
    P_S \cdot \prod_{i \in S} x_i &= P_S \cdot \prod_{i \in S} \left(p + \sqrt{p(1-p)} \cdot \chi^{(p)}_i \right) \\
    &= \sum_{T \subseteq S} P_S \cdot p^{|S|-|T|} \left(\sqrt{p(1-p)}\right)^{|T|} \cdot \prod_{i \in T} \chi^{(p)}_i,
\end{align*}
which clearly does not produce any monomials of degree greater than $|S|$.

This observation implies that whenever $f : \Bool^n \to \Bool$ is a $\PTF$ of degree $k$, a corresponding polynomial $Q$ over $\{ \chi_i^{(p)} : i \in [n] \}$ of degree $k$ can be found such that $f = \sgn(Q)$. We formalize this with the notion of \emph{sign-representation}.

\begin{definition}[\intro{Sign-representation}]\AP
    Suppose that $f: \Bool^n \to \Bool$ is a $\PTF$ of degree $k$ and $p \in (0,1)$. By a degree-$k$ \kl{sign-representation} of $f$ over $\cube{p}$, we call a degree-$k$, \kl(PTF){normalized}, multilinear polynomial $Q$ in $n$ variables such that
    \[
        f(x_1, \dots, x_n) = \sgn \left( Q \left(\chi_1^{(p)}(x_1), \dots, \chi_n^{(p)}(x_n) \right) \right).
    \]
\end{definition}

By the above argument, we know that every degree-$k$ $\PTF$ has a \kl{sign-representation} of degree $k$ for every $p \in (0,1)$.

Moving forward, we introduce a more concise notion of coordinate influence. Given a polynomial $Q$ of $n$ variables, for every $i \in [n]$, we denote $w_i^2(Q) = \sum_{S \ni i} Q_S^2$. For the rest of this section, we will assume that coordinates are sorted in a decreasing order with respect to $w_i^2(Q)$, i.e., $w_1^2(Q) \geq w_2^2(Q) \geq \dots \geq w_n^2(Q)$. We also write $\sigma^2_i(Q) = \sum_{j \geq i} w_j^2(Q)$. For conciseness, we write $w_i^4(Q)$ and $\sigma_i^4(Q)$ instead of $(w_i^2(Q))^2$ and $(\sigma_i^2(Q))^2$. The central notion of this section is the \emph{regularity}, defined as follows.

\begin{definition}[\intro(PTF){Regularity}]\AP
  A multilinear polynomial $Q$ is $\varepsilon$-regular if
  \[
    \sum_{i \in [\ar(Q)]} w_i^4(Q) \leq \varepsilon^2 \cdot \sigma^4_1(Q).
  \]
  A function $f \in \PTF_k$ is $\varepsilon$-regular over $\cube{p}$ if it has a $\varepsilon$-regular \kl{sign-representation} over $\cube{p}$.
\end{definition}

We note that the notion of regularity implies that $\max_i w_i^2(Q)$ is small.

\begin{claim}\label{noise_sensitivity:claim:regularity_is_small_influence}
    If $Q$ is a \kl(PTF){normalized}, $\varepsilon$-\kl(PTF){regular}, multilinear polynomial of degree $k$, then
    \[
        \max_{i \in [\ar(Q)]} w_i^2(Q) \leq \varepsilon \cdot k.
    \]
\end{claim}

\begin{proof}
    From \kl(PTF){normalization} and \kl(PTF){regularity}, for every $j \in [\ar(Q)]$ we have
    \begin{align*}
        w_j^4(Q) &\leq \sum_{i \in [\ar(Q)]} w_i^4(Q) \leq \varepsilon^2 \cdot \left( \sum_{i \in [\ar(Q)]} w_i^2(Q) \right)^2 \\
        &\leq \varepsilon^2 \cdot \left( \sum_{0 < |S| \leq k} |S| \cdot Q_S^2\right)^2 \leq (\varepsilon \cdot k)^2. \qedhere
    \end{align*}
\end{proof}

The most crucial idea behind the proof of \cref{noise_sensitivity:theorem:main_theorem} is that $\varepsilon$-\kl(PTF){regular}, bounded-degree polynomials admit \emph{anti-concentration} over the $p$-\kl{biased distributions}. In other words, for every range $I \subseteq \mathbb{R}$, the probability over $\cube{p}$ that the polynomial's value ends up in $I$, diminishes as $|I|$ and $\varepsilon$ approach $0$. This fact is deduced by combining two profound analytical results: the \emph{Invariance Principle} of Mossel, O’Donnell and Oleszkiewicz \cite{invariance_principle}, and the \emph{anti-concentration bound for Gaussian polynomials} of Carbery and Wright \cite{carbery_wright}.

In both of the following two statements, $\mathcal{N}$ denotes the standard Gaussian distribution with mean $0$ and variance $1$. By $\chardist{p,n}$ we denote the distribution of tuples $(x_1, \dots, x_n)$ such that $x_i = \chi_i^{(p)}(z_i)$ given $(z_1, \dots, z_n) \sim \cube{p,n}$.

\begin{theorem}[\intro{Invariance Principle}, {\cite[Theorem 2.1]{invariance_principle}}]\AP
  Suppose that $\lambda \in (0,1/2)$, $p \in (\lambda, 1-\lambda)$ and $Q$ is a \kl(PTF){normalized}, multilinear polynomial of degree $k$ in $n$ variables. If $Q$ is $\varepsilon$-\kl(PTF){regular}, then for every $t \in \mathbb{R}$:
  \[
    \left| \Pr_{\tuple x \sim \chardist{p,n}} \big[ Q (\tuple x ) \leq t \big] - \Pr_{\tuple y \sim \mathcal{N}^n} \big[ Q (\tuple y) \leq t \big] \right| \leq \mathcal{O}_{\lambda, k} \left( \varepsilon^{1/8k} \right).
  \]
\end{theorem}

\begin{theorem}[\intro{Carbery-Wright anti-concentration bound}, {\cite[Theorem 8]{carbery_wright}}]\AP
   Suppose that $Q$ is a \kl(PTF){normalized} polynomial of degree $k$ in $n$ variables, $I \subseteq \mathbb{R}$ is an interval of length $\alpha$. Then
  \[
    \Pr_{\tuple y \sim \mathcal{N}^n} \big[ Q(\tuple y) \in I \big] \leq \mathcal{O}_k \left( \alpha^{1/k} \right).
  \]
\end{theorem}

By combining the two above theorems, we obtain the following anti-concentration bound for bounded-degree multilinear polynomials over the $p$-\kl{biased distribution}.

\begin{corollary}\label{noise_sensitivity:corollary:hypercube_anticoncentration}
  Suppose that $\lambda \in (0,1/2), p \in (\lambda, 1-\lambda)$ and $Q$ is a \kl(PTF){normalized}, multilinear polynomial of degree $k$ in $n$ variables. If $Q$ is $\varepsilon$-regular, then for every range $I \subseteq \mathbb{R}$ of length at most $\alpha$, we have
  \[
    \Pr_{\tuple x \sim \chardist{p,n}}\big[ Q(\tuple x) \in I \big] \leq \mathcal{O}_{\lambda,k} \left( \alpha^{1/k} + \varepsilon^{1/8k} \right).
  \]
\end{corollary}

Another well known result we will use is the \emph{hypercontractivity} theorem of Bonami, Gross and Beckner for the $p$-\kl{biased distribution}. Hypercontractivity is perhaps one of the most significant concepts in analysis of Boolean (and not only Boolean) functions, with a rich history and a plethora of applications. Here, we only cite the inequality we will use\footnote{We note that the original statement of \cite[Theorem 6]{tail_bounds_boolean_functions} regards bounded-degree functions on $\{-1,1\}^n$, where every variable $x_i$ is drawn to be $1$ with probability $p$. This is equivalent to our rephrasing, as this distribution and $\chardist{p,n}$ are related by a linear transformation $x_i \mapsto (x_i + 1 - 2p)/(2\sqrt{p(1-p)})$, which does not increase the function degree. Furthermore, the constant $B(p)$ in the original statement depends only on $\lambda$ such that $p \in (\lambda, 1-\lambda)$, which can be easily verified by analyzing the explicit value they provide, which originates from \cite{oleszkiewicz_khinchine}.}; we refer to \cite[Chapters 9 and 10]{odonnell} for more on this topic.

\begin{theorem}[\intro{(2,4)-hypercontractivity}, { \cite[Theorem 6]{tail_bounds_boolean_functions}}]\label{noise_sensitivity:theorem:hypercontractivity}\AP
  Suppose that $\lambda \in (0,1/2)$, $p \in (\lambda, 1-\lambda)$ and $Q$ is a multilinear polynomial of degree $k$ in $n$ variables. Then
  \[
    \E[\tuple x \sim \chardist{p,n}]{ Q(\tuple x)^4 } \leq \mathcal{O}_{\lambda, k}(1) \cdot \E[\tuple x \sim \chardist{p,n}]{ Q(\tuple x)^2 }^2.
  \]
\end{theorem}

\begin{corollary}[\intro{Two-function hypercontractivity}]\label{noise_sensitivity:corollary:two_function_hypercontractivity}\AP
  Suppose that $\lambda \in (0,1/2)$, $p \in (\lambda, 1-\lambda)$ and $P, Q$ are multilinear polynomials of degree $k$ in $n$ variables. Then
  \[
    \E[\tuple x \sim \chardist{p,n}]{ P(\tuple x)^2 \cdot Q(\tuple x)^2 } \leq \mathcal{O}_{\lambda, k}(1) \cdot \E[\tuple x \sim \chardist{p,n}]{ P(\tuple x)^2 } \cdot \E[\tuple x \sim \chardist{p,n}]{ Q(\tuple x)^2 }.
  \]
\end{corollary}

\begin{proof}
  The statement follows from combining Cauchy-Schwarz inequality $\E{\mathbf{X}\mathbf{Y}}\leq \sqrt{\E{\mathbf{X}^2} \E{\mathbf{Y}^2}} $ with \cref{noise_sensitivity:theorem:hypercontractivity}:
  \begin{align*}
    \E{ P(\tuple x)^2 \cdot Q(\tuple x)^2 } &\leq \sqrt{\E{ P(\tuple x)^4 } \cdot \E{ Q(\tuple x)^4 }} \\
    &\leq \mathcal{O}_{\lambda, k}(1) \cdot \E{ P(\tuple x)^2 } \cdot \E{ Q(\tuple x)^2 }. \qedhere
  \end{align*}
\end{proof}

The last ingredient we introduce before proceeding to our proofs is the classic \emph{Paley-Zygmunt inequality}, which we prove for the reader's convenience.

\begin{theorem}[\intro{Paley-Zygmunt inequality}]\label{noise_sensitivity:theorem:paley_zygmunt}\AP
  Suppose that $\mathbf{X} \geq 0$ is a random variable with finite variance and $t \in [0,1]$. Then
  \[
    \Pr \Big[ \mathbf{X} > t \cdot \E{\mathbf{X}} \Big] \geq (1- t)^2 \cdot \frac{\E{\mathbf{X}}^2}{\E{\mathbf{X}^2}}.
  \]
\end{theorem}

\begin{proof}
    Let $\mathbf{1}^{\leq}$ be equal to $1$ if $\mathbf{X} \leq t \cdot \E{\mathbf{X}}$ and $0$ otherwise. Similarly define $\mathbf{1}^{>}$. Observe that $\E{\mathbf{X} \cdot \mathbf{1}^{\leq}} \leq t \cdot \E{\mathbf{X}}$. On the other hand, we have $\E{\mathbf{X} \cdot \mathbf{1}^>} \leq \sqrt{\E{\mathbf{X}^2} \cdot \Pr[\mathbf{X} > t \cdot \E{\mathbf{X}}]}$ by Cauchy-Schwarz. We obtain that
    \begin{align*}
        \E{\mathbf{X}} &= \E{\mathbf{X} \cdot \mathbf{1}^\leq} + \E{\mathbf{X} \cdot \mathbf{1}^>} \\
        &\leq t \cdot \E{\mathbf{X}} + \E{\mathbf{X}^2}^{1/2} \cdot \Pr\left[\mathbf{X} > t \cdot \E{\mathbf{X}}\right]^{1/2}.
    \end{align*}
    The statement now follows by rearranging the terms in above inequality.
\end{proof}

\subsection{Noise sensitivity of regular polynomial threshold functions}

We now proceed to the proof of \cref{noise_sensitivity:theorem:main_theorem}. We start with the proof for the special case when $f$ is $\varepsilon$-regular over the corresponding distribution, to which we will eventually reduce the general case, as we shall see later.

\begin{lemma}\label{noise_sensitivity:lemma:the_regular_case}
  Suppose that $\lambda \in (0,1/2), p \in (\lambda, 1-\lambda)$, and $f :\Bool^n \to \Bool$ is a $\PTF$ of degree $k$. If $f$ is $\varepsilon$-\kl(PTF){regular} over $\cube{p}$, then
  \[
    \NS[(p)]{f,1-\varepsilon} \leq \mathcal{O}_{\lambda, k} \left( \varepsilon^{1/8k} \right).
  \]
\end{lemma}

\begin{proof}
    \newcommand{\X}[0]{\mathbf{X}}
    \newcommand{\Y}[0]{\mathbf{Y}}
    Fix $\lambda, p, f$, and $\varepsilon$ as in the statement. Let $Q$ be a degree-$k$, $\varepsilon$-regular \kl{sign-representation} of $f$ over $\cube{p}$. Let $\tuple x \sim \cube{p, n}$ and $\tuple y \sim \N[(p)]{\tuple x, 1-\varepsilon}$. Then, let $\tuple \X = (\X_1, \dots, \X_n) = (\chi_1^{(p)}(x_1), \dots, \chi_n^{(p)}(x_n))$ and $\tuple \Y = (\Y_1, \dots, \Y_n) = (\chi_1^{(p)}(y_1), \dots, \chi_n^{(p)}(y_n))$. Note that both $\tuple \X$ and $\tuple \Y$ have marginal distributions $\chardist{p,n}$. We can bound the \kl{noise sensitivity} of $f$ as follows. For every $\Delta \in (0,1)$, we have:
  \begin{align*}
    \NS[(p)]{f,1-\varepsilon} &= \Pr \left[ \sgn \left(Q \left(\tuple \X \right) \right) \neq \sgn \left(Q \left(\tuple \Y \right) \right) \right] \\
                              &\leq \Pr \left[ \left| Q \left( \tuple \X \right) \right| \leq \sqrt{\varepsilon/\Delta} \right] + \Pr \left[ \left| Q \left(\tuple \X \right) - Q \left(\tuple \Y \right) \right| \geq \sqrt{\varepsilon/\Delta} \right].
  \end{align*}
  The first term can be simply bound using \cref{noise_sensitivity:corollary:hypercube_anticoncentration}:
  \[
    \Pr \left[ \left| Q \left( \tuple \X \right) \right| \leq \sqrt{\varepsilon/\Delta} \right] \leq \mathcal{O}_{\lambda, k} \left(\varepsilon^{1/2k} / \Delta^{1/2k} + \varepsilon^{1/8k} \right) \tag{{\color{ao(english)}$1$}}\label{noise_sensitivity:eq:first_part}
  \]
  We now proceed to the second term. Using the law of total expectation $\E{\X \Y} = \E{\X \E{\Y | \X}}$, $\E{\X_i} = 0$ and $\E{\X_i^2} = 1$, we obtain that for every $i \in [n]$:
  \begin{align*}
    \E{\X_i \cdot \Y_i} &= \E{\X_i \cdot \E{\Y_i | \X_i}} = \E{ \X_i \left( (1-\varepsilon)\X_i + \varepsilon \E{\X_i} \right)} \\ &= (1-\varepsilon) \E{\X_i^2} = (1-\varepsilon).
  \end{align*}
  For every $S \subseteq [n]$, let $\X_S = \prod_{i \in S} \X_i$ and $\Y_S = \prod_{i \in S} \Y_i$. Observe that $\E{\X_i - \Y_i} = 0$, and $\E{(\X_S - \Y_S)(\X_T - \Y_T)} = 0$ whenever $S \neq T$. The following bound follows:
  \begin{align*}
      \E{ \left( Q \left( \tuple \X \right) - Q \left( \tuple \Y \right) \right)^2} &= \E{ \left( \sum_{S \subseteq [n]} Q_S \cdot \left( \X_S - \Y_S \right) \right)^2 } \\
      &= \sum_{S \subseteq [n]} Q_S^2 \cdot \E{(\X_S - \Y_S)^2} \\
      &= 2 \sum_{S \subseteq [n]} Q_S^2 \cdot \left( 1 - \E{\X_S \cdot \Y_S} \right) \\
      &= 2 \sum_{S \subseteq [n]} Q_S^2 \cdot \left(1- (1-\varepsilon)^{|S|} \right).
  \end{align*}
  Furthermore, observe that $(1-(1-\varepsilon)^{|S|}) \leq \varepsilon \cdot |S|$ by the Bernoulli's inequality\footnote{Bernoulli's inequality states that $(1+x)^r \geq 1+rx$ for every integer $r \geq 0$ and real number $x \geq -1$.}, so the last term is at most $2 \varepsilon \cdot k$, because $Q$ is \kl(PTF){normalized} and has degree $k$.

  With the bound on expectancy in hand, we apply the Markov inequality to obtain that
  \[
    \Pr \left[ \Big| Q \left( \tuple \X \right) - Q \left(\tuple \Y \right) \Big| \geq \sqrt{\varepsilon/\Delta} \right] \leq \frac{2 \varepsilon \cdot k}{\varepsilon/\Delta} = 2 \Delta \cdot k. \tag{{\color{ao(english)}$2$}}\label{noise_sensitivity:eq:second_part}
  \]
  Combining \eqref{noise_sensitivity:eq:first_part} and \eqref{noise_sensitivity:eq:second_part}, we obtain
  \[
    \NS[(p)]{f,1-\varepsilon} \leq \mathcal{O}_{\lambda, k} \left( \varepsilon^{1/2k} / \Delta^{1/2k} + \varepsilon^{1/8k} + \Delta \right).
  \]
  It remains to plug in $\Delta \coloneqq \varepsilon^{1/8k}$.
\end{proof}

\subsection{Random restrictions of polynomial threshold functions}
\AP
In order to reduce \cref{noise_sensitivity:lemma:the_regular_case} to the general case, we will prove two results regarding \emph{random restrictions} of \kl{polynomial threshold functions}. Suppose that $f \colon A^n \to B$, $M \leq n$ and $\tuple x \in A^M$; by $\rest{f}{M}{\tuple x}$ we denote the \intro(PTF){restriction} of $f$ obtained by assigning values from $\tuple x$ to the first $M$ coordinates, i.e., the $(n-M)$-ary function defined as $\tuple y \mapsto f(\tuple x \tuple y)$.

We now introduce the notion of \emph{critical indices}, which are designed to capture the boundary separating influential coordinates of a function.

\begin{definition}[\intro{Critical index}]\AP
    Given $\varepsilon > 0$ and a multilinear polynomial $Q$, we define the $\varepsilon$-\kl{critical index} $K(Q, \varepsilon)$ of $Q$ as\footnote{Recall that we assume $w_1^2(Q) \geq w_2^2(Q) \geq \dots \geq w_{\ar(Q)}^2(Q)$.}
    \[
        K(Q, \varepsilon) = \min \left\{ i \geq 0 : \forall \, j > i, \ w_j^2(Q) \leq \varepsilon^2 \cdot \sigma_{i+1}^2(Q) \right\}.
    \]
\end{definition}

We now show a quite intuitive fact; if we randomly assign values to all variables up to the \kl{critical index}, we obtain a \kl(PTF){regular} \kl(PTF){restriction} with constant probability.

\begin{lemma}\label{noise_sensitivity:lemma:random_restriction_regular}
    Suppose that $\lambda \in (0,1/2), p \in (\lambda, 1-\lambda)$ and $k \geq 1$. There exist constants $\alpha = \alpha(\lambda, k) > 0, \tau = \tau(\lambda, k) > 0$ such that the following holds. Suppose that $\varepsilon > 0$ and $Q$ is a degree-$k$, multilinear polynomial in $n$ variables. Let $K(Q, \varepsilon)$ be the $\varepsilon$-\kl{critical index} of $Q$. Then
    \[
      \Pr_{\tuple x \sim \chardist{p, K}} \Big[ \rest{Q}{K}{\tuple x} \text{ is $(\alpha \cdot \varepsilon)$-\kl(PTF){regular}} \, \Big] \geq \tau.
    \]
\end{lemma}

\begin{proof}
    \newcommand{\X}[0]{\mathbf{X}}
    \newcommand{\Y}[0]{\mathbf{Y}}
    Fix $\lambda, p, k, Q$ and $K$ as in the statement. Let $\tuple x \sim \chardist{p, K}$ and denote $R = \rest{Q}{K}{\tuple x}$. We write $R$ as
    \begin{gather*}
        R(\tuple y) = \sum_{S \subseteq [K+1, n]} R_S(\tuple x) \cdot \prod_{i \in S} y_i, \text{ where }\\
        \forall \, S \subseteq [K+1, n] : R_S(\tuple x) = \sum_{T \subseteq [K]} Q_{S \cup T} \cdot \prod_{i \in T} x_i.
    \end{gather*}
    \sloppy Our first goal is to upper-bound the expected value of $\X = \sum_{i > K} w_i^4(R)$. Observe that $\E{R_S(\tuple x)^2} = \sum_{T \subseteq [K]} Q_{S \cup T}^2$ from \kl{Parseval identity}. For every $i \in [K+1, n]$, we apply \kl{two-function hypercontractivity} to obtain:
    \begin{align*}
        \E{ w_i^4(R) } &= \E{ \left( \sum_{S \ni i} R_S(\tuple x)^2 \right)^2 }
                       = \sum_{S \cap T \ni i}\E{ R_S(\tuple x)^2 \cdot R_T(\tuple x)^2 } \\
                       &\leq \mathcal{O}_{\lambda, k}(1) \cdot \sum_{S \cap T \ni i} \E{ R_S(\tuple x)^2 } \cdot \E{ R_T(\tuple x)^2 } \\
                       &= \mathcal{O}_{\lambda, k}(1) \cdot \left( \sum_{S \ni i} \E{ R_S(\tuple x)^2 } \right)^2
                       = \mathcal{O}_{\lambda, k}(1) \cdot w_i^4(Q).
    \end{align*}
  By the definition of \kl{critical index}, for every $i > K$, we have $w_i^2(Q) \leq \varepsilon^2 \cdot \sigma_{K+1}^2(Q)$. By summing the above bound over coordinates $i > K$, we obtain that
  \begin{align*}
    \E{\X} &\leq \mathcal{O}_{\lambda, k}(1) \cdot \sum_{i > K} w_i^4(Q) \\
    &\leq \mathcal{O}_{\lambda, k}(1) \cdot \left( \varepsilon^2 \sigma_{K+1}^2 \right)\left( \sum_{i > K} w_i^2(Q) \right)\\
    &= \mathcal{O}_{\lambda, k}(1) \cdot \varepsilon^2 \sigma^4_{K+1}(Q). \tag{{\color{ao(english)}$1$}}\label{noise_sensitivity:eq:random_restriction_1}
  \end{align*}
  Moving forward, we want to show that the random variable $\Y = \sigma_{K+1}^2(R)$ is bounded away from zero with constant probability. To this end, we want to bound the first and second moments of $\Y$ and apply the \kl{Paley-Zygmunt inequality}. We start with the expected value:
  \[
    \E{\Y} = \sum_{i > K} \E{ \sum_{S \ni i} R_S(\tuple x)^2 } = \sum_{i > K} w_i^2(Q) = \sigma^2_{K+1}(Q).
  \]
  Next, we bound the second moment using the already established $\E{w_i^4(R)} \leq \mathcal{O}_{\lambda, k}(w_i^4(Q))$ and Cauchy-Schwarz inequality:
  \begin{align*}
    \E{ \Y^2 } &= \E{ \left( \sum_{i > K} w_i^2(R) \right)^2 } = \sum_{i, j > K} \E{ w_i^2(R) \cdot w_j^2(R) } \\
    &\leq \sum_{i, j > K} \sqrt{ \E{ w_i^4(R) } \E{ w_j^4(R) } } \\
                                    &\leq \mathcal{O}_{\lambda, k}(1) \cdot \sum_{i, j > K} w_i^2(Q) \cdot w_j^2(Q) \\
                                    &= \mathcal{O}_{\lambda, k}(1) \cdot \left( \sum_{i > K} w_i^2(Q) \right)^2 = \mathcal{O}_{\lambda, k}(1) \cdot \sigma^4_{K+1}(Q).
  \end{align*}
  Combining the two above bounds with \kl{Paley-Zygmunt inequality} for $t = 1/2$, we obtain that there exists a constant $\tau = \tau(\lambda, k) > 0$, such that
  \[
    \Pr \left[ \sigma_{K+1}^2(R) \geq \frac{1}{2} \cdot \sigma^2_{K+1}(Q) \right] \geq \left(1- \frac{1}{2}\right)^2 \cdot \frac{\E{\Y}^2}{\E{\Y^2}} \geq 2\tau. \tag{{\color{ao(english)}$2$}}\label{noise_sensitivity:eq:random_restriction_2}
  \]
  To finish the proof, we observe that combining \eqref{noise_sensitivity:eq:random_restriction_1} with Markov inequality for $\X$, we deduce that there exists a constant $\alpha = \alpha(\lambda, k) > 0$, such that
  \[
    \Pr \left[ \sum_{i > K} w_i^4(R) \geq (\alpha \cdot \varepsilon)^2 \cdot  \frac{\sigma^4_{K+1}(Q)}{4} \right] \leq \frac{\E{\X}}{(\alpha \cdot \varepsilon)^2 (\sigma_{K+1}^4(Q)/4)} \leq \tau.
  \]
  Combining this inequality and \eqref{noise_sensitivity:eq:random_restriction_2}, we finally obtain that
  \begin{align*}
    \Pr[ R \text{ is $(\alpha \cdot \varepsilon)$-\kl(PTF){regular}}] = \Pr \left[ \sum_{i > K} w_i^4(R) \leq (\alpha \cdot \varepsilon)^2 \sigma^4_{K+1}(R) \right] \geq \tau,
  \end{align*}
  which finishes the proof.
\end{proof}

\AP
We say that a polynomial $Q$ in $n$ variables is $\varepsilon$-\intro{determined} over $\chardist{p, n}$ if either $\Pr[Q(\tuple x) > 0]$ or $\Pr[Q(\tuple x) \leq 0]$ is at most $\varepsilon$, given that $\tuple x \sim \chardist{p, n}$. The next lemma asserts that if the \kl{critical index} is sufficiently large, then a random \kl(PTF){restriction} obtained by fixing values of variables up to the critical index yields a strongly determined function with constant probability.

\begin{lemma}\label{noise_sensitivity:lemma:random_restriction_determined}
    Suppose that $\lambda \in (0,1/2), p \in (\lambda, 1-\lambda)$ and $k \geq 1$. There exist constants $\beta = \beta(\lambda, k) > 0,c = c(\lambda, k) > 0,\tau = \tau(\lambda, k) > 0$ such that the following holds. Suppose that $\varepsilon > 0$ and $Q$ is a degree-$k$, multilinear polynomial in $n$ variables, and let $m = c \cdot \log_{1-\varepsilon^2}(\varepsilon/2)$. If the $\varepsilon$-\kl{critical index} $K(Q, \varepsilon)$ is at least $m$, then
    \[
        \Pr_{\tuple x \sim \chardist{p, m}} \left[ \rest{Q}{m}{\tuple x} \text{ is $(\beta \cdot \varepsilon)$-\kl{determined} } \right] \geq \tau.
    \]
\end{lemma}

\begin{proof}
    Fix $\lambda, p, k, \varepsilon, Q$ and $m$ as in the statement. Let $\tuple x \sim \chardist{p, n}$. We decompose $Q$ as $Q^{\subseteq m} + Q^{\not \subseteq m}$ where
    \[
        Q^{\subseteq m}(\tuple x) = \sum_{S \subseteq [m]} Q_S \cdot \prod_{i \in S} x_i.
    \]
    First, we show that $(Q^{\subseteq m})^2$ is significant with constant probability. To this end, we combine the \kl{Paley-Zygmunt inequality} for $t = 1/4$ with \kl{(2,4)-hypercontractivity} to deduce that there exists a constant $\tau = \tau(\lambda, k) > 0$, such that:
    \[
        \Pr \left[ \left( Q^{\subseteq m} \right)^2 \geq \frac{1}{4} \cdot \E{\left( Q^{\subseteq m} \right)^2 }\right] \geq \frac{ 9 \cdot \E{ \left( Q^{\subseteq m} \right)^2 }^2 }{ 16 \cdot \E{ \left( Q^{\subseteq m}\right)^4 }} \geq 2\tau. \tag{{\color{ao(english)}$1$}}\label{noise_sensitivity:eq:P_L_is_large}
    \]
    Moving forward, we show that as $c$ grows, the expected value $(Q^{\not \subseteq m})^2$ vanishes. Observe that for every $i < m$, from the definition of \kl{critical index} and $m \leq K(Q, \varepsilon)$, we have $w_i^2(Q) > \varepsilon^2 \cdot \sigma_i^2(Q)$. Therefore,
    \[
        \sigma_i^2(Q) = w_i^2(Q) + \sigma_{i+1}^2(Q) > \varepsilon^2 \cdot \sigma^2_i(Q) + \sigma_{i+1}^2(Q).
    \]
    In particular, we have $\sigma_{i+1}^2(Q) < (1-\varepsilon^2) \cdot \sigma^2_i(Q)$. By iterating this inequality, we obtain that $\sigma^2_m(Q) \leq (1-\varepsilon^2)^{m-1} \cdot \sigma_1^2(Q)$. We now have
    \begin{align*}
        \E{\left( Q^{\not \subseteq m} \right)^2 } &= \sum_{S \not \subseteq [m]} Q_S^2  \leq \sigma_m^2(Q) \leq (1-\varepsilon^2)^{m-1} \cdot \sigma_1^2(Q) \\
        &\leq k (1-\varepsilon^2)^{m-1} \cdot \sum_{S \subseteq [n]} Q_S^2.
    \end{align*}
    For sufficiently large $c = c(\lambda, k) > 0$, we obtain that
    \begin{align*}
        \E{\left( Q^{\not \subseteq m} \right)^2 } &=\sum_{S \not \subseteq [m]} Q_S^2 \leq \frac{k(1-\varepsilon^2)^{m-1}}{1-k(1-\varepsilon^2)^{m-1}} \cdot \sum_{S \subseteq [m]} Q_S^2 \\ &< \frac{\varepsilon}{4} \cdot \E{\left( Q^{\subseteq m} \right)^2}.
    \end{align*}
    We continue with Markov inequality to upper-bound the probability that $(Q^{\not \subseteq m})^2$ is large:
    \[
        \Pr \left[ \left( Q^{\not \subseteq m} \right)^2 \geq \frac{1}{4} \cdot \E{ \left( Q^{\subseteq m} \right)^2 } \right] \leq \frac{\E{(Q^{\not \subseteq m})^2}}{\E{(Q^{\subseteq m})^2}/4} < \varepsilon. \tag{{\color{ao(english)}$2$}}\label{noise_sensitivity:eq:P_not_L_is_small}
    \]
    Let $\tuple y = (x_1, \dots, x_m)$ and $\tuple z = (x_{m+1}, \dots, x_n)$. We think about drawing $\tuple x$ as first drawing $\tuple y$, and then $\tuple z$. From \eqref{noise_sensitivity:eq:P_not_L_is_small}, we have that the event on the left-hand side of \eqref{noise_sensitivity:eq:P_not_L_is_small}, conditioned on $\tuple y$, can have probability at least $(\varepsilon/\tau)$ for at most a $\tau$-fraction of tuples $\tuple y$. Therefore, there is probability at least $\tau$ of drawing $\tuple y$ such that the event on left-hand side of \eqref{noise_sensitivity:eq:P_L_is_large} holds, and
    \[
        \Pr_{\tuple z} \left[ \left| Q^{\not \subseteq m}(\tuple y \tuple z) \right| \geq \left| Q^{\subseteq m}(\tuple y)\right| \right] < (1/\tau) \cdot \varepsilon.
    \]
    However, note that every such $\tuple y$ gives a $(\varepsilon/\tau)$-\kl{determined} \kl(PTF){restriction} $\rest{Q}{m}{\tuple y}$. Therefore, we finish the proof by setting $\beta := 1/\tau$, which is a positive constant depending only on $\lambda$ and $k$.
\end{proof}

\subsection{\texorpdfstring{A proof of \cref*{noise_sensitivity:theorem:main_theorem}}{A proof of the main noise sensitivity theorem}}

We are now in a position to prove \cref{noise_sensitivity:theorem:main_theorem} in full generality. We start with a claim combining \cref{noise_sensitivity:lemma:random_restriction_regular} with \cref{noise_sensitivity:lemma:random_restriction_determined} to deduce a bound on \kl{noise sensitivity} of random \kl(PTF){restrictions} of $\PTF$s.

\begin{claim}\label{noise_sensitivity:claim:aux_for_induction}
  Suppose that $\lambda \in (0,1/2), p \in (\lambda, 1-\lambda)$ and $k \geq 1$. There exist constants $c = c(\lambda, k) > 0, \Delta = \Delta(\lambda, k) > 0$ and $\tau = \tau(\lambda, k) > 0$ such that the following holds. Suppose that $\varepsilon > 0$ and $f \in \PTF_k$ has a degree-$k$ \kl{sign-representation} $Q$ over $\cube{p}$. If $M = \min\{K(Q, \varepsilon), c \cdot \log_{1-\varepsilon^2}(\varepsilon/2)\}$, then
  \[
    \Pr_{\tuple x \sim \cube{p,M}} \left[ \NS[(p)]{ \rest{f}{M}{\tuple x} , 1-\varepsilon} \leq \Delta \cdot \varepsilon^{1/8k} \right] \geq \tau.
  \]
\end{claim}
\begin{proof}
    \newcommand{\X}[0]{\mathbf{X}}
  Fix $\lambda, p, k, \varepsilon, f, Q$ and $K$. Let $\alpha, \beta, \tau$ be the constants provided by \cref{noise_sensitivity:lemma:random_restriction_regular} and \cref{noise_sensitivity:lemma:random_restriction_determined} for $\lambda$ and $k$ ($\tau$ is the minimum of the two provided by the lemmas). Let $\tuple x \sim \cube{p, M}$, $\tuple \X = (\X_1, \dots, \X_M) = (\chi_1^{(p)}(x_1), \dots, \chi_M^{(p)}(x_M))$, and $g(\tuple y) = \rest{f}{M}{\tuple x}$. Note that $\rest{Q}{M}{\tuple \X}$ is a degree-$k$ \kl{sign-representation} of $g$ over $\cube{p}$. We consider two cases, depending on the value of $M$:
  \begin{enumerate}[leftmargin=*]
    \item $M = K(Q, \varepsilon)$. \cref{noise_sensitivity:lemma:random_restriction_regular} gives that $\rest{Q}{M}{\tuple \X}$, and as a consequence $g$, is $(\alpha \cdot \varepsilon)$-\kl(PTF){regular} with probability at least $\tau$. If $\alpha \leq 1$, then $g$ is in particular $\varepsilon$-\kl(PTF){regular} and by \cref{noise_sensitivity:lemma:the_regular_case}, we have $\NS[(p)]{g, 1-\varepsilon} \leq \mathcal{O}_{\lambda, k}(\varepsilon^{1/8k})$. Otherwise, we have $\alpha > 1$ and $\NS[(p)]{g, 1-\varepsilon} \leq \NS[(p)]{g, 1-\alpha \varepsilon} \leq \mathcal{O}_{\lambda, k}(\varepsilon^{1/8k})$.

    \item $M = c \cdot \log_{1-\varepsilon^2}(\varepsilon/2)$. By \cref{noise_sensitivity:lemma:random_restriction_determined}, $\rest{Q}{M}{\tuple \X}$ is $(\beta \cdot \varepsilon)$-\kl{determined} with probability at least $\tau$. Therefore, there is a value $b \in \{0,1\}$ such that $\Pr_{\tuple y}[g(\tuple y) \neq b] \leq (\beta \cdot \varepsilon)$, given $\tuple y \sim \cube{p, n-M}$. If $\tuple z \sim \N[(p)]{\tuple y, 1-\varepsilon}$, then also $\Pr_{\tuple z}[g(\tuple z) \neq b] \leq (\beta \cdot \varepsilon)$ as the marginal distribution of $\tuple z$ is $\cube{p, n-M}$. Therefore, we have $\NS[(p)]{g, 1-\varepsilon} = \Pr[g(\tuple y) \neq g(\tuple z)] \leq 2(\beta \cdot \varepsilon) \leq \mathcal{O}_{\lambda, k}(\varepsilon^{1/8k})$. \qedhere
  \end{enumerate}
\end{proof}

Finally, we prove \cref{noise_sensitivity:theorem:main_theorem} by interpolating the bounds on the \kl{noise sensitivity} of random restrictions to the whole function.

\noisesensitivitymaintheorem*

\begin{proof}
    \newcommand{\X}[0]{\mathbf{X}}
    Fix $\lambda, p, k$ and $f$ as in the statement. Let $c, \Delta$ and $\tau$ be constants provided by \cref{noise_sensitivity:claim:aux_for_induction} for $\lambda$ and $k$. Furthermore, let $\varepsilon > 0$, $L = c \cdot \log_{1-\varepsilon^2}(\varepsilon/2)$ and $t = \log_{1-\tau}(\varepsilon)$. From the fact that $-\ln(1-x) \leq x/(1-x)$ for $x \in [0,1)$, we get $-\ln(1-\varepsilon^2) \leq \varepsilon^2/(1-\varepsilon^2) \leq \varepsilon$ for sufficiently small $\varepsilon$, and hence $\log_{1-\varepsilon^2}(\varepsilon/2) = \ln(2/\varepsilon)/(-\ln(1-\varepsilon^2)) \geq \ln(2/\varepsilon)/\varepsilon$. Therefore, $(L \cdot t) \geq \Omega_{\lambda, k}(\ln^2(1/\varepsilon)/\varepsilon)$.

    We will show that for
    \[
        \delta = \frac{\varepsilon^{1/8k}}{L \cdot t} = \mathcal{O}_{\lambda, k}\left( \frac{\varepsilon^{(8k + 1)/8k} }{\ln^2(1/\varepsilon)} \right),
    \]
    we have
    \[
        \NS[(p)]{f,1-\delta} = \mathcal{O}_{\lambda, k} \left( \varepsilon^{1/8k} \right). \tag{{\color{ao(english)}$1$}}\label{noise_sensitivity:eq:last_proof_what_we_show}
    \]
    Observe that $\delta < \varepsilon$ for sufficiently small $\varepsilon > 0$. Before proceeding, we explain how the theorem's statement follows from \eqref{noise_sensitivity:eq:last_proof_what_we_show}. For every sufficiently small $\zeta > 0$, there is $\varepsilon$ such that $\zeta = \varepsilon^{(8k+1)/(8k)}/\ln^2(1/\varepsilon)$ and $\zeta < \varepsilon$. Then, \eqref{noise_sensitivity:eq:last_proof_what_we_show} implies that
    \[
        \NS[(p)]{f, 1-\zeta} \leq \mathcal{O}_{\lambda, k} \left( \ln^{2/(8k+1)}(1/\zeta)  \zeta^{1/(8k+1)} \right) \leq \mathcal{O}_{\lambda, k} \left( \zeta^{1/(8k+2)} \right).
    \]

    We now proceed to proving \eqref{noise_sensitivity:eq:last_proof_what_we_show}. Denote $n = \ar(f)$ and let $Q$ be a degree-$k$ \kl{sign-representation} of $f$. For every $\tuple x \sim \cube{p, n}$, we define a sequence of sets $S_1, S_2, \dots$ as follows. Let $\tuple \X = (\X_1, \dots, \X_n) = (\chi_1^{(p)}(x_1), \dots, \chi_n^{(p)}(x_n))$. We let $S_1 \subseteq [n]$ be the set of $M_1 \leq L$ coordinates provided by \cref{noise_sensitivity:claim:aux_for_induction} for $Q$ and $\varepsilon$ (independent of $\tuple x$). For $i \geq 1$, we denote $Z_i = S_1 \cup \dots \cup S_i$ and let $S_{i+1}$ be the set of $M_{i+1} \leq L$ coordinates provided by  \cref{noise_sensitivity:claim:aux_for_induction} for $\rest{Q}{Z_i}{\tuple \X}$, which is the \kl(PTF){restriction} of $Q$ obtained by fixing values $\X_j$ according to $\tuple \X$, for every $j \in Z_i$. Finally, we denote $f_i^{(\tuple x)} = \sgn(\rest{Q}{Z_i}{\tuple \X})$. Note that $|Z_i| \leq i \cdot L$ and the definition of $f_i^{(\tuple x)}$ depends only on values $x_j$ such that $j \in Z_i$.

    We call a tuple $\tuple x$ \emph{good} if there is $i \in [t]$ such that $\NS[(p)]{f^{(\tuple x)}_i, 1-\varepsilon} \leq \Delta \cdot \varepsilon^{1/8k}$. If $\tuple x$ is good, then we denote $i(\tuple x)$ to be the index $i \in [t]$ witnessing this fact. By \cref{noise_sensitivity:claim:aux_for_induction} and the choice of $t$, we have
    \[
        \Pr_{\tuple x \sim \chardist{p, n}} \left[ \tuple x \text{ is not good} \right] \leq (1-\tau)^t = \varepsilon.
    \]
    Let $\tuple x \sim \cube{p, n}$ and $\tuple y \sim \N[(p)]{\tuple x, 1-\delta}$. If $\tuple x$ is good, $i = i(\tuple x)$ is defined and $|Z_i| \leq L t$. We have
    \[
        \Pr \left[ \exists \, j \in Z_i :  x_j \neq y_j \right] \leq Lt \cdot \delta = \varepsilon^{1/8k}.
    \]
    Furthermore, under the assumption that for every $j \in Z_i$, $x_j$ is fixed and $x_j = y_j$, the probability that $f(\tuple x) \neq f(\tuple y)$ is exactly $\NS[(p)]{f^{(\tuple x)}_i, 1-\delta} \leq \NS[(p)]{f_i^{(\tuple x)}, 1-\varepsilon} \leq \Delta \cdot \varepsilon^{1/8k}$, because $\delta < \varepsilon$. Combining all above observations, we deduce that if $f(\tuple x) \neq f(\tuple y)$, then either (1) $\tuple x$ is not good, or (2) there is $j \in Z_i$ such that $x_j \neq y_j$, or (3) in the remaining case, the probability that $f(\tuple x) \neq f(\tuple y)$ is at most $\Delta \cdot \varepsilon^{1/8k}$. This gives the final bound, finishing the proof of \eqref{noise_sensitivity:eq:last_proof_what_we_show} and the whole theorem:
    \[
        \NS[(p)]{f, 1-\delta} \leq \varepsilon + \varepsilon^{1/8k} + \Delta \cdot \varepsilon^{1/8k} = \mathcal{O}_{\lambda, k} \left( \varepsilon^{1/8k} \right). \qedhere
    \]
\end{proof}

\section{Ordered PCSPs: omitted proof}\label{appendix:ordered}
In this appendix, we provide a full proof of \cref{claim:a2}.
For the reader's convenience, we first briefly summarize the proof of~\cref{ordered:lem:kalai}, in the context of which~\cref{claim:a2} was stated.

\kalai*
\begin{proof}[Brief proof summary]
    Fix $\varepsilon > 0$ and $f$. Let $(L, R)$ be the interval
    \[
    \left\{p \in [0,1] : \E[p]{f} \in (\varepsilon, 1-\varepsilon)\right\} \cap \left(\frac{\varepsilon}{2}, 1-\frac{\varepsilon}{2}\right).
    \]
    of length at least $\varepsilon/2$.
    By the mean value theorem, there exists $p \in (L, R)$ such that the derivative
    \[
    \frac{d}{dp} \E[p]{f} = \frac{\E[R]{f} - \E[L]{f}}{R - L} \le \frac{2}{\varepsilon}.
    \]
    \noindent
    We conclude by the \kl{Margulis-Russo formula} that $\I[(p)]{f} \le 2/\varepsilon$.

    \begin{friedgutjunta}[See {\cite[Section 10.3]{odonnell}}]
        For every $\lambda, B, \nu > 0$, there exists $M = M(B, \nu, \lambda) > 0$ such that the following holds. Suppose $p \in [\lambda, 1-\lambda]$ and $f \colon \Bool^n \to \Bool$ is a function such that $\I[(p)]{f} \leq B$. Then there exists $J \subseteq [n]$ with $|J| \leq M$ and a $J$-junta $h \colon \Bool^n \to \Bool$ such that $\Pr_{\, \tuple x \sim \cube p \,}[f(\tuple x) \neq h(\tuple x)] \leq \nu$.
    \end{friedgutjunta}
    We apply this theorem to our function $f$ with $B = 2/\varepsilon$ and $\nu = \varepsilon^2/4$, and obtain a junta $h$ that approximates $f$.
    Note that $M$, the maximum possible number of \kl{essential} coordinates of $h$, is a constant depending only on $\varepsilon$.

    We prove that our functions must have influential coordinates in many \kl{biased} distributions:
    \atwo*
    \begin{proof}
        We shall assume that $q > p$ (otherwise, apply the proof to $\tuple x \mapsto 1 - f(1 - \tuple x)$ whose influences under $\cube{1-q}$ coincide with those of $f$ under $\cube q$).

        To simplify the presentation of the proof, we will identify any subset $S$ of $[n]$ with its indicator vector $\tuple 1_S \in \Bool^n$.
        In particular, by $S \sim \mu_q$ we denote a random set $S \subseteq [n]$ obtained by first sampling $\tuple x \sim \cube{q,n}$ and outputting the unique set $S$ such that $\tuple x = \tuple 1_S$.

        Let $P_1$ and $P_0$ be sampled independently from $\cube p$ conditioned on $P_1 \subseteq J$ and $P_0 \subseteq [n] \setminus J$.
        Put $P = P_0 \cup P_1$, and observe that $P \sim \cube{p}$.
        We first show that $f(P \cup J) = 1$ with probability at least $1-\varepsilon/2$.
        Note that $f(P) = f(P_0 \cup P_1) \le f(P_0 \cup J)$ by \kl(func){increasingness} of $f$.
        Furthermore $h(P) = h(P_1)$.
        Since $P_0$ and $P_1$ are independent, we have
        \begin{align*}
            \frac{\varepsilon^2}{4} = \nu &\ge \Pr_P[f(P) = 0 \text{ and } h(P) = 1] \\
            &\ge \Pr_{P_0}[f(P_0 \cup J) = 0] \cdot \Pr_{P_1} [h(P_1) = 1].
        \end{align*}
        On the other hand, observe that $\E[p]{h} \ge \E[p]{f} - \nu \ge \varepsilon/2$ because $\E[p]{f} > \E[L]{f} \ge \varepsilon$. Thus
        \[
            \Pr_{P_0}[f(P_0 \cup J) = 0] \cdot \Pr_{P_1} [h(P_1) = 1] \ge \Pr_{P}[f(P \cup J) = 0] \cdot \frac{\varepsilon}{2}.
        \]
        Combining these two inequalities, we obtain $\Pr[f(P \cup J) = 1] \ge 1 - \varepsilon/2$ as needed.

        Let $Q \supseteq P$ be obtained from $P$ by adding to it each $i \notin P$ with probability $(q-p)/(1-p)$; observe that $Q \sim \cube{q}$. From $\E[q]{f} < \E[R]{f} \le 1-\varepsilon$, we obtain that $f(Q) = 0$ with probability at least $\varepsilon$. The probability that, in addition, $f(Q \cup J) = 1$ is at least $\varepsilon - \varepsilon/2 = \varepsilon/2$ (by the union bound and \kl(func){increasingness} of $f$).

        To summarize, there is a family $\mathcal F \subseteq \mathcal P([n])$ of measure $\cube{q}(\mathcal F) \ge \varepsilon/2$ such that $f(Q) = 0$ and $f(Q \cup J) = 1$ for each $Q \in \mathcal F$.
        In order to bridge the gap between this family and the individual \kl{influences}, we recall that
        \begin{align*}
        \Inf[(q)]{f,j} &= \Pr_{\tuple x \sim \cube{q}}\left[ f(\tuple x) \neq f(\tuple x \oplus j) \right] &&\\
        &= \Pr_{\substack{S \sim \cube{q} \\ j \notin S}}\left[ f(S) = 0 \text{ and } f(S \cup \{j\}) = 1 \right] &&\text{(monotonicity)}\\
        &\ge \Pr_{S \sim \cube{q}}\left[ f(S) = 0 \text{ and } f(S \cup \{j\}) = 1 \right] &&\\
        &= \cube{q}\left( \left\{S \subseteq [n] \mid f(S) = 0 \text{ and } f(S \cup \{j\}) = 1 \right\} \right).&&
        \end{align*}
        Denote the family in the last term by $\mathcal S_j$.

        Consider adding elements of $J$ one by one to any $Q \in \mathcal F$ in an arbitrary order:
        there will be exactly one moment, when adding the next element switches the output of $f$ from $0$ to $1$.
        In other words, every $Q \in \mathcal F$ provides at least one pair $(S, j)$ such that $Q \subseteq S \subseteq Q \cup J$, and $j \in J \setminus S$, and $S \in \mathcal S_j$.
        On the other hand, the same pair $(S, j)$ may originate from up to $2^{|J|}$ different $Q$'s.
        Observe that the measure
        \[
        \cube{q}(S) \ge \cube{q}(Q) \cdot \left(\frac{q}{1-q}\right)^{|S \setminus Q|} \ge \cube{q}(Q) \cdot \left( \frac{\varepsilon}{2} \right)^{M-1}
        \]
        therefore
        \[
        \sum_{j \in J} \Inf[(q)]{f,j} \ge \sum_{j \in J} \cube{q}(\mathcal S_j) \ge \frac{\cube{q}(\mathcal F)}{2^M} \cdot \left( \frac{\varepsilon}{2} \right)^{M-1} \ge \left( \frac{\varepsilon}{4} \right)^M
        \]
        which means that some $j \in J$ satisfies
        \[
        \Inf[(q)]{f,j} \ge \frac{1}{M} \cdot \left( \frac{\varepsilon}{4} \right)^M
        \]\qedhere
    \end{proof}
\end{proof}

\section{Unate PCSPs: technical details}\label{appendix:unate}

\AP
We denote the set of \kl{increasing} coordinates of a Boolean function $f$ by $\intro*\domainup{f}$, and the set of \kl{decreasing} coordinates by $\intro*\domaindown{f}$.
A function $f$ is \reintro{unate} if $\domainup{f} \cup \domaindown{f} = [\ar(f)]$.
Note that $\domainup{f} \cap \domaindown{f}$ are exactly those coordinates on which $f$ does not depend.

\paragraph{Analytical properties of unate functions.}
We first recall the definition of $\E[p,q]{f}$. Given a \kl{unate} function $f \colon \Bool^n \to \Bool$ and numbers $p, q \in [0,1]$, let $\tuple r \in [0,1]^n$ be the tuple such that $r_i = p$ if $i$ is an \kl{increasing} coordinate, and $r_i = q$ otherwise. We will use $\E[p,q]{f}$ as a shortcut for $\E[\mu_{\tuple r}]{f}$.
In a similar way, we define
\[
\Inf[(p,q)]{f,i} = \E[\tuple x \sim \mu_{\tuple r}]{ \left( f(\tuple x) - f(\tuple x \oplus i) \right)^2 },
\]
and $\I[(p,q)]{f} = \sum_i \Inf[(p,q)]{f,i}$.

Recall the \kl{Margulis-Russo formula} for \kl(func){increasing} functions.
For \kl{unate} functions, a similar reasoning gives the following connection between $\E[p,q]{\cdot}$ and $\I[(p,q)]{\cdot}$:

\begin{restatable}{lemma}{unatepartialderiv}
    \label{lem:partderiv}
    Let $f$ be a \kl{unate} function. At any $(p, q) \in [0,1]^2$ the partial derivatives are
    \begin{align*}
    \frac{\partial}{\partial p} \E[p,q]{f} &= \sum_{\text{incr}~i} \Inf[(p,q)]{f,i} \\
    -\frac{\partial}{\partial q} \E[p,q]{f} &= \sum_{\text{decr}~i} \Inf[(p,q)]{f,i}
    \end{align*}
    Moreover, the map $(p,q) \mapsto \E[p,q]{f}$ is continuous on $[0,1]^2$ and smooth%
    \footnote{By smooth we mean infinitely differentiable.} on $(0,1)^2$.
\end{restatable}
\noindent We present a proof that is a direct generalization of the \kl{Margulis-Russo formula} proof from \cite{kalai2004shapley}.
\begin{proof}
    Let $n$ be the arity of $f$. Consider the general product probability distribution $\cube{\tuple p,n}$; to avoid double subscripts, we will sometimes denote it with $\mu(\tuple p)$. The measure of a tuple $\tuple x$ is

    \[
    \cube{\tuple p,n}(\tuple x) = \prod_{i : x_i = 1} p_i \prod_{i : x_i = 0} (1-p_i).
    \]

    Our goal is to compute the partial derivatives of $\E[\mu(\tuple p)]{f}$ w.r.t. all $p_j$'s.
    Observe that
    \begin{align*}
        \E[\mu(\tuple p)]{f} &= \sum_{\tuple x}\left\{\cube{\tuple p,n}(\tuple x) \,\,\big\lvert\,\, f(\tuple x) = 1\right\} = \\
        &= \sum_{\tuple x}\left\{ f(\tuple x) \cdot \cube{\tuple p,n}(\tuple x) \right\} = \\
        &= \sum_{\tuple x}\left\{ f(\tuple x) \cdot \prod_{i : x_i = 1} p_i \prod_{i : x_i = 0} (1-p_i) \right\}
    \end{align*}
    Fix some $j \in [n]$. We continue the chain of equalities
    \begin{align*}
        &= \sum_{\tuple x : x_j = 0}\left\{ \Big((1-p_j)f(\tuple x) + p_j f(\tuple x \oplus j)\Big) \cdot \prod_{i : x_i = 1} p_i \prod_{i \neq j : x_i = 0} (1-p_i)\right\}.
    \end{align*}
    Note that for those $\tuple x$ that satisfy $f(\tuple x) = f(\tuple x \oplus j)$ the summand does not depend on $p_j$. Now
    \begin{align*}
    \frac{\partial}{\partial p_j} \E[\mu(\tuple p)]{f} &
        \begin{aligned}[t]
            {}= \frac{\partial}{\partial p_j} \sum_{\tuple x : x_j = 0}\Bigg\{ \Big((1-p_j)f(\tuple x) + p_j f(\tuple x \oplus j)\Big) \cdot \\
            \cdot\prod_{i : x_i = 1} p_i \prod_{i \neq j : x_i = 0} (1-p_i) \Bigg\}
        \end{aligned}\\
    &
        \begin{aligned}
            {}= \sum_{\tuple x : x_j = 0}\Bigg\{ \frac{\partial}{\partial p_j}\Big((1-p_j)f(\tuple x) + p_j f(\tuple x \oplus j)\Big) \cdot \\
            \cdot\prod_{i : x_i = 1} p_i \prod_{i \neq j : x_i = 0} (1-p_i) \Bigg\}
        \end{aligned}\\
    &= \sum_{\tuple x : x_j = 0}\left\{ \Big(f(\tuple x \oplus j) - f(\tuple x)\Big) \cdot \prod_{i : x_i = 1} p_i \prod_{i \neq j : x_i = 0} (1-p_i) \right\} \\
    &= \pm\Inf[\mu(\tuple p)]{f,j}
    \end{align*}
    where the sign depends only on whether $j$ is \kl{increasing} or \kl{decreasing}. The first part of the lemma now follows from the chain rule.

    For the second part, note that $\E[p,q]{f}$ is a polynomial of $p$ and $q$ of degree at most $n$.
    Polynomials are known to be smooth on their entire domain $\mathbb R^2$.
    Therefore, the conclusion follows immediately.
\end{proof}
\noindent
An immediate corollary of \cref{lem:partderiv} is that $\E[p,q]{f}$ is increasing in $p$ and decreasing in $q$:
\begin{corollary}\label{cor:Emonotonicity}
    Let $f$ be a \kl{unate} function. If $p \le p'$ and $q \ge q'$, then
    \[
    \E[p,q]{f} \le \E[p',q']{f}.
    \]
\end{corollary}

\subsection{Proofs omitted in Tractability}

For \kl{unate} functions $f$,
we consider $\varepsilon$-level sets of $\E[p,q]{f}$, that is
\[
\curve{\varepsilon}{f} = \{(p, q) \in [0,1]^2 : \E[p,q]{f} = \varepsilon\}.
\]
Suppose $0 < \varepsilon < 1$.
If $f$ is a constant function, then $\curve{\varepsilon}{f}$ is the empty set.
Otherwise, we will argue that it is a smooth curve and prove several important properties\footnote{We will only consider \kl{idempotent} functions, because they are sufficient for the tractability proof. However, the level sets for any \kl{unate} function also satisfy similar properties}.
Note that $\curve{\varepsilon}{f}$ is the analog of ``the unique $p$ such that $f(p) = \varepsilon$'' for \kl(func){increasing} functions.
\begin{lemma}\label{lem:curve}
    Let $f$ be an \kl{idempotent} \kl{unate} function. For any $\varepsilon \in (0,1)$, the set $\curve{\varepsilon}{f}$ is a smooth monotonic curve from the bottom edge to the top edge of the unit square.
\end{lemma}
\begin{proof}

    By \cref{lem:partderiv} both partial derivatives of $\E[p,q]{f}$ exist and are non-negative. We claim that one of them is always strictly positive:
    \begin{claim}
        For any $(p,q) \in (0,1)\times[0,1]$, the partial derivative $\partial\E[p,q]{f}/\partial p > 0$.
    \end{claim}
    \begin{proof}
        Let $\tuple x \in \Bool^n$ be defined as
        \[
        x_i = \begin{cases}
            0, &\quad\text{if $i$ is \kl{increasing} or $q=0$} \\
            1, &\quad\text{if $i$ is \kl{decreasing} and $q > 0$}
        \end{cases}
        \]
        Let $S$ be any maximal set of \kl{increasing} coordinates of $f$ such that $f(\tuple x \oplus S) = 0$.
        Note that there exists some $j \in \domainup{f}\setminus S$ because $f(\tuple x \oplus \domainup{f}) \ge f(1, \dots, 1) = 1$.
        Now $\tuple x \oplus S$ witnesses $\Inf[(p,q)]{f,j} > 0$. Consequently
        \[
        \frac{\partial}{\partial p} \E[p,q]{f} = \sum_{\text{incr}~i} \Inf[(p,q)]{f,i} > 0 \qedhere
        \]
    \end{proof}
    We apply the Implicit Function Theorem to the two-variable function $F = [0,1]^2 \ni (p, q) \mapsto \E[p,q]{f} \in [0,1]$.
    Since $F$ is smooth by \cref{lem:partderiv}, we conclude that $\curve{\varepsilon}{f}$ is a graph of a smooth increasing function of $q \in [0,1]$.

\end{proof}

It immediately follows that the length of the curve $\curve{\varepsilon}{f}$ is bounded by 2:
\begin{observation}
    Let $f$ be an \kl{idempotent} \kl{unate} function. Then $L(\curve{\varepsilon}{f}) \le 2$ for any $\varepsilon \in (0,1)$.
\end{observation}
\begin{proof}
    \sloppy Fix a parameterization of $\curve{\varepsilon}{f}\colon [0,1] \ni t \mapsto (p(t), q(t)) \in [0,1]^2$. We have
    \begin{align*}
    L(\curve{\varepsilon}{f}) &= \int_0^1 \sqrt{\left(\frac{d}{dt}p(t)\right)^2 + \left(\frac{d}{dt}q(t)\right)^2}dt \\
    &\le \int_0^1 \left( \left|\frac{d}{dt}p(t)\right| + \left|\frac{d}{dt}q(t)\right| \right)dt &&\text{(monotonicity)}\\
    &= |p(1) - p(0)| + |q(1) - q(0)| \\
    &\le 1 + 1 = 2. &&\qedhere
    \end{align*}
\end{proof}

In the tractability proof, we rely on the following corollary of the fundamental Arzelà-Ascoli theorem when arguing uniform convergence of $1/2$-level sets:
\begin{theorem}[Theorem 2.5.14 in \cite{burago2001metricgeom}]\label{th:arzela-ascoli}
    Let $S$ be an infinite set of curves in $\mathbb{R}^2$. If there exists $N$ that bounds the length of every curve in $S$ from above, then $S$ contains a uniformly convergent sequence.
\end{theorem}
\noindent therefore we shall define uniform convergence now.

\begin{definition}[Uniform convergence]
    A sequence of curves $\{\gamma_n\}_{n \in \mathbb N}$, where each $\gamma_n : [0,1] \to \mathbb R^2$, is said to converge uniformly to a curve $\gamma : [0,1] \to \mathbb R^2$ if
    \[
    \sup_{t \in [0,1]} \norm{\gamma_n(t) - \gamma(t)} \to 0 \quad\text{as } n\to\infty.
    \]
\end{definition}
\noindent Uniform convergence implies (but is not equivalent to) pointwise convergence:
\begin{definition}[Pointwise convergence]
    A sequence of curves $\{\gamma_n\}_{n \in \mathbb N}$, where each $\gamma_n : [0,1] \to \mathbb R^2$, is said to converge pointwise to a curve $\gamma : [0,1] \to \mathbb R^2$ if
    \[
    \forall{t \in [0,1]}~\lim_{n\to\infty} \norm{\gamma_n(t) - \gamma(t)} = 0
    \]
\end{definition}

We are ready to approach the tractability proof for \kl{Unate PCSPs}.
Recall that the proof consists of three big steps. We first prove these steps one by one, and combine them at the end.

\begin{restate}{lemma}{lem:regions-to-curve}
    Let $\minion M$ be an \kl{idempotent} \kl{minion} such that
    \[
    \forall{\,\varepsilon > 0}~\exists{\, f \in \minion M} : \left\lvert\left\{(p, q) : \E[p,q]{f} \in (\varepsilon, 1-\varepsilon) \right\}\right\rvert < \varepsilon.
    \]
    Then there exist a sequence of functions $\{f_n\}_{n \in \mathbb N} \subseteq \minion M$ and a monotonic curve $\mathbf T$ from $(0,0)$ to $(1,1)$ such that
    \[
    \forall{P \in [0,1]^2\setminus \mathbf T}:~ \lim_{n\to\infty} \E[P]{f_n} = \begin{cases}
        1, &\quad\text{if } P > \mathbf T \\
        0, &\quad\text{if } P < \mathbf T
    \end{cases}
    \]
\end{restate}
\begin{proof}

    For each $n \in \mathbb{N}$, let $f_n \in \minion M$ be the function obtained from the assumption with $\varepsilon = 1/n$.
    Let $\mathbf T_n$ be the $1/2$-level set of $\E[p,q]{f_n}$, that is
    \[
    \mathbf T_n = \left\{(p, q) \in [0,1]^2 \,\,\Big\lvert\,\, \E[p,q]{f_n} = \frac{1}{2}\right\}.
    \]

    As observed in \cref{lem:curve}, every $\mathbf T_n$ is a monotonic curve from the bottom edge to the top edge of the unit square.
    Consequently, the set of curves $\{\mathbf T_n\}_n$ satisfies the assumption of \cref{th:arzela-ascoli}.

    Let $(n_k)_k$ be a sequence such that $(\mathbf T_{n_k})_k$ is uniformly convergent. In order to avoid double indexing, we replace the sequences $(\mathbf T_n)_n$, $(f_n)_n$, and $(\varepsilon_n)_n$ by the respective subsequences indexed by $(n_k)_k$. Denote by $\gamma$ the curve to which $(\mathbf T_n)_n$ converges uniformly.
    Note that $\gamma$ is a curve from the bottom edge to the top edge of the unit square, due to pointwise convergence.
    \begin{fact}
        A uniformly convergent sequence of monotonic curves converges to a monotonic curve.
    \end{fact}
    We let $\mathbf T$ be the curve $\gamma$ with the bottom endpoint extended to the point $(0,0)$ and the top endpoint extended to the point $(1,1)$.
    Formally, $\mathbf T$ is a piecewise continuous curve obtained by concatenating two horizontal linear segments with $\gamma$ on both sides.
    Observe that $\mathbf T$ is a monotonic curve.

    \sloppy Let $P \in [0,1]^2 \setminus \mathbf T$ be a point. It remains to prove that $\lim_{n\to\infty} \E[P]{f_n} \in \{0, 1\}$ depending on the location of $P$ w.r.t. the curve $\mathbf T$.

    \begin{claim}
        If $P < \mathbf T$, then $\lim_{n\to\infty} \E[P]{f_n} = 0$.
    \end{claim}
    \begin{proof}
        Let $\mathbf R$ be the open ray from $P$ towards $(1,0)$.
        By assumption, it must intersect $\mathbf T$.
        For any curve $\mathbf T'$, we will measure the distance from $P$ to $\mathbf T'$ along $\mathbf R$ as $d(\mathbf T') := \left\lVert P - (\mathbf T' \cap \mathbf R) \right\rVert$ if $\mathbf T'$ intersects $\mathbf R$; otherwise it is $-\infty$.
        By assumption, we have $d(\mathbf T) > 0$.

        Thanks to the uniform convergence of $\{\mathbf T_n\}_n$,
        the distance to $P$ also converges:
        \[
        \lim_{n \to \infty} d(\mathbf T_n) = d(\mathbf T)
        \]

        Let $\beta$ be the (positive) angle between $\mathbf R$ and the horizontal line $q = 0$.
        Fix any $n$ such that $d(\mathbf T_n) \ge d(\mathbf T)/2$ and $4\varepsilon_n < d(\mathbf T)^2\cos\beta\sin\beta$.
        Let $Q = \mathbf R \cap \mathbf T_n$. By the definition of $\mathbf T_n$, we have $\E[Q]{f_n} = 0.5$.
        Consider the axis-aligned rectangle $\mathbb A$ defined by $P$ and $Q$.
        Observe that for any point $P' \in \mathbb A$ holds
        \[
        \E[P]{f_n} \le \E[P']{f_n} \le \E[Q]{f_n} = 0.5
        \]
        by \cref{cor:Emonotonicity}.
        The area of $\mathbb A$ is $\lambda_2(\mathbb A) = d(\mathbf T_n)^2\cos\beta\sin\beta > \varepsilon_n$, which means that $E_{P}[f_n] \le \varepsilon_n$ by the assumption of \cref{th:tractability}.

        Since $\lim_{n\to\infty} d(\mathbf T_n) = d(\mathbf T)$ and $\lim_{n\to\infty} \varepsilon_n = 0$, it follows that $\lim_{n\to\infty} \E[P]{f_n} = 0$.\qedhere

    \end{proof}
    \noindent
    The reader may convince themselves that a similar reasoning proves the opposite case:
    \begin{claim}
        If $P > \mathbf T$, then $\lim_{n\to\infty} \E[P]{f_n} = 1$.
    \end{claim}
    This concludes the proof of the lemma.
\end{proof}

\begin{center}
    \includegraphics[width=0.5\linewidth]{pictures/unate_tractability.pdf}
\end{center}

Recall how random \kl{minors} produced symmetric functions for \kl{Ordered PCSPs} (see \cref{th:ordered-tractability}).
Here, we aim for \kl{unate} functions that are symmetric under permutations of \kl{increasing} and under permutation of \kl{decreasing} variables:
\begin{lemma}\label{lem:random-minor}
    Let $f : \Bool^n \to \Bool$ be a \kl{unate} function. Let $w, h \in \mathbb{N}$ and $m = w + h$.
    Then $f$ has a \kl{minor} $g$ of arity $m$ such that $\forall{\,\tuple x\in\Bool^w, \tuple y \in \Bool^h}$:
    \[
    g_{w,h}(\tuple x \tuple y) = \begin{cases}
        1, &\quad\text{if } \,\E[i/w,j/h]{f} > 1 - \frac{1}{2^m} \\
        1 \text{ or } \,0, &\quad\text{if } \,\E[i/w,j/h]{f} \in \left[\frac{1}{2^m}, 1-\frac{1}{2^m}\right] \\
        0, &\quad\text{if } \,\E[i/w,j/h]{f} < \frac{1}{2^m}
    \end{cases}
    \]
    where $i$ denotes $\ham{\tuple x}$ and $j$ denotes $\ham{\tuple y}$.
\end{lemma}
\noindent Our proof utilizes the probabilistic method akin to \cite[Lemma 3.4]{brakensiek2021conditional}.
\begin{proof}
    Let $\pi : [n] \to [m]$ be a random \kl{minor map} defined as follows: for any $i \in \domainup{f}$, let $\pi(i)$ be distributed uniformly in $[w]$; similarly for any $j \in \domaindown{f}$, let $\pi(j)$ be distributed uniformly in $[m]\setminus[w]$. Denote by $g$ the random \kl{minor} $\minor{f}{\pi}$. Let $\tuple x \in \Bool^w$ and $\tuple y \in \Bool^h$; denote by $i$ the number of ones in $\tuple x$ and by $j$ the number of ones in $\tuple y$. By construction of $\pi$ we have
    \[
    \E[\pi]{g(\tuple x \tuple y)} = \E[\frac{i}{w},\frac{j}{h}]{f}.
    \]
    By the union bound, the probability that $g$ behaves as expected is greater than zero. Therefore, such a \kl{minor} must exist.
\end{proof}

The following observation will be used later to argue that the curve $\mathbf T$ from \cref{lem:regions-to-curve} is a straight line.
More precisely, it states that, for functions in a \kl{unate} \kl{minion}, if the output does not change after flipping some coordinates from 0 to 1, then it also does not change after flipping only a part of those coordinates.
\begin{lemma}\label{lem:boundaryisline}
    Let $g$ be a function whose every \kl{minor} is \kl{unate}.
    Let $\tuple x \le \tuple x''$ and $\tuple y \le \tuple y''$ such that $\ham{\tuple x''}-\ham{\tuple x}$ and $\ham{\tuple y''}-\ham{\tuple y}$ are even. If $g(\tuple x\tuple y) = g(\tuple x''\tuple y'') = a$, then
    there exist $\tuple x \le \tuple x' \le \tuple x''$ and $\tuple y \le \tuple y' \le \tuple y''$ such that $\ham{\tuple x'} = (\ham{\tuple x} + \ham{\tuple x''})/2$ and $\ham{\tuple y'} = (\ham{\tuple y} + \ham{\tuple y''})/2$, and $g(\tuple x'\tuple y') = a$.
\end{lemma}
\begin{proof}
    Let $D_x = \{i \mid x''_i > x_i\}$ and $D_y = \{i \mid y''_i > y_i\}$. By assumption $|D_x|$ and $|D_y|$ are even.
    Without loss of generality, suppose $g(\tuple x\tuple y) = g(\tuple x'' \tuple y'') = 1$.
    Pick any sets $X \subseteq D_x$ and $Y \subseteq D_y$ of size $|D_x|/2$ and $|D_y|/2$, respectively.

    Let $g'$ be the \kl{minor} of $g$ given by the \kl{minor map} $\pi$ that identifies all coordinates within $X \cup Y$ and within $(D_x \cup D_y) \setminus (X \cup Y)$.
    If $\pi(X \cup Y)$ is an \kl{increasing} coordinate of $g'$, then $\tuple x' = \tuple x \oplus X$ and $\tuple y' = \tuple y \oplus Y$ are the required tuples as $g(\tuple x'\tuple y') \ge g(\tuple x\tuple y) = 1$.
    Otherwise it is \kl{decreasing}, and $\tuple x' = \tuple x'' \oplus X$ and $\tuple y' = \tuple y'' \oplus Y$ are the required tuples as $g(\tuple x'\tuple y') \ge g(\tuple x''\tuple y'') = 1$.
\end{proof}

We are in a position to prove that, for \kl{Unate PCSPs}, sharp threshold in its \kl{polymorphism} \kl{minion} is a sufficient condition for tractability (see \cref{th:tractablelist}).
\maintractability*

\begin{proof}
    First of all, we note that it suffices to prove the theorem for \kl{idempotent} \kl{minions}: suppose $\minion M$ consists of \kl{unate} functions. If it contains a constant function, then we are done. Otherwise it splits into two \kl{minions}: one consists of the functions with unary \kl{minor} $x \mapsto x$, and the other consists of the functions with unary \kl{minor} $x \mapsto 1-x$.
    One of them satisfies the assumption of the theorem.
    If it is the former, then it suffices to note that it is \kl{idempotent}; if it is the latter, then its \kl{negated} counterpart is \kl{idempotent}. Hence, as long as the theorem holds for \kl{idempotent} minions, the conclusion follows.

    From now on, we assume that $\minion M$ is \kl{idempotent}.
    \cref{lem:regions-to-curve} provides a sequence $\{f_n\}_n$ that, in the limit, exhibits a sharp threshold over a monotonic curve $\mathbf T$.

    Next, we argue that the sharp threshold provides a way to obtain \kl{minors} that are, on most inputs, symmetric under permutations of \kl{increasing} and under permutations of \kl{decreasing} coordinates:
    \begin{claim}\label{cl:sign-symmetric-minors}
        For every $w, h \in \mathbb N_+$ there exists a function $g_{w,h} \in \minion M$ of arity $w+h$ that satisfies the following: For any $\tuple x\in\Bool^w, \tuple y \in \Bool^h$ such that $P := (\frac{\ham{\tuple x}}{w}, \frac{\ham{\tuple y}}{h}) \notin \mathbf T$ holds
        \[
        g_{w,h}(\tuple x \tuple y) = \begin{cases}
            1, &\quad\text{if } P > \mathbf T \\
            0, &\quad\text{if } P < \mathbf T
        \end{cases}
        \]
    \end{claim}
    \begin{proof}
        \sloppy Fix $w, h \in \mathbb N_+$. The function $g_{w,h}$ is obtained from \cref{lem:random-minor} applied to $f_N$ for sufficiently large $N$. Namely, $N$ has to be chosen in such a way that for all points $P$ of the form $(\frac{i}{w}, \frac{j}{h}) \notin \mathbf T$, the expectation $\E[P]{f_N}$ is at a distance less than $1/2^{w+h}$ from the designated integer value (0 or 1). It remains to note that such $N$ always exists as the number of points of this form is finite.
    \end{proof}

    \newcommand{\openT}{\mathbf T_{\text{open}}}
    Let $\openT = \mathbf T \cap (0,1)^2$. We claim that $\openT$ must be a (possibly empty\footnote{The segment $\openT$ is empty e.g. for $\MAX$ and $\MIN$ \kl{minions}.}) open line segment. To that end, we will prove the following property:
    \begin{claim}\label{cl:midpoint}
        For any points $A, B \in \openT$,
        the midpoint of $A$ and $B$ lies on $\openT$.
    \end{claim}
    \noindent Since $\mathbf T$ is continuous, it would follow that the subcurve of $\openT$ bounded by any points $A$ and $B$ must be a line segment. Therefore the whole $\openT$ must be a line segment.
    \begin{proof}[Proof of \cref{cl:midpoint}]
        We argue that there exist two sequences of points $\{A^{(w)}\}_{w \in \mathbb N}$ and $\{B^{(w)}\}_{w \in \mathbb N}$ that converge to $A$ and $B$, respectively, and additionally, for every $w$,
        \begin{itemize}[leftmargin=*]
            \item the points $A^{(w)}$, $B^{(w)}$ as well as their midpoint lie on the grid $\left\{\left(\frac{i}{w}, \frac{j}{w}\right) \mid 0 \le i, j \le w\right\}$, and
            \item $A^{(w)}, B^{(w)} < \mathbf T$.
        \end{itemize}
        For example, let $A = (A_x, A_y)$ and $B = (B_x, B_y)$, and put
        \[
        A^{(w)} =
        \left( \frac{2}{w}\ceil*{\frac{A_x \cdot w}{2}-1}\,,\, \frac{2}{w}\floor*{\frac{A_y \cdot w}{2}+1} \right)
        \]
        and
        \[
        B^{(w)} =
        \left( \frac{2}{w}\ceil*{\frac{B_x \cdot w}{2}-1}\,,\, \frac{2}{w}\floor*{\frac{B_y \cdot w}{2}+1} \right)
        \]
        Because $0 < A_x, A_y, B_x, B_y < 1$, the points $A^{(w)}, B^{(w)} \in [0,1]^2$ satisfy all aforementioned properties.

        Let $M^{(w)}$ denote the midpoint of $A^{(w)}$ and $B^{(w)}$. We will argue that $M^{(w)} < \mathbf T$ or $M^{(w)} \in \mathbf T$.
        First, apply \cref{lem:boundaryisline} to $g_{w,w}$ --- take any tuples $\tuple x, \tuple y, \tuple x'', \tuple y'' \in \Bool^{w}$ such that $\ham{\tuple x} = wA^{(w)}_x, \ham{\tuple y} = wA^{(w)}_y, \ham{\tuple x''} = wB^{(w)}_x, \ham{\tuple y''} = wB^{(w)}_y$.
        Observe that $g_{w,w}(\tuple x\tuple y) = g_{w,w}(\tuple x''\tuple y'') = 0$ by \cref{cl:sign-symmetric-minors}, since $A^{(w)} < \mathbf T$ and $B^{(w)} < \mathbf T$.
        Thus we obtain $\tuple x'$ and $\tuple y'$ from \cref{lem:boundaryisline} such that $\ham{\tuple x'} = wM^{(w)}$, $\ham{\tuple y'} = wM^{(w)}$, and $g_{w,w}(\tuple x'\tuple y') = 0$.
        Thus by \cref{cl:sign-symmetric-minors} it cannot happen that $M^{(w)} > \mathbf T$.

        Obviously, $M^{(w)}$ converges to the midpoint $M$ of $A$ and $B$. Since $\mathbf T$ is continuous, this means that $M < \mathbf T$ or $M \in \mathbf T$.

        The reader may convince themselves that a similar reasoning performed from the opposite side implies that $M > \mathbf T$ or $M \in \mathbf T$. Consequently, $M$ must lie on $\mathbf T$.
        It remains to note that since $A$ and $B$ lie inside $(0,1)^2$, so does their midpoint.
    \end{proof}

    Now that we know $\openT$ is a line segment, there are 3 possible configurations of $\openT$ with respect to the diagonal line $\mathbf D = \{(p,q) \in (0,1)^2 \mid p = q\}$: they either are disjoint, or intersect at exactly one point, or coincide.
    We consider these cases separately.

    \begin{claim}
        If $\openT \cap \mathbf D = \emptyset$, then $\MAX \subseteq \minion M$ or $\MIN \subseteq \minion M$.
    \end{claim}
    \begin{proof}
        Since two line segments do not intersect, either $(p,p) > \openT$ for every $0 < p < 1$, or $(p,p) < \openT$ for every $0 < p < 1$.
        Without loss of generality, assume that the former holds.

        Fix any $m$. We claim that $\gen{max}{m}$ is a \kl{minor} of $g_{m,m}$ as follows:
        \[
        \gen{max}{m}(x_1, \dots, x_m) = g_{m,m}(x_1, \dots, x_m, x_1, \dots, x_m)
        \]
        Clearly $g_{m,m}(0, \dots, 0) = 0$ and $g_{m,m}(1, \dots, 1) = 1$ by \kl{idempotency}.
        For the remaining inputs, we apply \cref{cl:sign-symmetric-minors}: for all $0 < i < m$ the point $(\frac{i}{m}, \frac{i}{m}) > \openT$.
    \end{proof}

    \begin{claim}
        If a point $(t,t) = \openT \cap \mathbf D$ for some $0 < t < 1$, then $\THR{t} \subseteq \minion M$.
    \end{claim}
    \begin{proof}
        For every $m \in \mathbb N_+$, let
        \[
        e_m(x_1, \dots, x_m) = g_{m,m}(x_1, \dots, x_m, x_1, \dots, x_m)
        \]
        Clearly $e_m(0, \dots, 0) = 0$ and $e_m(1, \dots, 1) = 1$ by \kl{idempotency}.

        First, note that whenever $tm \notin \mathbb N$, the function $e_m$ is symmetric: for every $\tuple x$ the value $e_m(\tuple x)$ depends only on whether $(\frac{\ham{\tuple x}}{w}, \frac{\ham{\tuple y}}{h}) < \mathbf T$ or $> \mathbf T$. Since $e_m$ is symmetric and \kl{idempotent}, it must be \kl(func){increasing}.
        Pick any $m$ such that $\frac{1}{m} < t < 1-\frac{1}{m}$. Since the points $(\frac{1}{m}, \frac{1}{m})$ and $(1-\frac{1}{m}, 1-\frac{1}{m})$ lie on the opposite sides of $\mathbf T$, we have $e_m(1, 0, \dots, 0) \neq e_m(1, \dots, 1, 0)$ by \cref{cl:sign-symmetric-minors}. Moreover, since $e_m$ is \kl(func){increasing}, we have $e_m(1, 0, \dots, 0) < e_m(1, \dots, 1, 0)$. Therefore $(\frac{1}{m}, \frac{1}{m}) < \mathbf T$, which means that $\mathbf T$ is ''steeper'' than $\mathbf D$.

        Fix any $m$ such that $tm \notin \mathbb N$. We claim that $\genthr{t}{m} = e_m$.
        For all $0 < i < m$ such that $i < tm$ the point $(\frac{i}{m}, \frac{i}{m}) < \mathbf T$, and otherwise it is $> \mathbf T$.
        It suffices to apply \cref{cl:sign-symmetric-minors}.
    \end{proof}

    \begin{claim}
        If $\openT = \mathbf D$, then $\AT \subseteq \minion M$.
    \end{claim}
    \begin{proof}
        It must be that $\mathbf T = \{(p,p) \mid 0 \le p \le 1\}$.
        Fix any $m$.
        We claim that $\gen{at}{m} = g_{m+1,m}$.
        Clearly $g_{m+1,m}(0, \dots, 0) = 0$ and $g_{m+1,m}(1, \dots, 1) = 1$ by \kl{idempotency}.
        For the remaining inputs, we apply \cref{cl:sign-symmetric-minors}: for all $i \le j$ the point $(\frac{i}{m+1}, \frac{j}{m}) < \mathbf T$, and otherwise it is $> \mathbf T$.
    \end{proof}
    This concludes the proof of the theorem.
\end{proof}

\subsection{Proofs omitted in Hardness}

In this section, we prove the two propositions mentioned in \cref{sec:unate-hardness}: \kl{unate} functions can have large \kl{total influence} under $\shap$, and they are not \kl{low-degree concentrated} in general.

\begin{restatable}{proposition}{atlargeshapley}\label{proposition:at_has_large_shapey_value}
    $\I[\shap]{\gen{at}{n}} = \Theta(\sqrt{n})$.
\end{restatable}

\begin{proof}
    \newcommand{\X}[0]{\mathbf{X}}
    The upper bound is a simple consequence of \cref{unate:lemma:total_influence_is_sqrt}, \cref{shapley:claim:connection_between_shapley_and_biased} and the fact that $\int_0^1 \sqrt{1/x(1-x)} \, dx = \pi$:
    \begin{align*}
        \I[\shap]{\gen{at}{n}} &= \int_0^1 \I[(p)]{\gen{at}{n}} \, dp \\ &\leq  \sqrt{2n+1} \cdot \int_0^1 \frac{1}{\sqrt{p(1-p)}} \, dp = \mathcal{O}(\sqrt{n}).
    \end{align*}
    We now proceed to the lower bound. Suppose that $\{1, \dots, n+1\} = \domainup{\gen{at}{n}}$. Since $\gen{at}{n}$ is symmetric under permutations of these $n+1$ \kl{increasing} coordinates, they all share the same \kl{influence}, so $\I[\shap]{\gen{at}{n}} \geq (n+1) \cdot \Inf[\shap]{\gen{at}{n}, 1}$; it therefore suffices to show that $\Inf[\shap]{\gen{at}{n}, 1} = \Omega( 1/\sqrt{n})$. Let $X \subseteq \{0,1\}^{2n+1}$ be the set of tuples $\tuple x$ such that (1) $x_1 = 0$, and (2) $x_2 + \dots + x_{n+1} = x_{n+2} + \dots + x_{2n+1}$. Clearly, we have $\Inf[\shap]{\gen{at}{n}, 1} \geq \shap(X)$, so we now want to lower-bound the measure of $X$. By Stirling's approximation $\binom{2n }{n} \sim 4^n/\sqrt{\pi n}$, we have:
    \begin{align*}
        \shap&(X) = \frac{1}{2n+2} \cdot \sum_{k = 0}^{n} \frac{\binom{n }{k}^2}{\binom{2n+1 }{2k}} \\
                &= \frac{1}{(2n+2)(2n+1)\binom{2n }{n}} \cdot \sum_{k=0}^n (2n-2k+1) \cdot \binom{2(n-k) }{(n-k)} \cdot \binom{2k }{k} \\
                &\sim \frac{\sqrt{\pi n}}{4 \cdot 4^n \cdot n^2} \cdot \left( \sum_{k=0}^n \binom{2(n-k) }{(n-k)} \binom{2k }{k} + 2\sum_{k=0}^n (n-k) \binom{2(n-k) }{(n-k)} \binom{2k }{k} \right) \\
                &\geq \Omega\left( \frac{1}{4^n \cdot n^{3/2}} \right) \cdot \Big[x^n \Big] \Big( F^2(x) + 2 \cdot x F'(x) \cdot F(x) \Big), \tag{{\color{ao(english)}$1$}}\label{unate:eq:lower_bound_for_influence_of_AT}
    \end{align*}
    where $F(x)$ is the generating function of the sequence $\{ \binom{2n}{n} : n \in \mathbb{N} \}$. Using the well known facts that (1) $F(x) = 1/\sqrt{1-4x}$, (2) $1/(1-4x) = \sum_{t=0}^\infty 4^t x^t$, and (3) $1/(1-4x)^2 = \sum_{t=0}^\infty 4^t(t+1)x^t$, we obtain that the last term above is
    \[
        \Big[x^n \Big]\bigg( \frac{1}{1-4x} + \frac{4x}{(1-4x)^2} \bigg) = 4^n + 4 \cdot n \cdot 4^{n-1} = 4^n \cdot (n+1).
    \]
    Putting it into \eqref{unate:eq:lower_bound_for_influence_of_AT} gives $\shap(X) = \Omega(1/\sqrt{n})$.
\end{proof}

\begin{restatable}{proposition}{monotonenotlowdegreeconcentrated}\label{proposition:monotone_not_low_degree_concentrated}
    There exists $\varepsilon > 0$ such that for every $d \geq 1$ there is an \kl(func){increasing} \kl{Boolean function} $f$ with
    \[
        \sum_{|S| > d} \hat{f}^{(1/2)}(S)^2 > \varepsilon.
    \]
\end{restatable}

\begin{proof}
    \AP
    Let $s, b \in \mathbb N$ such that $b$ is large and $2^{s} \approx b/(\ln 2)$. We define the \kl{Boolean function} $\intro*\gentribes{s}{b}$ $: \{0,1\}^{sb} \to \{0,1\}$ as follows. We divide the coordinates into $b$ blocks (tribes) of size $s$: $(x_1, \dots, x_s)$, $(x_{s+1}, \dots, x_{2s})$, $\dots (x_{(b-1)s+1}, \dots, x_{bs})$ and let $\gentribes{s}{b}(\tuple x) = 1$ if and only if there is at least one tribe with all variables equal to 1; One can also think of this function as a DNF formula with $b$ clauses of size $s$. First, we observe that our choice of $s$ makes $\gentribes{s}{b}$ approximately unbiased:
    \[
        \E[1/2]{\gentribes{s}{b}} = 1 - (1-2^{-s})^b \approx 1 - \left(1-\frac{\ln 2}{b}\right)^b \to 1 - e^{-\ln 2} = \frac{1}{2}.
    \]
    Fix any $\delta > 0$. We want to show that $\NS[(1/2)]{\gentribes{s}{b}, 1-\delta}$ remains constant as $b$ increases and use its connection with the \kl{Fourier coefficients} to deduce the statement. Let $\tuple x \sim \cube{1/2, sb}$ and $\tuple y \sim \N[(1/2)]{\tuple x, 1-\delta}$. Consider the first tribe $(x_1, \dots, x_s)$. Let $A = x_1 \land \dots \land x_s$ and $B = y_1 \land \dots \land y_s$. Simple calculations yield the following:
    \begin{align*}
        \left(\Pr\left[A = 0\right]\right)^b &= (1-2^{-s})^b \approx \left( 1 - \frac{\ln 2}{b} \right)^b \to \frac{1}{2}; \\
        \left(\Pr \left[A = 0 \land B = 0 \right] \right)^b &= \left(1 - 2^{-s}\big(2 - (1-\delta/2)^s\big)\right)^b \to \frac{1}{4}.
    \end{align*}
    The second limit uses $2^{-s} \approx \ln 2 / b$ together with $(1-\delta/2)^s \to 0$, so that the exponent $b \cdot 2^{-s}\big(2 - (1-\delta/2)^s\big) \to 2\ln 2$. Note that $A$ and $B$ are positively correlated, so $\Pr[A = 0 \land B = 0] \geq 1 - 2^{1-s}$; the naive bound $(1-2^{1-s})^b$ is \emph{not} an upper bound here, but this does not matter since we only need the limit.
    Since the tribes are identical and independent, we obtain the following estimate (for sufficiently large $b$):
    \begin{align*}
        \NS[(1/2)]{\gentribes{s}{b}, 1-\delta} &\geq \Pr[\gentribes{s}{b}(\tuple x) = 0 \land \gentribes{s}{b}(\tuple y) = 1]  \\
                                               &= \left(\Pr\left[A = 0\right]\right)^b - \left(\Pr\left[A = 0 \land B = 0\right]\right)^b \to \frac{1}{2} - \frac{1}{4} = \frac{1}{4}.
    \end{align*}
    In particular, $\NS[(1/2)]{\gentribes{s}{b}, 1-\delta} \geq 1/4 - o(1)$ as $b \to \infty$. Combining this with $\E[1/2]{\gentribes{s}{b}} \to 1/2$, \cref{noise:claim:connection_between_noise_and_coefficients} gives that, as $b \to \infty$,
    \[
        \ip{\gentribes{s}{b}}{\TN[(1/2)]{\gentribes{s}{b}, 1-\delta}} = \E[1/2]{\gentribes{s}{b}} - \tfrac12 \NS[(1/2)]{\gentribes{s}{b}, 1-\delta} \leq 3/8 + o(1).
    \]
    We are now in position to show that $\gentribes{s}{b}$ is not \kl{low-degree concentrated}. Suppose by contrary, that for every $\varepsilon > 0$ there exists a constant $d \geq 1$ such that
    \begin{align*}
        \sum_{0 \leq |S| \leq d} \widehat{\gentribes{s}{b}}(S)^2 \geq \norm{\gentribes{s}{b}}^2 - \varepsilon \approx \frac{1}{2} - \varepsilon.
    \end{align*}
    However, this would imply a contradiction. Fix $\varepsilon < 1/8$, and then $\delta > 0$ small enough that $(1-\delta)^d \cdot (1/2 - \varepsilon) > 3/8$ (possible since $d = d(\varepsilon)$ does not depend on $\delta$). By \cref{proposition:noise:coefficients_of_noise_operator},
    \[
        \frac{3}{8} + o(1) \geq \sum_{0 \leq |S| \leq d} (1-\delta)^{|S|} \cdot \widehat{\gentribes{s}{b}}(S)^2 \geq (1-\delta)^d \cdot \left(  \frac{1}{2} - \varepsilon \right) - o(1) > \frac{3}{8},
    \]
    where the last inequality holds for sufficiently large $b$; this contradiction completes the proof. \qedhere
\end{proof}

\section{From \richconj\ to random 2-to-1 hardness condition}\label{appendix:hardness_reduction}
\newcommand{\PMC}[0]{\mathsf{PMC}}
In this section, we outline the reduction from the \kl{Rich 2-to-1 Conjecture}, recently postulated by \cite{onrich2to1}, to any \PCSP\ satisfying our \kl{random 2-to-1 condition}.
This reduction is encapsulated in \cref{prelims:theorem:random_2_to_1_condition_implies_hardness}; we note that the proof is based on the reduction to \kl{Ordered PCSPs} in \cite[Section 4.2]{brakensiek2021conditional}.
The reduction follows a standard framework established in \cite{austrin20172+sat,brakensiek2017promise,Bible}; in particular, it splits into two parts, reducing to an intermediate problem first.

\subsection{Rich 2-to-1 Conjecture}

The \kl{Rich 2-to-1 Conjecture} followed from the long-standing quest for a perfect-completeness analog of the Unique Games Conjecture, which asserts that it is \NP-hard to determine the approximate value of a certain type of game, known as a Unique Game. In the literature, these games are usually formalized by means of \emph{Label Cover}.

\begin{definition}[\intro{Label Cover}]\AP
    An instance of \emph{Label Cover} $\Psi = (L \cup R, E, \Sigma_L, \Sigma_R, \Pi)$ consists of a bipartite graph $(L \cup R, E)$, alphabets $[\Sigma_L]$, $[\Sigma_R]$ and a set $\Pi$ of functional \intro{constraints} $\pi_e : [\Sigma_L] \to [\Sigma_R]$ for every $e \in E$. Given a \intro{labeling} $\sigma$, which assigns an element of $[\Sigma_L]$ to every vertex in $L$ and an element of $[\Sigma_R]$ to every vertex in $R$, we say that constraint $\pi_e \in \Pi$ for $e = (u, v) \in E$ is satisfied if $\pi_e(\sigma(u)) = \sigma(v)$. The value of $\Psi$ is the maximum fraction of constraints that a labeling can satisfy simultaneously.
\end{definition}

Unique games correspond to \kl{Label Cover} instances where all $\pi_e$ are bijective functions; similarly, 2-to-1 games correspond to Label Cover instances where all $\pi_e$ are \kl{2-to-1} maps.
The Unique Games Conjecture postulates that it is \NP-hard to find the value of $\Psi$, even under the promise that this value is smaller than a fixed $\delta < 1$.
We usually formalize these promise problems under the name of \emph{Gap Label Cover}.

\begin{definition}[\intro{Gap Label Cover}]\AP
    Suppose that $n \geq 1$ and $0 \le \varepsilon \le \delta \le 1$. By $\Gap_n[\delta, \varepsilon]$ we denote the decision problem defined as follows: Given a \kl{Label Cover} instance $\Psi$ with $\Sigma_L, \Sigma_R \leq n$, answer $\mathsf{YES}$ if there is a \kl{labeling} satisfying at least a $\delta$-fraction of \kl{constraints}, and answer $\mathsf{NO}$ if all labelings satisfy at most an $\varepsilon$-fraction of constraints.
\end{definition}

For \PCSP\ applications, one usually requires that $\delta = 1$, which makes it a \emph{perfect-completeness} problem.
Finally, we define Rich Games.

\begin{definition}[\intro{Richness}]\AP
    An instance $\Psi$ of \kl{Label Cover} is \emph{Rich 2-to-1} (and called a \intro{Rich 2-to-1 Label Cover} instance) if
    \begin{itemize}[leftmargin=*]
        \item all \kl{constraints} $\pi_e$ are \kl{2-to-1} maps, and
        \item the process of uniformly choosing a constraint adjacent to any fixed vertex in $L$ yields the uniform distribution over all \kl{2-to-1} maps $[2n] \to [n]$.\footnote{We note that in the original formulation of the conjecture, the richness condition is weaker: the distribution of paritions of $[2n]$ into pairs induced by pre-images of constraints is required to be uniform. Our formulation is more convenient for our considerations --- it comes from \cite{brakensiek2021conditional}.}
    \end{itemize}
    \AP The variant of \kl{Gap Label Cover} restricted to Rich 2-to-1 instances, called \intro{Rich 2-to-1 Gap Label Cover}, is denoted by $\GapRich_n[\delta,\varepsilon]$.
\end{definition}

\begin{conjecture}[\intro{Rich 2-to-1 Conjecture} \cite{onrich2to1}]\AP
    For every $0 < \varepsilon < 1$, the problem $\GapRich_n[1, \varepsilon]$ is \NP-hard for a sufficiently large $n$.
\end{conjecture}

\subsection{\texorpdfstring{Proof of \cref*{prelims:theorem:random_2_to_1_condition_implies_hardness}}{Proof of the random 2-to-1 hardness condition theorem}}
\AP
Before proceeding to the reduction, we introduce the notions necessary to define the intermediate problem.
A \intro{minor condition} is a finite set of \intro{identities} of the form
\[
    f(x_{\pi(1)}, \dots, x_{\pi(n)}) \approx g(x_1, \dots, x_m),
\]
\AP
where $f$ and $g$ are functional symbols and $\pi : [n] \to [m]$. In addition, we require that the sets of symbols on left and right-hand sides are disjoint. Given a \kl{minion} on $(A, B)$, that is, $\minion{M} \subseteq \{f\colon A^n \to B \mid n \ge 1\}$, \AP we say that an identity is \intro{satisfied} in $\minion{M}$ if there exists a \intro{symbol interpretation} $\zeta$ that assigns to every symbol a function in $\minion{M}$, such that
\[
    \forall a_1, \dots, a_m \in A : \zeta(f)(a_{\pi(1)}, \dots, a_{\pi(n)}) = \zeta(g)(a_1, \dots, a_m),
\]
or, equivalently, $\zeta(g) = \zeta(f)^\pi$.
Furthermore, a minor condition is satisfied in $\minion{M}$ if there exists an interpretation that satisfies all identities simultaneously.

\AP
A \kl{minor condition} is \intro{trivial} if it is satisfied in every \kl{minion}. We note that if a minor condition $\Sigma$ is satisfied in the \kl{minion} of \kl{projections} on a set of at least two elements, then it is automatically trivial. To see why, observe that if such a satisfying \kl{interpretation} assigns projections $\zeta(f)(\tuple x) = x_i$ and $\zeta(g)(\tuple y) = y_j$ related by an identity, then it must hold that $a_i = a_{\pi(j)}$ for every choice of $\tuple a$, and since the underlying set has at least two elements, this implies $i = \pi(j)$. Now, to find a satisfactory interpretation of $\Sigma$ in another minion on $(A, B)$, we use a small trick: take a unary function $h : A \to B$ in the minion and generate all \kl{minors} of $h$ of the form $(x_1, \dots, x_n) \mapsto h(x_i)$. The functions obtained are not exactly projections, but they all have at most one \kl{essential} coordinate. In particular, the interpretation $\zeta'$ defined as $\zeta'(f)(\tuple x) = h(x_i)$ if and only if $\zeta(f)(\tuple x) = x_i$, satisfies $\Sigma$.

Intuitively, every \kl{Label Cover} instance can be treated as a \kl{minor condition}.
This will comprise the first link in our reduction to the intermediate problem called \textit{Promise Minor Condition}.

\begin{definition}[\intro{Promise Minor Condition}]\AP
    Given $n \geq 1$ and a \kl{minion} $\minion{M}$, the \textit{Promise Minor Condition} problem $\PMC_n(\minion{M})$ is the problem defined as follows: for an input \kl{minor condition} $\Sigma$ with \kl{identities} of arity at most $n$, answer $\mathsf{YES}$ if $\Sigma$ is \kl{trivial}, and answer $\mathsf{NO}$ if $\Sigma$ is not \kl{satisfied} in $\minion{M}$.
\end{definition}

We emphasize that the parameter $n$ is treated as a constant and the size of an instance is the number of \kl{identities}. As it turns out, if $(\A, \B)$ is a \PCSP\ \kl(PCSP){template} and $\minion{M} = \Pol(\A, \B)$, then $\PMC_n(\minion{M})$ is log-space reducible to $\PCSP(\A, \B)$. This fact follows from a construction that resembles \textit{long code tests}, and is one of the most fundamental results in algebraic approach to \PCSP s. We refer the reader to \cite[Section 3.3]{Bible} for more details.

The remaining part is to show how \kl{Rich 2-to-1 Gap Label Cover} can be reduced to \kl{Promise Minor Condition}. Suppose that $\Psi = (L \cup R, E, 2n, n, \Pi)$ is a \kl{Rich 2-to-1 Label Cover} instance. We construct a \kl{minor condition} from $\Psi$ in the following way. Every vertex $u \in L$ is replaced with a symbol $f_u$ of arity $2n$ and $v \in R$ with a symbol $g_v$ of arity $n$. Moreover, for every edge $e = (u, v)$ and its corresponding constraint $\pi_e \in \Pi$, we write the following identity:
\[
    f_u\big(x_{\pi_e(1)}, \dots, x_{\pi_e(2n)}\big) \approx g_v (x_1, \dots, x_n). \tag{{\color{ao(english)}$1$}} \label{eq:minor_condition_construction_from_label_cover}
\]
\AP
Denote the resulting minor condition by $\Sigma$. Observe that $\Sigma$ is essentially an equivalent presentation of $\Psi$: the \kl{labelings} of $\Psi$ are in 1-to-1 correspondence with the interpretations of $\Sigma$ in \kl{projections}. This automatically yields the \intro{completeness} of our reduction, i.e. preservation of $\mathsf{YES}$-instances --- if all \kl{constraints} of $\Psi$ can be simultaneously satisfied, then $\Sigma$ is \kl{trivial}.

\AP
It is the other side that is challenging: the preservation of $\mathsf{NO}$-instances, which is called \intro{soundness}. For this part to work, additional assumptions about $\minion{M}$ are needed. In our case, this assumption is the \kl{random 2-to-1 condition}. We show how to utilize this condition in the proof of the following proposition, which, together with the fact that $\PMC_n(\Pol(\A, \B))$ reduces to $\PCSP(\A, \B)$, finishes the reduction and implies \cref{prelims:theorem:random_2_to_1_condition_implies_hardness}.

\begin{proposition}
    If a \kl{minion} $\minion{M}$ satisfies the \kl{random 2-to-1 condition}, then there is a constant $\varepsilon = \varepsilon(\minion{M}) > 0$ such that $\GapRich_n[1, \varepsilon]$ poly-time reduces to $\PMC_{2n}(\minion{M})$ for every $n \geq 1$.
\end{proposition}

\begin{proof}
    Fix a \kl{minion} $\minion{M}$ which satisfies the \kl{random 2-to-1 condition} with constants $\mathcal C, \tau > 0$ and a choice function $\Sel$. We put $\varepsilon = \tau/(2M^2)$.
    Given a \kl{Rich 2-to-1 Label Cover} instance $\Psi = (L \cup R, E, 2n, n, \Pi)$, we construct a \kl{minor condition} $\Sigma$ as in \eqref{eq:minor_condition_construction_from_label_cover}. It is obvious that $\Sigma$ can be constructed in polynomial time. Therefore, it remains to show that this is a proper reduction.

    \paragraph{Completeness.} Suppose that $\Psi$ is a $\mathsf{YES}$-instance of $\GapRich_n[1,\varepsilon]$ and all \kl{constraints} in $\Pi$ are satisfied with a \kl{labeling} $\sigma : L \cup R \to [2n]$. As discussed above, $\sigma$ induces an \kl{interpretation} that assigns the \kl{projection} $\tuple x \mapsto x_{\sigma (w)}$ to the symbol corresponding to vertex $w$. This interpretation satisfies all \kl{identities}, and thus $\Sigma$ is \kl{trivial}.

    \paragraph{Soundness.} By contraposition, suppose that $\Sigma$ is satisfied in $\minion{M}$, which is witnessed by an \kl{interpretation} $\zeta$. Our goal is to show there is a \kl{labeling} $\sigma : L \cup R \to [2n]$, which satisfies more than $\varepsilon$-fraction of \kl{constraints} in $\Pi$. We apply the probabilistic method. Consider a random labeling $\sigma$, such that for every $u \in L$, the value $\sigma(u)$ is chosen uniformly at random from $\Sel(\zeta(f_u))$. Similarly, $\sigma(v)$ is chosen uniformly at random from $\Sel(\zeta(g_v))$ for every $v \in R$.

    Observe that for every identity as in \eqref{eq:minor_condition_construction_from_label_cover}, the function $\zeta(g_v)$ is a \kl{minor} of $\zeta(f_u)$ with respect to the \kl{2-to-1} map $\pi_e$. Fix any $u \in L$ and let $\Sigma(u)$ be the set of \kl{identities} in $\Sigma$ involving the symbol $f_u$. \kl{Richness} of $\Psi$ ensures that $\zeta(g_v)$, obtained by picking the right-hand side symbol $g_v$ in a random identity in $\Sigma(u)$, is distributed uniformly over all \kl{2-to-1} \kl{minors} of $\zeta(f_u)$. By the \kl{random 2-to-1 condition} then, at least a $\tau$-fraction of identities in $\Sigma(u)$ satisfies $\pi_e(\Sel(\zeta(f_u))) \cap \Sel(\zeta(g_v)) \neq \emptyset$. For every such identity, since both sets $\Sel(\zeta(f_u))$ and $\Sel(\zeta(g_v))$ have sizes at most $M$, we have $\pi_e(\sigma(u)) = \sigma(v)$ with probability at least $1/M^2$ over the choice of $\sigma$. Hence, the expected fraction of satisfied constraints in $\Pi$ adjacent to $u$ is at least $\tau/M^2$. By averaging over $u \in L$, we obtain that the expected number of constraints satisfied by $\sigma$ is at least $\tau/M^2 > \varepsilon$.
\end{proof}

\end{document}